\documentclass[12pt]{article}
\usepackage{amsfonts}
\usepackage[dvips]{color}
\usepackage{xcolor}

\usepackage{amssymb}
\usepackage{graphicx}
\usepackage{amsmath}
\usepackage{amsthm}
\usepackage{bbm}
\usepackage{mathrsfs}
\usepackage{natbib}
\usepackage{rotating}
\usepackage{afterpage}
\usepackage{cancel}
\usepackage{mathtools}
\usepackage{setspace}
\usepackage{multirow}
\usepackage{float}
\usepackage{booktabs}
\usepackage{tabularx}
\usepackage{threeparttable}
\usepackage{makecell}
\usepackage{enumitem}
\setlist[itemize,2]{left=0pt}

\newcommand{\be}{\begin{equation}}
\newcommand{\ee}{\end{equation}}
\newcommand{\bea}{\begin{eqnarray}}
\newcommand{\eea}{\end{eqnarray}}
\newcommand{\bean}{\begin{eqnarray*}}
\newcommand{\eean}{\end{eqnarray*}}

\usepackage{url}
\usepackage[margin=1in]{geometry}

\newtheorem{theorem}{Theorem}
\newtheorem*{theorem*}{Theorem}

\newtheorem{assumption}{Assumption}

\newtheorem{corollary}{Corollary}

\newtheorem{example}{Example}

\newtheorem{lemma}{Lemma}

\newtheorem{proposition}{Proposition}
\newtheorem{remark}{Remark}

\newcommand{\argmax}{\operatorname*{argmax}}

\newcommand{\diag}{\operatorname*{diag}}

\newcommand{\tr}{\operatorname*{tr}}

\newcommand{\E}{\mathbb E}
\newcommand{\V}{\mathbb V}
\newcommand{\Cov}{\operatorname{Cov}}
\newcommand{\Var}{\operatorname{Var}}
\newcommand{\R}{\mathbb R}

\newcommand{\Reg}{\operatorname{Reg}}

\usepackage{booktabs}
\usepackage{indentfirst}
\usepackage{hyperref}
\usepackage{caption}

\newcommand{\runin}[1]{\medskip\noindent\textbf{#1.} }

\begin{document}

\title{\vspace*{-0.5 in}  \textbf{Using Prior Studies to Design Experiments: An Empirical Bayes Approach}}

\author{Zhiheng You\thanks{
		\setlength{\baselineskip}{4mm}  Correspondence: Z. You: Department of Economics, University of Pennsylvania, Philadelphia, PA 19104-6297. Email: zhyou@sas.upenn.edu.} \\
	{\em \large University of Pennsylvania}}

\date{\large \today}
\maketitle

\begin{abstract}
We develop an empirical Bayes framework for experimental design that leverages information from prior related studies. When a researcher has access to estimates from previous studies on similar parameters, they can use empirical Bayes to estimate an informative prior over the parameter of interest in the new study. We show how this prior can be incorporated into a decision-theoretic experimental design framework to choose optimal design. The approach is illustrated via propensity score designs in stratified randomized experiments. Our theoretical results show that the empirical Bayes design achieves oracle-optimal performance as the number of prior studies grows, and characterize the rate at which regret vanishes. To illustrate the approach, we present two empirical applications---oncology drug trials and the Tennessee Project STAR experiment. Our framework connects the Bayesian meta-analysis literature to experimental design and provides practical guidance for researchers seeking to design more efficient experiments.
\end{abstract}



\newpage

\section{Introduction}
\label{sec:introduction}

Researchers often design new randomized controlled trials (RCTs) in settings where related evidence already exists.
That evidence may come from earlier RCTs, quasi-experiments, or observational studies that estimate conceptually similar parameters in different contexts---such as labor supply elasticities, returns to schooling, or treatment effects of related interventions.
How can results from prior studies be used ex ante to make a new experiment more informative for the decisions it is meant to support?
This paper develops an empirical Bayes (EB) approach to optimal experimental design that turns collections of prior-study estimates into a data-driven prior and then optimizes the new design under that prior.

The core idea is to treat study-level parameters as draws from an unknown cross-study distribution $G$.
A collection of published estimates---often only point estimates and standard errors---contains information about $G$, and hence about the parameter in the next setting.
Empirical Bayes provides a natural way to estimate $G$ and obtain a predictive prior for the new study.
The designer can then choose sampling and assignment rules to maximize expected downstream performance under a specified objective.
For example, a planner running a stratified RCT could use prior subgroup-effect estimates to shift experimental effort toward groups where external evidence indicates higher uncertainty or larger expected gains.

Our starting point is the classical Bayesian experimental design framework of \cite{Lindley1972}, which formulates design as a statistical decision problem: the researcher first chooses a design; then observes experimental data generated under the chosen design; and finally takes an action (such as an estimator, a treatment rule, or a policy recommendation) to maximize a welfare function that depends on both the action and the unknown parameter. A design is valuable when it tends to generate data that lead to better downstream actions given uncertainty about the study parameter.

Implementing this framework hinges on the prior distribution that represents that uncertainty.
In applied work, researchers often want to avoid subjective priors, yet they also do not want to ignore the information embedded in a growing body of related studies.
We therefore learn the prior from prior-study results.
Rather than specifying a hyperprior over $G$ as in fully Bayesian hierarchical meta-analysis, we estimate $G$ from prior evidence and treat the resulting $\hat G_n$ as the prior for design and analysis of the new experiment.
We consider two practical routes: a parametric Gaussian specification and a nonparametric prior estimator given by the nonparametric maximum likelihood estimator (NPMLE) for Gaussian location mixtures \citep{KieferWolfowitz1956}.

To make the approach concrete, we study propensity-score designs for stratified randomized experiments in Section~\ref{sec:eb_pscore_design}.
A key motivation is that prior studies often report treatment effect estimates of common subgroups---such as race, income, or gender---and the new experiment is designed as a stratified RCT that adopts those subgroups as its randomization strata.
The experimenter chooses stratum-specific treatment propensities subject to feasibility constraints (e.g., a fixed overall budget), creating tradeoffs across strata.
We characterize EB-optimal propensity designs under three canonical objectives: (i) minimizing Bayes risk for estimating a target parameter such as the average treatment effect, (ii) maximizing in-experiment welfare when treatment is beneficial \citep{CariaEtAl2024}, and (iii) maximizing the expected value of a post-experiment policy adoption decision \citep{KasySautmann2021}.
Across objectives, the EB prior affects design through different channels---prior variances under estimation objectives, prior means under welfare objectives, and both means and variances under policy objectives---yielding distinct assignment rules.

We then provide theoretical guarantees for the plug-in EB design in Section~\ref{sec:theory}.
We first establish a finite-sample oracle inequality: the design regret relative to the oracle that knows $G$ is bounded by twice the uniform error in the value functional induced by replacing $G$ with $\hat G_n$.
Regret consistency then follows from the weaker, verifiable primitive that $\hat G_n$ converges weakly to $G$: under mild continuity and envelop conditions, weak convergence of the prior estimator implies that the uniform value error $\Delta_n\to 0$ in probability, which in turn drives regret to zero.
For a broad class of decision problems in which welfare is affine in the payoff-relevant state, we derive first-order regret rates that mirror the accuracy of the prior estimator: $\Reg_n=O_{p}(n^{-1/2})$ under Gaussian EB and $\Reg_n=O_{p}((\log n)^{(d+\max(d/2,4))/2}/\sqrt{n})$ under NPMLE.
Under additional smoothness and curvature conditions on the design objective, regret becomes a second-order effect and the corresponding rates improve to $O_{p}(n^{-1})$ and $O_{p}((\log n)^{d+\max(d/2,4)}/n)$, respectively. We also characterize design comparisons in the Gaussian experiment, where one observes a normal observation with design-dependent noise covariance. In particular, we show that the prior is irrelevant for design choice only in the genuinely univariate parameter case; with multivariate parameters, the optimal design generally depends on the prior, and this is precisely where empirical Bayes can improve design.

In Section~\ref{sec:emp.app}, we present two empirical applications that illustrate the practical implementation of the propensity-score design developed in Section~\ref{sec:eb_pscore_design}. The first application uses oncology drug trial results from an online clinical-trials database to construct an EB prior for the design of a new immunotherapy trial. Here, the main practical challenge is assembling a usable prior-study dataset from heterogeneous clinical studies with varying estimands and subgroup definitions. The second application treats each site in the Tennessee Project STAR class-size experiment as a prior study, and uses the cross-site distribution of stratum-level treatment effects to design a stratified RCT for a new site. This application highlights how different objectives lead to qualitatively different EB-optimal designs, as each objective channels prior information through a distinct feature of the prior distribution.

This paper relates to several strands of literature. First, the paper relates to the Bayesian optimal experimental design literature, which formalizes design as expected utility maximization under a prior \citep{Lindley1972, ChalonerVerdinelli1995,RainforthEtAl2024}. Most of this work takes the prior as given and focuses on approximating and optimizing expected utilities in complex models. Our focus is complementary: we study how the prior used for design can be learned from a growing archive of prior-study estimates via empirical Bayes, and we provide regret guarantees that quantify the impact of prior estimation on design performance.

Second, it connects to Bayesian evidence synthesis and the construction of informative priors from historical studies.
Bayesian meta-analysis uses cross-study heterogeneity to form predictive distributions for effects in new settings \citep{SuttonAbrams2001,Meager2019,Meager2022,CrostaEtAl2024,SchorfheideYou2025}, and the clinical-trials literature develops explicit ``prior-from-history'' tools such as meta-analytic-predictive (MAP) priors and robust MAP mixtures \citep{NeuenschwanderEtAl2010,SchmidliEtAl2014}, as well as commensurate and power-prior formulations \citep{HobbsEtAl2011,IbrahimChen2000,IbrahimEtAl2015}.
\citet{LinEtAl2022} study how hierarchical meta-analysis can inform phase~I trial design.
Relative to this literature, our contribution is to (i) embed the predictive distribution from prior studies directly inside an explicit optimal design problem; (ii) use empirical Bayes to estimate the cross-study distribution, avoiding hyperprior specification while still borrowing strength; and (iii) analyze a general objective class that includes not only estimation accuracy and testing but also welfare- and policy-oriented criteria.

Third, the paper is related to work on incorporating prior information into experimental practice.
\citet{IacovoneEtAl2025} illustrate how informative priors---elicited from experts and stakeholders---can sharpen Bayesian impact evaluation in small-sample field experiments.
\citet{FinanPouzo2026} develop a Bayesian model-averaging approach that combines multiple prior studies with uncertain external validity to improve learning about treatment effects, and \citet{FinanPouzo2024} study how multiple prior sources can be incorporated into the design of adaptive experiments.
We differ by focusing on the design of a new one-shot experiment when the researcher has access to a growing collection of prior-study estimates, and by providing oracle-style regret guarantees for the resulting empirical Bayes design.

Fourth, the paper relates to the empirical Bayes literature.
Foundational work includes \citet{Robbins1964} and \citet{KieferWolfowitz1956}, with modern results on rates and efficiency \citep{JiangZhang2009,SahaGuntuboyina2020}.
Empirical Bayes and related shrinkage ideas are widely used in economics to improve precision in high-dimensional estimation problems (for example, in teacher value-added and neighborhood effect measurement; see, e.g., \citealp{ChettyFriedmanRockoff2014,ChettyHendren2018}, and in experimental settings with many treatment effects to estimate; see \citealp{AzevedoEtAl2019, AdusumilliEtAl2025}). Unlike the classical compound-decision view, we use empirical Bayes to learn the prior entering the Bayesian experimental design objective. We then study how plugging in this estimated prior affects the performance of the resulting design choice, via regret relative to the oracle design that knows the true prior.

The remainder of the paper is organized as follows.
Section~\ref{sec:setup} presents the general decision-theoretic framework.
Section~\ref{sec:prior_estimation} discusses how to construct the EB prior $\hat G_n$ from prior-study summaries under parametric Gaussian and nonparametric NPMLE approaches.
Section~\ref{sec:eb_pscore_design} studies EB propensity-score designs in stratified randomized experiments under several canonical objectives.
Section~\ref{sec:theory} develops regret bounds and rates for the EB approach, and characterizes the comparison of Gaussian experiments.
Section~\ref{sec:emp.app} presents two empirical applications on drug-trials and Project STAR.

\section{General Setup}
\label{sec:setup}

This section sets up (i) how prior studies provide information about study-level parameters, (ii) a decision-theoretic criterion for evaluating designs, and (iii) an empirical Bayes procedure that estimates the cross-study distribution and then chooses a design optimal under the estimated prior.

\runin{Prior studies and a cross-study distribution}
A researcher observes results from $n$ prior independent studies, indexed by $i=1,\ldots,n$, and wishes to design a new experiment (study $n{+}1$) that targets a parameter $\theta_{n+1}\in\R^d$. Prior study $i$ reports an estimate $\hat\theta_i$ and a consistent covariance estimate $\Sigma_i\in\R^{d\times d}$, which we treat as known and condition on. We use the standard asymptotic normal approximation
\[
\hat{\theta}_i \mid \theta_i, \Sigma_i \ \sim \ \mathcal N(\theta_i,\Sigma_i).
\]
Let $\mathcal D_n=\{(\hat\theta_i,\Sigma_i)\}_{i=1}^n$ denote the historical evidence.

To connect studies, we impose exchangeability: $\theta_i \stackrel{iid}{\sim} G$ for $i=1,\ldots,n{+}1$. This assumption captures the idea that studies estimate conceptually similar but context-specific parameters, with heterogeneity summarized by $G$. A natural extension allows conditional exchangeability given study features $x_i$, i.e.\ $\theta_i\mid x_i\sim G_{x_i}$; for clarity we work with the unconditional case.

A leading parametric specification is the parametric normal model $\theta_i \mid \tau,\nu \stackrel{iid}{\sim}\mathcal N(\tau,\nu)$. Here, $\nu$ captures the degree of external validity \citep{SchorfheideYou2025}: when $\nu$ is small, historical studies are tightly clustered and are informative about the next study; when $\nu$ is large, learning from history is weak and the optimal design reverts toward a prior-free benchmark.

\runin{Design as a decision problem}
Following \cite{Lindley1972}, design is chosen before data are observed to maximize ex-ante expected welfare. The timing of the researcher is: (1) observe previous evidence $\mathcal D_n$ and choose design $\eta\in\mathcal H$; (2) observe experimental data $Y\in\mathcal Y$ generated under likelihood $p(y\mid\theta_{n+1};\eta)$; (3) choose an action $a\in\mathcal A$ based on $(\mathcal D_n,Y)$ to maximize welfare $W(a,\theta_{n+1})$. The design affects what the researcher learns about $\theta_{n+1}$, not $\theta_{n+1}$ itself. We assume that  the new data $Y$ are independent of past evidence $\mathcal D_n$ conditional on $\theta_{n+1}$. 

Given a prior $\Pi(\cdot\mid\mathcal D_n)$ over $\theta_{n+1}$, the ex-ante value of design $\eta$ is
\begin{equation}
\begin{aligned}
U(\eta)
&=
\mathbb E\!\left[
\sup_{a\in\mathcal A}\,
\mathbb E\!\left[ W(a,\theta_{n+1}) \,\middle|\, \mathcal D_n,\,Y \right]
\,\middle|\, \mathcal D_n;\eta
\right],\\
&=
\int_{\mathcal Y}
\left[
\sup_{a\in\mathcal A}
\int_{\Theta} W(a,\theta)\,p(y\mid \theta;\eta)\,\Pi(d\theta\mid \mathcal D_n)
\right]\,\lambda(dy),
\end{aligned}
\label{eq:dm.opt}
\end{equation}
where $\lambda$ is a dominating measure on $\mathcal Y$, and the second equality follows from the conditional independence assumption. Intuitively, a design is valuable if it tends to produce data that move posterior beliefs in directions that matter for downstream decisions. The exact form of $\eta$ is application-specific (treatment assignment rules, sampling algorithm, etc.); its only role here is to determine $p(\cdot\mid\theta;\eta)$.

\runin{Empirical Bayes prior construction and EB-optimal design}
A fully Bayesian approach would place a hyperprior on $G$ and update twice--first using evidence from previous studies and then using data from the new experiment. The EB approach instead estimates $G$ from $\mathcal D_n$ and plugs the estimated distribution into the design problem. Let $\hat G_n$ be an estimator of $G$ based on $\mathcal D_n$. The EB posterior after observing $Y=y$ is
\[
\hat\Pi_{\eta}(d\theta\mid \mathcal D_n,y)
:=
\frac{p(y\mid \theta;\eta)\,\hat G_n(d\theta)}
{\int_{\Theta} p(y\mid \vartheta;\eta)\,\hat G_n(d\vartheta)}.
\]
Define the plug-in EB design objective by replacing $\Pi(\cdot\mid\mathcal D_n)$ with $\hat G_n$ in \eqref{eq:dm.opt}:
\[
U_{\hat G_n}(\eta)
=
\int_{\mathcal Y}
\left[
\sup_{a\in\mathcal A}
\int_{\Theta} W(a,\theta)\,p(y\mid \theta;\eta)\,\hat G_n(d\theta)
\right]\lambda(dy).
\]
The EB-optimal design is any maximizer $\eta^{EB}\in\argmax_{\eta\in\mathcal H}U_{\hat G_n}(\eta)$.

\section{Estimating the Prior from Previous Studies}
\label{sec:prior_estimation}

The EB design approach treats the cross-study distribution $G$ as an object that can be learned from a collection of prior studies.
This section briefly discusses how $G$ can be estimated in practice from the study-level summaries that are typically available in applied work.
The goal is not to introduce new estimation methods, but rather to clarify the two common estimation approaches and their trade-offs.

\subsection{Gaussian Prior}

A widely-used working model in meta-analysis is the parametric Gaussian specification
\begin{equation}
\theta_i \ \stackrel{iid}{\sim}\ \mathcal N(\tau,V),
\qquad i=1,\dots,n,
\label{eq:gauss_re}
\end{equation}
combined with $\hat\theta_i\mid \theta_i \sim \mathcal N(\theta_i,\Sigma_i)$.
Integrating out $\theta_i$ yields the marginal model
\begin{equation}
\hat\theta_i \ \stackrel{approx}{\sim}\ \mathcal N\!\left(\tau,\ \Sigma_i+V\right).
\label{eq:gauss_marginal}
\end{equation}
The Gaussian EB estimator $(\hat\tau,\hat V)$ is typically defined as a maximizer of the profile marginal likelihood induced by~\eqref{eq:gauss_marginal}, possibly under structure on $V$ (diagonal, block-diagonal, or low-rank) to improve interpretability and stability.

A convenient feature of~\eqref{eq:gauss_marginal} is that, for fixed $V$, the mean $\tau$ has a closed-form generalized least squares (GLS) estimator:
\begin{equation}
\hat\tau(V)
=
\Bigg(\sum_{i=1}^n (\Sigma_i+V)^{-1}\Bigg)^{-1}
\Bigg(\sum_{i=1}^n (\Sigma_i+V)^{-1}\hat\theta_i\Bigg).
\label{eq:tau_gls_general}
\end{equation}
Plugging $\hat\tau(V)$ into the log marginal likelihood yields a profile objective in $V$:
\begin{equation}
\ell_n(V)
=
-\frac{1}{2}\sum_{i=1}^n
\left\{
\log\big|\Sigma_i+V\big|
+
\big(\hat\theta_i-\hat\tau(V)\big)'\big(\Sigma_i+V\big)^{-1}\big(\hat\theta_i-\hat\tau(V)\big)
\right\},
\label{eq:profile_ll_V}
\end{equation}
and $\hat V$ can be obtained by numerical maximization of~\eqref{eq:profile_ll_V}.
Given $\hat V$, the EB prior for the new study is $\hat G_n=\mathcal N(\hat\tau(\hat V),\hat V)$.

\subsection{Nonparametric Prior via NPMLE}

When the Gaussian working model is too restrictive, a natural alternative is to estimate $G$ nonparametrically.
Under $\hat\theta_i\mid\theta_i\sim \mathcal N(\theta_i,\Sigma_i)$, the marginal density of $\hat\theta_i$ is the Gaussian location mixture
$f_{G,i}(x)=\int \varphi_{\Sigma_i}(x-\vartheta)\,dG(\vartheta)$, where $\varphi_{\Sigma}$ denotes a mean-zero multivariate normal density with covariance $\Sigma$.
The Kiefer--Wolfowitz nonparametric MLE (NPMLE) estimates $G$ by maximizing the mixture likelihood over all probability measures on $\R^d$ \citep{KieferWolfowitz1956}:
\begin{equation}
\hat G_n \in \argmax_{G}\ \sum_{i=1}^n \log\!\left(\int \varphi_{\Sigma_i}\big(\hat\theta_i-\vartheta\big)\,dG(\vartheta)\right).
\label{eq:npmle_general}
\end{equation}
The optimization problem~\eqref{eq:npmle_general} is infinite-dimensional but convex in $G$.

\runin{Computation in multidimensional $\theta$}
In moderate dimension, a standard approach is to approximate~\eqref{eq:npmle_general} by restricting $G$ to a finite grid $\{\vartheta_k\}_{k\le K}\subset\R^d$ and optimizing over weights $(w_k)_{k\le K}$ on that grid:
\[
\max_{w\in\Delta_K}\ \sum_{i=1}^n \log\!\Big(\sum_{k=1}^K w_k\,\varphi_{\Sigma_i}(\hat\theta_i-\vartheta_k)\Big),
\qquad \Delta_K:=\{w\ge 0:\sum_k w_k=1\}.
\]
This yields a finite-dimensional convex program (often solved by interior-point methods or EM-type fixed point iterations) and is the basis of widely-used implementations in the univariate case \citep{KoenkerMizera2014, KoenkerGu2017}.
For multivariate and heteroscedastic Gaussian location mixtures, \citet{FengDicker2018} propose an approximate NPMLE based on convex optimization, and \citet{SoloffGuntuboyinaSen2025} develop tractable approximations with oracle properties for the resulting empirical Bayes posterior means.
Complementary theoretical results for the multivariate Gaussian location-mixture NPMLE are developed by \citet{SahaGuntuboyina2020}.
Scalable computation for large-scale mixture NPMLEs is discussed by \citet{ZhangCuiSenToh2024}.

\runin{Tradeoffs}
Relative to the parametric Gaussian specification, the NPMLE allows heavy tails, multimodality, and other shapes in $G$ that are empirically common in cross-study heterogeneity.
The tradeoff is that nonparametric estimation is computationally harder, and the resulting regret rates in Section~\ref{sec:rate_n} are typically slower than parametric rates.

\section{Propensity Score Designs in Stratified Randomized Experiments}
\label{sec:eb_pscore_design}

A recurring feature of applied empirical work is that prior studies report heterogeneous treatment effects broken down by a common set of subgroups---defined by characteristics such as race, income, gender, or age.
When a researcher plans a new experiment in a related setting, a natural design choice is to use exactly those reported subgroups as randomization strata, running a stratified RCT.
This section specializes the EB framework to that setting.
The central design problem is that resource constraints force tradeoffs across strata: increasing precision in one stratum mechanically reduces precision elsewhere. We show that an EB prior constructed from previous studies can materially change optimal treatment assignment.

\subsection{Design Class and Design-Induced Likelihood}
\label{subsec:pscore_design_class}

Let $(Y(0),Y(1),X)$ denote potential outcomes and pre-treatment covariates in the target population.
The experimenter observes $X$ prior to assignment and chooses a covariate-dependent propensity score
\[
e(x)\equiv \Pr(D=1\mid X=x), \qquad e:\mathcal X \to [\underline e,\,1-\underline e],
\]
where $\underline e\in(0,1/2)$ is an overlap bound.
Conditional on $X$, treatment is assigned independently as $D\mid X\sim \mathrm{Bernoulli}(e(X))$, and the realized outcome is
$Y = DY(1)+(1-D)Y(0)$.

\runin{Design class} Stratified randomization is implemented through a fixed partition of the covariate space that is shared with the prior studies.
Let $S=s(X)\in\{1,\dots,G\}$ denote the pre-specified stratification and let $\pi_g := \Pr(S=g)$.
A propensity-score design is a vector $e=(e_1,\dots,e_G)'\in[\underline e,1-\underline e]^G$ such that
\[
\Pr(D=1\mid S=g)=e_g,\qquad g=1,\dots,G.
\]
We impose a standard resource constraint
\begin{equation}
\sum_{g=1}^G \pi_g c_g e_g \le B,\qquad c_g\ge 0,
\label{eq:pscore_cost_constraint}
\end{equation}
and we write $\mathcal H\subset[\underline e,1-\underline e]^G$ for the resulting feasible set.

\runin{Parameter of interest} Define the design-relevant parameter as the stratum-specific average treatment effects (ATE)
\[
\tau_g := \E\!\left[Y(1)-Y(0)\mid S=g\right],\qquad g=1,\dots,G,
\]
and collect them in $\theta:=(\tau_1,\dots,\tau_G)'\in\R^G$.
From $\theta$ one can recover common estimands in applied work: the population ATE is $\tau^{ATE}=\sum_g \pi_g \tau_g=\pi'\theta$, while subgroup effects correspond to selecting coordinates of $\theta$.

\runin{Gaussian likelihood induced by the design}
Let $\hat\theta=(\hat\tau_1,\dots,\hat\tau_G)'$ denote the within-stratum difference-in-means estimators computed from the new experiment.
We adopt the following normal approximation, which is standard in large samples and convenient for design analysis.

\begin{assumption}[Gaussian sampling model for design]
\label{ass:gaussian_sampling}
Conditional on $\theta$ and a feasible design $e\in\mathcal H$,
\begin{equation}
\hat\theta \mid \theta;\,e \ \stackrel{approx}{\sim}\ \mathcal N\!\left(\theta,\ \Sigma(e)\right),
\qquad
\Sigma(e)=\mathrm{diag}\!\left(s_1^2(e_1),\dots,s_G^2(e_G)\right),
\label{eq:newstudy_sampling}
\end{equation}
where $s_g^2(e_g)$ is the (asymptotic) sampling variance in stratum $g$ as a function of $e_g$.
\end{assumption}

Under within-stratum difference-in-means, a canonical expression is
\begin{equation}
s_g^2(e_g)
=
\frac{\sigma_{1g}^2}{N_g e_g}
+
\frac{\sigma_{0g}^2}{N_g (1-e_g)},
\qquad
\sigma_{dg}^2 := \mathbb V\!\left(Y(d)\mid S=g\right),
\label{eq:sg2_main}
\end{equation}
with $N_g$ the stratum sample size (typically $N_g\approx N\pi_g$).
The nuisance quantities $\sigma_{dg}^2$ can be estimated from pilot data, administrative sources, or related experiments.
In what follows we treat them as fixed for design purposes; extending the EB prior to include them is conceptually straightforward but notationally burdensome.

\subsection{Prior-Study Information}
\label{subsec:eb_prior_formation}

The EB prior summarizes what related studies reveal about treatment effect heterogeneity in the new setting.
For the design analysis below, it is convenient to work with a harmonized parameterization: each prior study delivers an estimate of the same stratum-effect vector.

Formally, let $\theta_i\in\R^G$ denote the study-$i$ analogue of the stratum-effect vector (defined using the same stratification map $s(\cdot)$ after harmonization), and suppose study $i$ reports $(\hat\theta_i,\Sigma_i).$ We maintain the same large-sample approximation and exchangeability assumptions as before. Denote the historical evidence as
$\mathcal D_n:=\{(\hat\theta_i,\Sigma_i)\}_{i=1}^n$, and the empirical Bayes prior for the new study is
$
Q^{EB}:=\hat G_n,
$
where $\hat G_n$ is constructed from $\mathcal D_n$ using the Gaussian or NPMLE procedures in Section~\ref{sec:prior_estimation}.

In practice, prior studies rarely report a common, harmonized vector of subgroup effects.
Instead they often report noisy estimates of \emph{linear functionals} of $\theta_i$, which can be written in the generic form
\[
\hat\psi_i \mid \theta_i \ \stackrel{approx}{\sim}\ \mathcal N(R_i\theta_i,\Sigma_i),
\]
for a known reporting operator $R_i$ and reported covariance $\Sigma_i$.
This covers the main reporting regimes in applied work:
\begin{enumerate}[leftmargin=*]
\item ATE-only reporting: a study reports only an overall ATE (a single linear functional).
\item Partial subgroup tables: a study reports effects for only a subset of strata (selection matrices).
\item Heterogeneity regressions: a study reports coefficients from a heterogeneity regression (linear maps from stratum means to regression coefficients).
\item Microdata available: one can impose a harmonized definition of $S$ and compute $(\hat\theta_i,\Sigma_i)$ directly under a common analysis plan.
\end{enumerate}
The analysis in the main text assumes the harmonized case to keep the design problem transparent.
The estimation strategies in Section~\ref{sec:prior_estimation} extend to the above setting by replacing each study's likelihood contribution with the corresponding density of $\hat\psi_i$ under the assumed mixing distribution.
For Gaussian random effects, the marginal model becomes $\hat\psi_i\sim \mathcal N(R_i\tau,\Sigma_i+R_i V R_i')$.
For the NPMLE, the mixture likelihood contribution becomes $\int \varphi_{\Sigma_i}(\hat\psi_i-R_i\vartheta)\,dG(\vartheta)$.

\subsection{EB Designs and No-Information Benchmarks}
\label{subsec:eb_design_definition}

For a given prior $Q$ on $\theta$ and a given objective, define the ex-ante value of a design $e\in\mathcal H$ as
\begin{equation}
U_Q(e)
:=
\E_{\theta\sim Q}\Big[\, \E_{\hat\theta\mid\theta;\,e}\big[\, \sup_{a\in\mathcal A} W(a,\theta)\,\big]\,\Big],
\label{eq:UQ_general}
\end{equation}
where $\mathcal A$ is the post-experiment action space and $W$ is the corresponding payoff.
An optimal design under $Q$ is $$e^*(Q)\in\arg\max_{e\in\mathcal H}U_Q(e),$$ and the optimal EB design is $e^{EB}\equiv e^*(Q^{EB}).$

To make the comparison ``EB versus no information'' exact, fix a reference prior $Q^0$ that does not depend on $\mathcal D_n$ (for example, a diffuse Gaussian prior, or a symmetric benchmark prior).
The corresponding no-information design is thus $e^{NI}\equiv e^*(Q^0).$ The difference between $e^{EB}$ and $e^{NI}$ isolates how external evidence changes assignment probabilities, holding fixed the feasible design class.

\subsection{Optimal EB Designs Under Different Objectives}
\label{subsec:optimal_eb_objectives_formal}

We now characterize the optimal design under three canonical objectives that arise in applied experimental work.
Across objectives, the EB prior enters through different channels, and these differences are what generate distinct assignment rules.

\medskip

\runin{Objective I: estimation under quadratic loss}
Let the post-experiment action be an estimate $a\in\R^K$ of a target linear functional $L\theta$, where $L\in\R^{K\times G}$ is fixed (e.g., $L=\pi'$ for the ATE, or $L=I_G$ for the full heterogeneity profile).
Consider quadratic payoff
\[
W(a,\theta)=-(a-L\theta)'\Lambda(a-L\theta),\qquad \Lambda\succeq 0,
\]
so maximizing $U_Q(e)$ is equivalent to minimizing Bayes risk.

\begin{proposition}[Quadratic-loss design criterion under a Gaussian prior]
\label{prop:quadratic_design}
Suppose $Q=\mathcal N(m,V)$ and Assumption~\ref{ass:gaussian_sampling} holds.
Then the posterior covariance of $\theta$ given $(\hat\theta,e)$ is
\[
V^{post}(e) = \left(V^{-1}+\Sigma(e)^{-1}\right)^{-1},
\]
which depends on the design but not on the realized data.
Moreover, the ex-ante Bayes risk for estimating $L\theta$ equals
\begin{equation}
\mathcal R_Q(e)
=
\mathrm{tr}\!\left(\Lambda\,L V^{post}(e) L'\right).
\label{eq:bayes_risk_trace}
\end{equation}
Under a diffuse no-information prior (formally $V^{-1}\to 0$), the criterion reduces to minimizing $\mathrm{tr}(\Lambda\,L \Sigma(e) L')$.
\end{proposition}

Proposition~\ref{prop:quadratic_design} makes the role of EB transparent: with quadratic loss, design quality is entirely governed by the posterior covariance, and the prior affects the optimizer only through $V$.
When $\mathcal H$ imposes a binding constraint across strata, the optimizer must decide \emph{where} to place sampling precision; EB shifts precision toward strata whose effects remain externally uncertain.

To see the wedge relative to $e^{NI}$ in a simple case, suppose $V=\diag(v_1,\dots,v_G)$ and $L=I_G$ so the goal is to learn the heterogeneity profile.
Then the criterion is separable:
\[
\min_{e\in\mathcal H}\ \sum_{g=1}^G \lambda_g \cdot \frac{v_g\,s_g^2(e_g)}{v_g+s_g^2(e_g)},
\qquad \Lambda=\diag(\lambda_1,\dots,\lambda_G).
\]
Under the cost constraint~\eqref{eq:pscore_cost_constraint}, any interior optimum satisfies the KKT system
\begin{equation}
\lambda_g\cdot \frac{v_g^2}{\big(v_g+s_g^2(e_g)\big)^2}\cdot
\left(-\frac{\sigma_{1g}^2}{N_g e_g^2}+\frac{\sigma_{0g}^2}{N_g (1-e_g)^2}\right)
+\lambda \pi_g c_g = 0,\qquad g=1,\dots,G,
\label{eq:kkt_quadratic_formal}
\end{equation}
with $\lambda$ chosen so that $\sum_g \pi_g c_g e_g=B$ when the cost constraint binds, and truncation to $[\underline e,1-\underline e]$ applied when the interior solution violates overlap.
In the no-information limit $v_g\to\infty$, the factor $v_g^2/(v_g+s_g^2)^2\to 1$ and~\eqref{eq:kkt_quadratic_formal} collapses to the classical prior-free condition.
When EB yields heterogeneous $v_g$, the factor downweights strata with strong external evidence (small $v_g$), so the constrained optimizer reallocates informativeness toward strata where the prior leaves more to learn.

\medskip
\runin{Objective II: in-experiment welfare}
Some experiments are implemented as part of a real program rollout, so the designer values outcomes realized during the experiment.
Holding the population shares $\pi_g$ fixed, expected in-experiment net welfare is
\[
W^{exp}(e,\delta)=\sum_{g=1}^G \pi_g e_g \tau_g.
\]
Objective~II is a one-stage design problem: the designer chooses $e$ ex ante and welfare accrues during the experiment itself; the experimental data are not used for any subsequent action.
Since this objective depends on $\delta$ only through its mean, EB affects the design exclusively through predictive means $m_g:=\E_Q[\tau_g]$.

\begin{proposition}[Optimal EB propensity design for in-experiment welfare]
\label{prop:insample_formal}
Let $m_g=\E_{Q}[\tau_g]$ under a prior $Q$ on $\delta$.
Under the cost constraint~\eqref{eq:pscore_cost_constraint} and overlap bounds, an optimal design solves
\[
\max_{e\in\mathcal H}\ \sum_{g=1}^G \pi_g e_g m_g.
\]
Generically, any solution is bang-bang: there exists a cutoff $\lambda\ge 0$ such that
\[
e_g^{*}(Q)\in\arg\max_{t\in[\underline e,1-\underline e]} \{t(m_g-\lambda c_g)\}
=
\begin{cases}
1-\underline e, & m_g>\lambda c_g,\\
\underline e, & m_g<\lambda c_g,
\end{cases}
\]
with $\lambda$ chosen so that the cost constraint binds (and $\lambda=0$ if the constraint is slack).
\end{proposition}

This threshold rule reflects a basic economic force: because welfare is linear in treatment probabilities, the designer targets treatment toward strata with higher predicted net gains until the aggregate constraint binds.
A symmetric no-information benchmark with $\E_{Q^0}[\tau_g]\equiv 0$ yields a flat objective in $e$, so any feasible non-targeted assignment is optimal; a canonical benchmark sets a common propensity $e_g\equiv \bar e$ chosen to satisfy $\sum_g \pi_g c_g \bar e = B$ (truncated to $[\underline e,1-\underline e]$ if needed). EB makes the objective nondegenerate by translating historical evidence into heterogeneous $m_g$, which in turn induces systematic targeting.

\medskip
\runin{Objective III: post-experiment policy choice}
A common purpose of experimentation is to inform a subsequent policy that can target treatment by covariates.
We model the post-experiment action as a stratum-specific adoption rule $a=(a_1,\dots,a_G)'\in\{0,1\}^G$, where $a_g=1$ means treat stratum $g$ in the post-experiment rollout.
Define $\kappa_g$ as the implementation cost of policy and $\delta_g=\tau_g-\kappa_g$ as net effects. Write welfare as
\[
W^{pol}(a,\delta)=\sum_{g=1}^G \pi_g a_g \delta_g.
\]
Given experimental data $D(e)$, the Bayes-optimal adoption rule under prior $Q$ is $a_{g,Q}^*(D(e))=\mathbf 1\{\E_{Q}[\delta_g\mid D(e)]>0\}$. Under Assumption~\ref{ass:gaussian_sampling}, $D(e)$ enters only through $\hat\delta_g$ and its sampling variance $s_g^2(e_g)$, so we write equivalently $a_{g,Q}^*(\hat\delta_g;\,e_g)=\mathbf 1\{\E_{Q}[\delta_g\mid \hat\delta_g;\,e_g]>0\}$.
To obtain a closed-form design criterion, we assume a Gaussian prior $\delta_g \sim \mathcal N(m_g,v_g).$

\begin{lemma}[Distribution of the posterior mean]
\label{lem:postmean_distribution}
Let $\mu_g:=\E_{Q}[\delta_g\mid \hat\delta_g;\,e_g]$ denote the posterior mean under the Gaussian model.
Then
\[
\mu_g \sim \mathcal N\!\left(m_g,\ \sigma_{B,g}^2(e_g)\right),
\qquad
\sigma_{B,g}^2(e_g):=\frac{v_g^2}{v_g+s_g^2(e_g)}.
\]
Moreover, $\E_{Q}[\delta_g\mid \mu_g]=\mu_g$.
\end{lemma}

\begin{proposition}[Optimal design under Gaussian model]
\label{prop:policy_voI_formal}
Under the Gaussian prior and likelihood, the ex-ante expected welfare under the post-experiment optimal policy satisfies
\begin{equation}
U_Q^{pol}(e)
=
\sum_{g=1}^G \pi_g\,\Gamma_g(e_g),
\label{eq:U_policy_sum}
\end{equation}
where
\begin{equation}
\Gamma_g(e_g)
=
m_g\,\Phi\!\Big(\frac{m_g}{\sigma_{B,g}(e_g)}\Big)
+
\sigma_{B,g}(e_g)\,\phi\!\Big(\frac{m_g}{\sigma_{B,g}(e_g)}\Big),
\qquad
\sigma_{B,g}^2(e_g)=\frac{v_g^2}{v_g+s_g^2(e_g)},
\label{eq:Gamma_policy_formal}
\end{equation}
and $\phi$ and $\Phi$ denote the standard normal pdf and cdf.
If the cost constraint~\eqref{eq:pscore_cost_constraint} binds and the optimum is interior, the KKT system can be written as
\begin{equation}
\phi\!\Big(\frac{m_g}{\sigma_{B,g}(e_g)}\Big)\cdot
\frac{\partial \sigma_{B,g}(e_g)}{\partial e_g}
+\lambda c_g=0,\qquad g=1,\dots,G,
\label{eq:kkt_policy_formal}
\end{equation}
with truncation to $[\underline e,1-\underline e]$ applied at the boundary.
\end{proposition}

The form~\eqref{eq:Gamma_policy_formal} highlights an exploration--exploitation logic inside a randomized design.
The factor $\phi(m_g/\sigma_{B,g}(e_g))$ in~\eqref{eq:kkt_policy_formal} is largest when $m_g$ is close to zero relative to posterior-mean dispersion, so (holding costs fixed) experimentation concentrates where the adoption decision is most uncertain given prior evidence.
EB matters twice: $m_g$ determines how close the decision is to the threshold, and $v_g$ determines how responsive posterior beliefs are to additional data.

A useful benchmark is the symmetric no-information case with $m_g= 0$ and $v_g= \bar{v}$.
Then $\Gamma_g(e_g)=\sigma_{B,g}(e_g)\phi(0)$ is strictly increasing in informativeness, so $e^{NI}$ coincides with a variance-minimizing (information-maximizing) allocation subject to $\mathcal H$.
Under EB, heterogeneity in $(m_g,v_g)$ creates a wedge relative to $e^{NI}$: the designer explores more in strata where the policy decision is near the boundary and external evidence is imprecise, and explores less where prior evidence already makes the sign clear (large $|m_g|$) or externally precise (small $v_g$).

\medskip
\runin{Takeaway} Across objectives, EB affects the propensity design through distinct features of the EB prior.
Under quadratic-loss estimation, EB enters through the prior dispersion $V$ in the posterior variance formula~\eqref{eq:bayes_risk_trace}, reshaping constrained allocations by downweighting strata that are already well learned from external evidence.
Under in-experiment welfare, EB enters only through prior means and induces targeting toward strata with higher predicted treatment effects.
Under post-experiment policy choice, EB enters through both means and variances via~\eqref{eq:Gamma_policy_formal}, concentrating information where it is most likely to change the sign of the adoption decision.

\begin{remark}[Experiment with noncompliance]
If the experiment randomizes a voucher $Z$ rather than treatment $D$ itself, the design choice is the assignment rate
$e_g:=\Pr(Z=1\mid S=g)$.
Suppose the payoff-relevant estimand is the intent-to-treat effects (ITT),
\[
\tau_g^{ITT}:=\E[Y\mid Z=1,S=g]-\E[Y\mid Z=0,S=g],
\]
which is a natural choice for policymakers deciding whether to offer a program or not.
Relative to the perfect-compliance setting, all Objective I--III expressions are unchanged in form after substituting
$
\tau_g \ \leadsto\ \tau_g^{ITT},
$
and
$$
s_g^2(e_g)\ \leadsto\ s_{g,ITT}^2(e_g)
:=\frac{\sigma_{Y,1g}^2}{N_g e_g}+\frac{\sigma_{Y,0g}^2}{N_g(1-e_g)},
\quad \sigma_{Y,zg}^2:=\Var(Y\mid Z=z,S=g).
$$
The resource constraint is unchanged if resources are spent on vouchers. If instead resources scale with realized treatment, we need to replace $e_g$ by expected take-up
$q_g(e_g)=\E[D\mid S=g]=d_{0g}+e_g\rho_g$ with $d_{0g}:=\E[D\mid Z=0,S=g]$ and $\rho_g:=\E[D\mid Z=1,S=g]-\E[D\mid Z=0,S=g]$, which gives $\sum_g \pi_g c_g q_g(e_g)\le B$. A special case is when $d_{0g}=0$, then the constraint is simply $\sum_g \pi_g c_g e_g\rho_g\le B.$

To satisfy the exchangeability assumption on the study-level parameters, prior evidence should be encoded for the same estimand,
i.e., in terms of $\tau_g^{ITT}$ for a comparable voucher $Z$; evidence reported as effects of treatment $D$ (e.g., treatment effect on the treated TOT, local average treatment effect LATE) is not directly exchangeable across settings when compliance differs, and should only be used after mapping to ITT (e.g., $\tau_g^{ITT}=\rho_g\tau_g^{LATE}$ when the compliance rate $\rho_g$ is available).

\end{remark}

\section{Theoretical Analysis}
\label{sec:theory}

This section develops theoretical guarantees for EB design. The organizing idea is simple: a design is chosen by maximizing a value functional that depends on the prior. If we estimate the prior and then optimize, performance depends on how estimation error propagates through the optimization step.

\subsection{Regret Bounds and Gains over No Prior Evidence}
\label{subsec:oracle-inequality}

For any probability measure $Q$ on $\R^d$, define the ex-ante value of design $\eta$ under prior $Q$ as
\begin{equation}
U_Q(\eta)
:=\int \Bigg[\sup_{a\in\mathcal A}\int W(a,\theta)\,p_\eta(y\mid \theta)\,dQ(\theta)\Bigg]\,d\lambda(y),
\label{eq:value-functional}
\end{equation}
where we write $p_\eta(y|\,\theta)$ for the likelihood under design $\eta$.
The oracle design $\eta^O \in \argmax_{\eta\in\mathcal H} U_G(\eta)$ maximizes value under the true prior $G$. The EB design $\eta^{EB}\in \argmax_{\eta\in\mathcal H} U_{\hat G_n}(\eta)$ maximizes value under the estimated prior $\hat G_n$.

The first step is to ensure $U_Q(\eta)$ is well-defined and uniformly controlled over the design class. We thus impose the following assumption.

\begin{assumption}[Well-posed value]
\label{ass:value-envelope}
Let $\mathcal Q$ be a collection of probability measures on $\R^d$ such that $G\in\mathcal Q$ and $\Pr(\hat G_n\in\mathcal Q)\to 1$.
For $Q\in\mathcal Q$ and $\eta\in\mathcal H$, define
$u_Q(y;\eta):=\sup_{a\in\mathcal A}\int W(a,\theta)\,p_\eta(y\mid\theta)\,dQ(\theta)$.
Assume:
(i) for each $Q\in\mathcal Q$ and $\eta\in\mathcal H$, $u_Q(\cdot;\eta)$ is $\lambda$-measurable and finite $\lambda$-a.e., and the supremum is attained by a measurable selector $a_Q^*(y;\eta)$;
(ii) $\sup_{Q\in\mathcal Q}\sup_{\eta\in\mathcal H}\int |u_Q(y;\eta)|\,d\lambda(y)<\infty$.
\end{assumption}

Assumption~\ref{ass:value-envelope} rules out pathological cases where the Bayes act fails to exist or the value functional is not integrable. In the leading examples, it is satisfied by simple moment conditions (e.g.\ bounded $W$, or quadratic loss combined with finite second moments under $Q$).

Our first result is a finite-sample oracle inequality. It shows that the welfare loss from optimizing under $\hat G_n$ rather than $G$ is controlled by how well $U_{\hat G_n}(\eta)$ approximates $U_G(\eta)$ uniformly over $\eta$.

\begin{theorem}[Finite-sample oracle inequality]
\label{thm:oracle-ineq}
Suppose Assumption~\ref{ass:value-envelope} holds. Then, for any realization of $\mathcal D_n$,
\begin{equation}
0 \le U_G(\eta^O)-U_G(\eta^{EB}) \le 2\Delta_n,
\qquad
\Delta_n := \sup_{\eta\in\mathcal H}\big|U_{\hat G_n}(\eta)-U_G(\eta)\big|.
\label{eq:oracle-ineq}
\end{equation}
We define $\Reg_n:=U_G(\eta^O)-U_G(\eta^{EB})$ as the EB regret.
\end{theorem}

The intuition for Theorem~\ref{thm:oracle-ineq} is geometric. Think of $\eta\mapsto U_G(\eta)$ and $\eta\mapsto U_{\hat G_n}(\eta)$ as two objective surfaces. If the surfaces are uniformly close, then maximizing the estimated surface cannot lead to a design that performs much worse under the true surface.

To relate EB to a benchmark that ignores previous evidence, fix a baseline prior $G_0$ independent of $\mathcal D_n$ and define the no-information design $\eta^{NI}\in\argmax_{\eta\in\mathcal H}U_{G_0}(\eta)$.

\begin{corollary}[EB versus no-information benchmark]
\label{cor:eb-vs-np}
Under Assumption~\ref{ass:value-envelope},
\[
U_G(\eta^{EB})-U_G(\eta^{NI})
\ge
\big(U_G(\eta^O)-U_G(\eta^{NI})\big)-2\Delta_n.
\]
In particular, if $\Delta_n \le \frac{1}{2}(U_G(\eta^O)-U_G(\eta^{NI}))$, then $U_G(\eta^{EB})\ge U_G(\eta^{NI})$.
\end{corollary}

Corollary~\ref{cor:eb-vs-np} provides a simple sufficient condition for comparing the EB design with the no-information benchmark, expressed in terms of the uniform deviation $\Delta_n$. Because $\Delta_n$ is a worst-case measure of the maximum discrepancy between the two objective surfaces over the entire design class, the resulting sufficient condition can be conservative in applications.

\begin{remark}[A less conservative sufficient condition]
\label{rem:local-eb-ni}
Let $\delta_G:=U_G(\eta^O)-U_G(\eta^{NI})\ge 0$ denote the oracle gain over the no-information benchmark.
Then
\[
U_G(\eta^{EB})-U_G(\eta^{NI})=\delta_G-\Reg_n.
\]
Moreover, writing $x_+:=\max\{x,0\}$, the decomposition in the proof of Theorem~\ref{thm:oracle-ineq} implies the sharper bound
\[
\Reg_n
\le
\big(U_G(\eta^O)-U_{\hat G_n}(\eta^O)\big)_+
+
\big(U_{\hat G_n}(\eta^{EB})-U_G(\eta^{EB})\big)_+.
\]
Thus, a sufficient condition for $U_G(\eta^{EB})\ge U_G(\eta^{NI})$ is that the two one-sided errors on the right-hand side are jointly smaller than $\delta_G$.
This condition depends only on approximation errors evaluated at $\eta^O$ and $\eta^{EB}$, rather than the uniform deviation $\Delta_n$.
\end{remark}

To illustrate the comparison to $\eta^{NI}$ in a setting where the oracle gain $\delta_G$ and the event $\{U_G(\eta^{EB})\ge U_G(\eta^{NI})\}$ admit closed-form characterizations, we consider the following two-stratum quadratic-loss allocation problem, which is a simple special case of the quadratic-loss criterion in Proposition~\ref{prop:quadratic_design}.

\begin{example}[Two-stratum quadratic loss]
\label{emp:two_stratum_quadratic}
In the quadratic-loss Gaussian framework of Proposition~\ref{prop:quadratic_design} (Section~\ref{sec:eb_pscore_design}), impose the simplifications $G=2$, $L=I_2$, $\Lambda=I_2$, and consider stratum sample allocation designs $\eta=(N_1,N_2)$ with $N_1+N_2=N$ and within-stratum estimator variance $s^2/N_g$.
Let $(v_1,v_2)$ be the true prior variances and $(\hat v_1,\hat v_2)$ the EB plug-in estimates.
Take the no-information benchmark to be the diffuse-prior limit, so $\eta^{NI}=(N/2,N/2)$, and assume the EB allocation is interior.
Then $U_G(\eta^{EB})\ge U_G(\eta^{NI})$ if and only if
\[
\left(\frac{1}{\hat v_2} - \frac{1}{\hat v_1}\right) \left(\frac{1}{v_2} - \frac{1}{v_1}\right)\ge 0
\quad\text{and}\quad
\left|\frac{1}{\hat v_2} - \frac{1}{\hat v_1}\right|\le 2\left|\frac{1}{v_2} - \frac{1}{v_1}\right|.
\]
\end{example}

In this simple setting, EB beats the no-information design as long as it (i) gets the direction of the cross-stratum prior precision gap right (equivalently, which stratum is more externally uncertain), and (ii) does not exaggerate that gap by more than a factor of two. Appendix~\ref{app:two_stratum_quadratic} provides the full details.

\subsection{Regret Consistency}
\label{subsec:asymp-oracle}

The finite-sample oracle bound reduces the problem of regret consistency to showing $\Delta_n\to 0$. The next theorem makes that reduction explicit.

\begin{theorem}[Regret consistency]
\label{thm:regret-consistency}
Suppose Assumption~\ref{ass:value-envelope} holds and $\Delta_n\to 0$ in probability. Then $U_G(\eta^O)-U_G(\eta^{EB}) \to 0$ in probability.
\end{theorem}

To show $\Delta_n\to 0$, a natural starting point is weak convergence $\hat G_n\Rightarrow G$. But weak convergence alone does not automatically imply uniform convergence of the value functional. The next assumption imposes continuity and moment conditions that allow one to pass limits through the integrals defining $U_Q(\eta)$ uniformly over $\eta$.

\begin{assumption}[Continuity and envelope]
\label{ass:cont-envelope}
Assume:
(i) $\mathcal H$ and $\mathcal A$ are compact metric spaces;
(ii) for $\lambda$-a.e.\ $y$, the map $(\eta,\theta)\mapsto p_\eta(y\mid\theta)$ is continuous, and there exists an integrable envelope $\bar p:\mathcal Y\to[0,\infty)$ with $\sup_{\eta,\theta}p_\eta(y\mid\theta)\le \bar p(y)$;
(iii) $W(a,\theta)$ is continuous in each argument, and there exists a continuous coercive function $w:\R^d\to[0,\infty)$ with $\sup_{a}|W(a,\theta)|\le w(\theta)$;
(iv) there exists $\delta>0$ such that $\sup_{Q\in\mathcal Q}\int w(\theta)^{1+\delta}\,dQ(\theta)<\infty$.
\end{assumption}

Assumption~\ref{ass:cont-envelope} is a standard ``uniform dominated convergence'' device: continuity gives pointwise convergence when priors converge weakly, the envelope $\bar p$ controls likelihood variation across designs, and the moment condition controls the tails that matter when welfare grows with $\|\theta\|$.

The following Lemma~\ref{lem:weak-to-Delta-unbounded} provides the main bridge: any prior estimator that is weakly consistent (and satisfies the moment conditions) delivers regret consistency via Theorem~\ref{thm:regret-consistency}.

\begin{lemma}[Weak convergence implies value convergence]
\label{lem:weak-to-Delta-unbounded}
Under Assumptions~\ref{ass:value-envelope} and~\ref{ass:cont-envelope}, if $\hat G_n\Rightarrow G$ in probability and $\Pr(\hat G_n\in\mathcal Q)\to 1$, then $\Delta_n\to 0$ in probability.
\end{lemma}

In Gaussian parametric EB, parameter convergence implies weak convergence and all required moments exist automatically, so the following corollary is a direct application of Lemma~\ref{lem:weak-to-Delta-unbounded}.

\begin{corollary}[Gaussian EB is oracle-optimal]
\label{cor:gaussian-eb-oracle-unbounded}
Assume that $\theta_i \ \stackrel{iid}{\sim}\ \mathcal N(\tau,V),$ $i=1,\dots,n+1.$
Suppose Assumptions~\ref{ass:value-envelope} and~\ref{ass:cont-envelope} hold and that $(\hat\tau_n,\hat V_n)\to(\tau,V)$ in probability.
Then $\Delta_n\to 0$ in probability, and consequently $U_G(\eta^{EB})-U_G(\eta^O)\to 0$ in probability.
\end{corollary}

A useful special case of Corollary~\ref{cor:gaussian-eb-oracle-unbounded} arises under a quadratic loss objective in a Gaussian–Gaussian experiment, where both the prior and the likelihood are Gaussian. In this case, the value admits a closed-form expression, and continuity in Assumption~\ref{ass:cont-envelope} can be verified directly.

\begin{proposition}[Quadratic welfare: $\Delta_n\to 0$ under $\hat V_n\to V$]
\label{prop:quad-gaussian-Delta}
Assume a Gaussian-Gaussian experiment: $\theta_i \ \stackrel{iid}{\sim}\ \mathcal N(\tau,V)$ and $Y_\eta=\theta_{n+1}+\varepsilon_\eta$ with $\varepsilon_\eta\sim\mathcal N(0,\Sigma(\eta))$ independent of $\theta_i$ and $\Sigma(\eta)\succ 0$.
Let $W(a,\theta)=-\|a-\theta\|_2^2$ and $\mathcal A=\R^d$.
If $\sup_{\eta\in\mathcal H}\|\Sigma(\eta)^{-1}\|<\infty$ and $\hat V_n\to V$ in probability, then $\Delta_n\to 0$ in probability.
\end{proposition}

The key intuition is that with quadratic loss, ex-ante value depends on the design only through posterior covariance, and posterior covariance depends continuously on $V$ as long as designs do not produce degenerate covariance matrices.

For the NPMLE, we impose two standard regularity conditions: compact support of the mixing distribution $G$ and uniform eigenvalue bounds on the reported covariance matrices.
Both conditions are also used in \cite{SoloffGuntuboyinaSen2025} for proving the rate of convergence of prior estimation.

\begin{assumption}[Compact support and uniform eigenvalue bounds for NPMLE]
\label{ass:npmle-consistency}
The cross-study distribution $G$ is supported on $\Theta_R:=\{\theta\in\mathbb R^d:\ \|\theta\|_2\le R\}$ for some $R<\infty$. Define $\mathcal G_R:=\{Q:\mathrm{supp}(Q)\subseteq \Theta_R\}$.
There exist constants $0<\underline\sigma^2\le \overline\sigma^2<\infty$ such that the reported covariances in the historical evidence satisfy
\[
\underline\sigma^2 I_d \preceq \Sigma_i \preceq \overline\sigma^2 I_d,\qquad i=1,2,\dots, n.
\]
\end{assumption}

\begin{proposition}[NPMLE prior is weakly consistent under heteroskedastic Gaussian errors]
\label{prop:npmle-weak-consistency}
Let $\hat G_n$ denote a maximizer of the NPMLE objective~\eqref{eq:npmle_general} over the restricted class $\mathcal G_R$ from Assumption~\ref{ass:npmle-consistency}.
If Assumption~\ref{ass:npmle-consistency} holds, then $\hat G_n \Rightarrow G$ almost surely.

If, in addition, Assumptions~\ref{ass:value-envelope} and~\ref{ass:cont-envelope} hold with $\mathcal Q=\mathcal G_R$,
then $\Delta_n\to 0$ in probability by Lemma~\ref{lem:weak-to-Delta-unbounded}, and the EB design based on the NPMLE is regret consistent by
Theorem~\ref{thm:regret-consistency}.
\end{proposition}

The proof follows analogous logic to that used for establishing consistency of the NPMLE under i.i.d. mixture models (in our setting, all $\Sigma_i$s are equal), as reviewed in \citet[Sec.~2]{Chen2017}.
Specifically, we verify the condition in the ``trivialized Wald theorem'' \citep[Thm.~2.1]{Chen2017} by showing that outside any weak neighborhood of $G$ the average expected log-likelihood is uniformly smaller.
In the i.i.d.\ mixture case, Chen verifies this condition using Pfanzagl-type inequalities \citep{Pfanzagl1988}; in our heteroskedastic error setting we establish the same condition uniformly over the covariance class implied by Assumption~\ref{ass:npmle-consistency}, and then replace the strong law of large numbers for i.i.d. variables with Kolmogorov’s strong law for independent, non-identically distributed variables to accommodate varying $\Sigma_i$.

\subsection{Rates of Convergence}
\label{sec:rate_n}

The consistency results above reduce EB optimality to a single quantitative object: the uniform value-approximation error
\[
\Delta_n=\sup_{\eta\in\mathcal H}\big|U_{\hat G_n}(\eta)-U_G(\eta)\big|.
\]
By the finite-sample oracle inequality (Theorem~\ref{thm:oracle-ineq}), $\Reg_n\le 2\Delta_n$, so rate statements for regret follow once we control
$\Delta_n$. 

\bigskip
\runin{Gaussian limit experiments}
In many randomized experiments, the raw outcome vector is high-dimensional while the payoff-relevant parameter is low-dimensional.
In such settings, it is often natural to work with sufficient statistics (e.g., regression coefficients) or with limit-experiment approximations.
A theoretical justification is that, under regularity conditions, a broad class of experiments can be approximated by Gaussian limit experiments uniformly over design classes \citep{Higbee2024}.
Accordingly, to obtain explicit regret rates, we work with the Gaussian limit experiment induced by the low-dimensional statistic that is used for downstream decision-making.

\begin{assumption}[Gaussian working likelihood for rate analysis]
\label{as:gaussian_summary_lik}
Let $T_\eta(\cdot)$ be the statistic from the new study that is used for downstream decision-making
(e.g., an estimator of $\theta$ computed from the raw data $Y$), and write $\hat\theta_\eta := T_\eta(Y)\in\mathbb R^d$.
We take the design-dependent working likelihood to be Gaussian:
\[
\hat\theta_\eta \ |\ \theta;\eta\ \sim\ \mathcal N(\theta,\Sigma(\eta)),
\]
where $\Sigma(\eta)$ is the sampling covariance matrix of $\hat\theta_\eta$ under design $\eta$.
Thus, the likelihood entering the value $U_Q(\eta)$ is $p_\eta(\hat\theta\mid\theta)=\varphi_{\Sigma(\eta)}(\hat\theta-\theta)$ rather than the raw-data likelihood $p_\eta(y\mid\theta)$.

Moreover, there exist constants $0<\underline\sigma^2\le \overline\sigma^2<\infty$ such that for all $\eta\in\mathcal H$,
\[
\underline\sigma^2 I_d \preceq \Sigma(\eta)\preceq \overline\sigma^2 I_d.
\]
For notational convenience, in the remainder of this subsection we write $Y$ for the observed statistic $\hat\theta_\eta$.
\end{assumption}

\medskip
\runin{Class of decision problems}
The rate results in this section are stated for a broad class of decision problems characterized by Assumption~\ref{as:ql_state}. This class encompasses many widely used settings, including point estimation, hypothesis testing, portfolio choice, binary treatment adoption, and best-arm selection. Table~\ref{tab:ql_examples} summarizes these canonical examples and their corresponding continuation values.

\begin{assumption}[Quasi-linear-in-state welfare]
\label{as:ql_state}
There exist functions $w_0:\Theta\to\mathbb R$, $u:\mathcal A\to\mathbb R$, and $v:\mathcal A\to\mathbb R^d$ such that
$W(a,\theta)=w_0(\theta)+u(a)+v(a)^\prime\theta$, and the continuation-value map
$\Psi(m):=\sup_{a\in\mathcal A}\{u(a)+v(a)^\prime m\}$ is finite and measurable.
\end{assumption}

\begin{table}[t]
\centering
\caption{Common decision problems satisfying Assumption~\ref{as:ql_state}}
\label{tab:ql_examples}

\setlength{\tabcolsep}{6pt}
\renewcommand{\arraystretch}{1.35} 

\begin{threeparttable}
\begin{tabular}{p{5.8cm} p{5.5cm} p{4.6cm}}
\toprule
Typical use & Welfare $W(a,\theta)$ & Cont. value $\Psi(m)$ \\
\midrule

\makecell[l]{Point estimation}
& $-(a-\theta)^\prime\Lambda(a-\theta)$
& $m^\prime\Lambda m$
\\ \addlinespace[3pt]

\makecell[l]{Portfolio choice}
& $a^\prime\theta-\frac{\gamma}{2}a^\prime\Sigma a$
& $\frac{1}{2\gamma}m^\prime\Sigma^{-1}m$
\\ \addlinespace[3pt]

\makecell[l]{Binary adoption}
& $\sum_{g=1}^G\pi_g a_g\theta_g,\ \ a\in\{0,1\}^G$
& $\sum_{g=1}^G\pi_g\max\{0,m_g\}$
\\ \addlinespace[3pt]

\makecell[l]{Ranking \& selection}
& \makecell[l]{choose $j\in\mathcal{I}$, $W(j,\theta)=\theta_j$}
& $\max_{j\in\mathcal{I}}\, m_j$
\\ \addlinespace[3pt]

\makecell[l]{Hypothesis testing}
& $W(\varphi,\theta)=-L(\theta,\varphi)$
& $\max\{-a_0(1-m),-a_1 m\}$
\\

\bottomrule
\end{tabular}

\begin{tablenotes}[flushleft]
\footnotesize
\item \textit{Notes: For the hypothesis testing, a standard weighted $0$--$1$ loss is
\[
L(\theta,\varphi)=
\begin{cases}
0, & \varphi=\mathbbm{1}\{\theta \in \Theta_1\} \quad \text{(correct decision)},\\
a_0, & \varphi=1,\ \theta \in \Theta_0 \quad \text{(Type I error)},\\
a_1, & \varphi=0,\ \theta \in \Theta_1 \quad \text{(Type II error)}.
\end{cases}
\]
Let $s(\theta):=\mathbbm{1}\{\theta\in\Theta_1\}$ and $m=\E[s(\theta)\mid Y]$, then
$\Psi(m)=\max\{-a_0(1-m),-a_1 m\}$.}
\end{tablenotes}

\end{threeparttable}
\end{table}

The following lemma characterizes a defining feature of this class of decision problems: the welfare is affine in the payoff-relevant state $\theta$, so the effect of the prior can be summarized through posterior means.

\begin{lemma}[Value reduction to posterior means]
\label{lem:value_reduction}
Suppose Assumption~\ref{as:ql_state} holds and all expectations below are finite. Then for any prior $Q$ and design $\eta$,
\[
U_Q(\eta)=\mathbb E_{Y\sim f_{Q,\eta}}\big[\Psi(\mu_{Q,\eta}(Y))\big]+\mathbb E_Q[w_0(\theta)],
\]
where $\mu_{Q,\eta}(Y)$ is the posterior mean of $\theta$ given data $Y$, and $f_{Q,\eta}$ is the predictive density for $Y$. In particular,
$\max_{\eta\in\mathcal H}U_Q(\eta)$ is equivalent to $\max_{\eta\in\mathcal H}\mathbb E[\Psi(\mu_{Q,\eta}(Y))]$.
\end{lemma}

Lemma~\ref{lem:value_reduction} makes clear what must be controlled to obtain rates: we need bounds on how $\Psi(\mu_{Q,\eta}(Y))$ changes when
$Q$ is replaced by $\hat G_n$, uniformly over $\eta\in\mathcal H$.

\bigskip
\runin{First-order rates}
First-order regret bounds treat the design problem as a plug-in maximization problem: once $\hat G_n$ is close to $G$ and the objective is stable
to small prior perturbations, the uniform deviation $\Delta_n$ is small and $\Reg_n\le 2\Delta_n$ inherits the same rate.
We present this logic for (i) Gaussian parametric EB and (ii) nonparametric EB via NPMLE.

\begin{assumption}[Gaussian EB estimation rate]
\label{as:gauss_est_rate}
Assume $\theta_i \ \stackrel{iid}{\sim}\ \mathcal N(\tau_0,V_0)$, $i=1,\dots,n+1$, where $V_0$ is positive definite.
There exists an estimator $\hat G_n=\mathcal N(\hat\tau_n,\hat V_n)$ such that
$\|\hat\tau_n-\tau_0\|_2 + \|\hat V_n - V_0\|_{\mathrm F} = O_{p}(n^{-1/2})$.
\end{assumption}

\begin{assumption}[Growth-controlled Lipschitz continuation value]
\label{as:Psi_Lip}
There exist constants $L<\infty$ and $p\ge 0$ such that for all $m,m'\in\mathbb R^d$,
\[
|\Psi(m)-\Psi(m')|\le L\,(1+\|m\|_2^p+\|m'\|_2^p)\,\|m-m'\|_2.
\]
\end{assumption}

Assumption~\ref{as:Psi_Lip} is a stability condition on the decision problem. It allows $\Psi$ to be non-smooth (e.g.\ threshold rules) and to be unbounded (e.g.\ quadratic loss), while still ensuring that changes in posterior means translate into controlled changes in value. In Appendix~\ref{app:Psi_Lip_examples}, we verify that the continuation values associated with the five decision problems summarized in Table~\ref{tab:ql_examples} all satisfy Assumption~\ref{as:Psi_Lip}.

\begin{theorem}[Gaussian EB: $\sqrt{n}$ rate]
\label{thm:gaussian_rootn}
Under Assumptions~\ref{as:gaussian_summary_lik},~\ref{as:ql_state},~\ref{as:gauss_est_rate}, and~\ref{as:Psi_Lip} with Gaussian EB,
$\Delta_n = O_{p}(n^{-1/2})$ and hence $\Reg_n := U_G(\eta^O)-U_G(\eta^{EB}) = O_{p}(n^{-1/2})$.
\end{theorem}

Theorem~\ref{thm:gaussian_rootn} shows that, for this class of decision problems, parametric EB design regret inherits the parametric prior-estimation
rate. The reason is that under Gaussian conjugacy, posterior means (and predictive distributions) depend smoothly on $(\tau,V)$, so a
$\sqrt{n}$ perturbation in the estimated hyperparameters produces a uniform $O_{p}(n^{-1/2})$ perturbation of the value functional.

We next state the corresponding first-order bound for nonparametric EB based on NPMLE.

\begin{assumption}[NPMLE accuracy for prediction and posterior means]
\label{as:npmle_inputs}
Suppose Assumptions~\ref{ass:npmle-consistency} and~\ref{as:gaussian_summary_lik} hold.
For any positive definite $\Sigma$, write
\[
f_{G,\Sigma}(x):=\int \varphi_{\Sigma}(x-\vartheta)\,dG(\vartheta),
\qquad
\mu_{G,\Sigma}(x):=\E[\vartheta\mid X=x],\ \ X\mid\vartheta\sim\mathcal N(\vartheta,\Sigma),
\]
for the predictive density and posterior mean under prior $G$.
The (approximate) NPMLE $\hat G_n$ satisfies the following accuracy bounds:
\begin{enumerate}[label=(\roman*),leftmargin=*]
\item Predictive density:
$\displaystyle \sup_{\underline\sigma^2 I_d\preceq\Sigma\preceq\overline\sigma^2 I_d} h^2(f_{\hat G_n,\Sigma},f_{G,\Sigma})\le C_1(\log n)^{d+1}/n$.
\item Posterior mean:
$\displaystyle \sup_{\underline\sigma^2 I_d\preceq\Sigma\preceq\overline\sigma^2 I_d}\E_{Z\sim f_{G,\Sigma}}\!\big[\|\mu_{\hat G_n,\Sigma}(Z)-\mu_{G,\Sigma}(Z)\|_2^2\big]\le C_2(\log n)^{d+\max(d/2,4)}/n$.
\end{enumerate}
Here $h$ denotes Hellinger distance.
The rates in (i)--(ii) correspond to the compact-support cases of \citet[Corollary~8 and Theorem~9]{SoloffGuntuboyinaSen2025} for the Gaussian location-mixture NPMLE; we state them in this uniform-over-$\Sigma$ form because our design criterion ranges over $\eta\in\mathcal H$.
\end{assumption}

\begin{remark}
\label{rem:npmle_rates_which_metric}
Assumption~\ref{as:npmle_inputs} controls the predictive density $f_{G,\Sigma}$ and the posterior mean $\mu_{G,\Sigma}$ rather than a distance between mixing distributions $G$ and $\hat G_n$.
This is intentional: by Lemma~\ref{lem:value_reduction}, the value $U_G(\eta)$ depends on the prior only through these objects.
\citet[Theorem~10]{SoloffGuntuboyinaSen2025} also provide deconvolution bounds for the mixing distribution itself under additional structure (diagonal covariances), e.g.\ $\E[W_2^2(\hat G_n,G)]\lesssim 1/\log n$; these rates are slower and are not needed for our regret bounds.
\end{remark}

The following theorem establish the first-order rate of regret under NPMLE. 
\begin{theorem}[NPMLE EB: first-order rate]
\label{thm:npmle_rootn}
Under Assumptions~\ref{as:gaussian_summary_lik},~\ref{as:ql_state},~\ref{as:Psi_Lip}, and~\ref{as:npmle_inputs},
$\Delta_n = O_{p}\left(\frac{(\log n)^{(d+\max(d/2,4))/2}}{\sqrt{n}}\right)$ and hence $\Reg_n = O_{p}\!\left(\frac{(\log n)^{(d+\max(d/2,4))/2}}{\sqrt{n}}\right).$
\end{theorem}

Theorem~\ref{thm:npmle_rootn} shows that EB remains oracle-optimal with a nonparametric prior estimator, but the regret convergence rate is slower (up to logs) because estimating an entire mixing distribution is intrinsically harder than estimating finitely many hyperparameters. 
As a result, if the number of prior studies is small and a Gaussian distribution is a plausible approximation for the distribution of study-level parameters, then a Gaussian EB prior is a sensible choice.

\bigskip
\runin{Second-order rates}
The first-order bounds above control regret through the uniform value error $\Delta_n$, and therefore apply to a wide range of (possibly non-smooth)
decision problems covered by Assumptions~\ref{as:ql_state} and~\ref{as:Psi_Lip}. In many smooth design problems, regret is locally quadratic: if the
oracle objective $U_G$ has curvature at its maximizer, then using a design that is close to $\eta^O$ induces only a second-order welfare loss.
The next two assumptions formalize this refinement.

\begin{assumption}[Strong concavity in design]
\label{as:strong_concavity}
$\mathcal H\subset\mathbb R^k$ is convex and compact, $U_G$ is continuously differentiable on $\mathcal H$, and there exists $m>0$ such that for all $\eta_1,\eta_2\in\mathcal H$:
\[
U_G(\eta_2) \le U_G(\eta_1) + \nabla U_G(\eta_1)^\prime(\eta_2-\eta_1) - \frac{m}{2}\|\eta_2-\eta_1\|_2^2.
\]
Moreover, $\eta^O$ is the unique maximizer and lies in the relative interior of $\mathcal H$.
\end{assumption}

\begin{assumption}[Uniform gradient approximation]
\label{as:grad_approx}
There exists a random sequence $r_n=o_{p}(1)$ such that
$\sup_{\eta\in\mathcal H}\|\nabla U_{\hat G_n}(\eta)-\nabla U_G(\eta)\|_2 \le r_n$ with probability approaching one.
\end{assumption}

Assumption~\ref{as:strong_concavity} imposes differentiability and curvature of $U_G$ on the
feasible design set, yielding a unique interior maximizer and a quadratic upper bound on value losses as a function of $\|\eta-\eta^O\|_2$.
Assumption~\ref{as:grad_approx} replaces uniform value control by a uniform bound on the gradient error induced by using $\hat G_n$; in regular
parametric settings this gradient error typically inherits the $n^{-1/2}$ rate of hyperparameter estimation. These conditions are most natural for smooth objectives (e.g.\ quadratic-loss estimation or mean--variance criteria under smooth propensity-score or allocation design classes). In contrast, for maximum or threshold objectives (binary adoption, best-alternative selection, and hypothesis testing), the induced objective can be non-differentiable and optimizers often occur at corners, so second-order guarantees are generally inappropriate.

\runin{Example (Cont.): two-stratum quadratic loss}
In Example~\ref{emp:two_stratum_quadratic}, both Assumptions~\ref{as:strong_concavity} and~\ref{as:grad_approx} can be verified explicitly.
In Appendix~\ref{app:two_stratum_secondorder}, we derive a closed-form expression for $U_G(\eta)$ under the two-stratum Gaussian quadratic-loss
allocation problem and show that $U_G$ is strongly concave on the feasible allocation set $\mathcal H=\{(N_1,N_2)\in\R_+^2:\ N_1+N_2=N\}$,
with a unique interior maximizer under the maintained interior condition from Example~\ref{emp:two_stratum_quadratic}.
Moreover, if the stratum prior variances $v_g$ admit EB estimates $\hat v_g$ satisfying $|\hat v_g - v_g|=O_p(n^{-1/2})$, then the induced plug-in objective satisfies the uniform gradient approximation condition
$\sup_{\eta\in\mathcal H}\|\nabla U_{\hat G_n}(\eta)-\nabla U_G(\eta)\|_2=O_{p}(n^{-1/2})$.

Theorem~\ref{thm:second_order_general} isolates the mechanism behind faster rates: curvature $m$ converts gradient errors into quadratic value losses.
\begin{theorem}[Second-order design regret bound]
\label{thm:second_order_general}
Suppose Assumptions~\ref{as:strong_concavity} and~\ref{as:grad_approx} hold and that $\hat\eta$ is an interior maximizer of $U_{\hat G_n}$ over $\mathcal H$. Then
\[
\Reg_n \le \frac{r_n^2}{2m}.
\]
\end{theorem}

As an immediate consequence of Theorem~\ref{thm:second_order_general}, we obtain the second-order rates for Gaussian EB and NPMLE EB.
\begin{corollary}[Second-order rate for Gaussian EB]
\label{cor:second_order_gaussian}
Under Assumptions~\ref{as:gaussian_summary_lik},~\ref{as:ql_state},~\ref{as:gauss_est_rate},~\ref{as:Psi_Lip},~\ref{as:strong_concavity}, and~\ref{as:grad_approx} with $r_n=O_p(n^{-1/2})$, we have $\Reg_n = O_p(n^{-1})$.
\end{corollary}

\begin{corollary}[Second-order rate for NPMLE EB]
\label{cor:second_order_npmle}
Under Assumptions~\ref{as:gaussian_summary_lik},~\ref{as:ql_state},~\ref{as:Psi_Lip},~\ref{as:npmle_inputs},~\ref{as:strong_concavity}, and~\ref{as:grad_approx} with $r_n=O_p\left((\log n)^{(d+\max(d/2,4))/2}/\sqrt{n}\right)$, we have $\Reg_n = O_p\left((\log n)^{d+\max(d/2,4)}/n\right)$.
\end{corollary}

\subsection{Characterization of Gaussian Experiments}
\label{subsec:gaussian-comparison}

This section compares Gaussian experiments, i.e., designs $\eta$ that generate a normal observation with design-dependent noise covariance.
We emphasize three points.
First, Loewner improvements in sampling noise deliver prior- and objective-robust dominance (Blackwell dominance) and hence leave no scope for EB.
Second, the only benchmark in which all feasible Gaussian designs are completely ordered---so the prior never affects the design choice---is the genuinely univariate case $d=1$.
Third, even when $d>1$, a tractable and widely used class of problems arises when welfare depends on $\theta$ only through a scalar linear index $\omega=\alpha'\theta$, and the prior is Gaussian.
In this Gaussian--Gaussian environment, every scalar-index decision problem ranks designs by the posterior variance of $\omega$; hence the optimal design can be characterized by minimizing that posterior variance.
Unlike the $d=1$ benchmark, this criterion depends on the true prior covariance and therefore is exactly where EB can matter.

Consider a Gaussian experiment:
\[
Y_\eta = \theta_{n+1} + \varepsilon_\eta,
\qquad 
\varepsilon_\eta \sim \mathcal N\!\big(0,\Sigma(\eta)\big),
\qquad 
\Sigma(\eta) \succ 0 .
\]

Theorem~\ref{thm:blackwell-gaussian} is an immediate consequence of Blackwell’s
informative theorem \citep{Blackwell1951}. In Gaussian experiments, the corresponding comparison via Loewner ordering of covariance matrices is established in \cite{HansenTorgersen1974}.

\begin{theorem}[Loewner order implies Blackwell dominance]
\label{thm:blackwell-gaussian}
If $\Sigma(\eta_1)\preceq \Sigma(\eta_2)$ (Loewner order), then experiment $\eta_1$ Blackwell-dominates $\eta_2$:
for any prior $Q$ and any welfare function $W$ with well-defined value, $U_Q(\eta_1)\ge U_Q(\eta_2)$.
\end{theorem}

Theorem~\ref{thm:blackwell-gaussian} isolates a robust notion of ``better design.''
If one design reduces sampling noise in the Loewner sense, it is uniformly better for \emph{every} objective and \emph{every} prior, so there is no role for EB in choosing between them.
When $\Sigma(\eta_1)$ and $\Sigma(\eta_2)$ are not comparable in Loewner order---the generic case in multivariate design problems, where increasing precision in one direction typically requires sacrificing precision in another---Theorem~\ref{thm:blackwell-gaussian} is silent: neither design uniformly dominates the other.
The remainder of this subsection discusses two useful simplifications and their implications for when the prior can affect the optimal design.

\runin{Univariate $\theta$}
When $d=1$, Loewner order reduces to the usual scalar order on variances: $\Sigma(\eta_1)\le \Sigma(\eta_2)$.
Hence, in scalar Gaussian experiments, designs are completely Blackwell-ordered by their sampling variances, so the most informative feasible experiment is optimal regardless of the prior.
This is the only setting in which the prior is irrelevant for \emph{all} design comparisons: there is no cross-direction trade-off for the prior to resolve.

\runin{Scalar payoff-relevant state}
The univariate benchmark is special because feasible Gaussian designs are completely ordered.
With $d>1$, that complete order typically fails.
A central intermediate case is when payoffs depend on $\theta$ only through a scalar index $\omega=\alpha'\theta$ (e.g., an ATE or a portfolio value), and the prior is Gaussian.
In the Gaussian--Gaussian model, the induced posterior for $\omega$ is normal with a deterministic variance, which implies that every scalar-index decision problem ranks designs by that posterior variance.
Importantly, unlike the $d=1$ benchmark, the posterior variance depends on the true prior covariance $V$, so EB can change the optimal design through its estimate of $V$ even though the form of the objective function is irrelevant within this class. 

The next lemma and corollary make this reduction explicit by constructing a scalar sufficient statistic $Z_\eta$ that is Blackwell equivalent to the multivariate Gaussian signal for learning the payoff-relevant index $\omega$.
Related dynamic, continuous-time versions of Corollary \ref{cor:scalar-blackwell} appear in \cite{LiangEtAl2022}.
\begin{lemma}[Gaussian reduction to a payoff-relevant index]
\label{lem:reduce-to-scalar}
Let $Q=\mathcal N(\mu,V)$ and $Y_\eta=\theta+\varepsilon_\eta$ with $\varepsilon_\eta\sim\mathcal N(0,\Sigma(\eta))$ independent of $\theta$.
Fix $\alpha\in\R^d$ and define $\omega=\alpha'\theta$.
Then the posterior distribution of $\omega$ given $Y_\eta$ is Gaussian with deterministic variance given $\eta$:
\begin{equation}
\V(\omega\mid Y_\eta)
=
\alpha'\big(V^{-1}+\Sigma(\eta)^{-1}\big)^{-1}\alpha.
\label{eq:postvar_scalar_index}
\end{equation}
Moreover, there exists a scalar statistic $Z_\eta$ and a noise variance $\tau^2(\eta)\ge 0$ such that,
\[
Z_\eta = \omega + \xi_\eta,
\qquad 
\xi_\eta\sim \mathcal N\!\big(0,\tau^2(\eta)\big)
\quad\text{independent of }\omega,
\]
and the posterior law of $\omega$ given $Y_\eta$ coincides with the posterior law given $Z_\eta$.
\end{lemma}

\begin{corollary}[Objective-independent design criterion]
\label{cor:scalar-blackwell}
Maintain the assumptions of Lemma~\ref{lem:reduce-to-scalar} and fix the Gaussian prior $Q=\mathcal N(\mu,V)$.
For any decision problem whose welfare depends on $\theta$ only through $\omega=\alpha'\theta$, the optimal design $\eta^*$ should minimize $\V(\omega\mid Y_\eta).$
\end{corollary}

Corollary~\ref{cor:scalar-blackwell} follows because scalar normal experiments are completely Blackwell-ordered by precision, and $\tau^2(\eta)$ in Lemma~\ref{lem:reduce-to-scalar} is a monotonic transformation of $\V(\omega\mid Y_{\eta})$.
The key message is that, in scalar-index problems, the objective does not affect the optimal design choice, but the prior typically does through $V$.
In contrast, when welfare depends on multiple components of $\theta$ (as in Objective~I with $L=I_G$), no such complete ordering exists and the objective and prior jointly determine how the designer trades off precision across directions.

\begin{remark}[Beyond Gaussian priors and linear indices]
Lemma~\ref{lem:reduce-to-scalar} uses normal--normal conjugacy and a linear payoff index.
With non-Gaussian priors on $\theta$ or with nonlinear payoff-relevant states $s=f(\theta)$, the posterior of the payoff-relevant object is generally not Gaussian and typically is not summarized by a deterministic posterior variance.
In such settings, design rankings can depend on richer features of the prior and on the objective.
\end{remark}

\section{Empirical Applications}
\label{sec:emp.app}

This section provides two empirical applications with complementary lessons about how empirical Bayes design works in practice.
The first uses oncology drug trials drawn from a public online database of clinical research studies, where the main challenge is constructing a usable prior-study dataset: aggregating heterogeneous studies requires choices about poolability (which estimands are comparable enough to combine), harmonizing subgroup definitions, and handling multiple eligible estimates per study.
The second is a multisite RCT setting in economics, where treating sites as ``studies'' makes the evidence base more naturally comparable and easier to aggregate; the focus then shifts to design choice itself, illustrating how different objectives can imply different optimal treatment allocations even under the same EB prior.

\subsection{Oncology Drug Trials}
\label{subsec:drug_empirical}

This section presents our first empirical application based on oncology trial results from ClinicalTrials.gov, an online database of clinical research studies and information about their results. We consider a researcher who is designing a new immunotherapy drug trial and wants to use published results from similar drugs to inform the design. We focus on PD-L1 (programmed death-ligand 1) expression as the key subgroup variable. PD-L1 is a biomarker that predicts response to cancer immunotherapy drugs, with higher expression usually associated with greater benefit. Because treatment effects differ across PD-L1 levels, trial designers routinely pre-specify subgroup analyses---and sometimes stratify randomization---by PD-L1 level (see, e.g., \citealp{ImEtAl2024, MilesEtAl2021}). We use PD-L1 subgroup effect estimates from prior trials to construct an empirical Bayes prior that informs how to allocate experimental effort across PD-L1 subgroups in the new trial.

\runin{Data construction and poolability}
We query ClinicalTrials.gov API using a broad oncology therapy basket, keep studies with posted results, and extract subgroup analyses that contain both an estimate and standard error. This yields 821 PD-L1 subgroup estimates from 87 studies. Each estimate corresponds to a treatment effect for a particular outcome measure. We classify outcomes into four broad families: overall survival (OS), progression-free survival (PFS), tumor response, and other efficacy outcomes (see Table~\ref{tab:drug_pool_counts}). As a basic poolability requirement, we restrict attention to estimates within the same outcome family.

\begin{table}[t!]
\caption{Drug-trial dataset used for prior construction}
\begin{center}
\begin{tabular}{lcc}
\toprule
Outcome family & Distinct studies & Subgroup estimates \\
\midrule
Overall survival (OS) & 71 & 202 \\
Progression-free survival (PFS) & 62 & 181 \\
Tumor response & 48 & 117 \\
Other efficacy outcomes & 39 & 321 \\
\midrule
Total & 87 & 821 \\
\bottomrule
\end{tabular}
\end{center}
\label{tab:drug_pool_counts}
{\footnotesize {\em Notes}: Counts are from the study pool with PD-L1 subgroup estimates.}
\end{table}

Three practical issues arise when aggregating prior studies.
First, subgroup definitions are not standardized across studies. PD-L1 can be reported using different scoring systems: positive versus negative, high versus low, combined positive score (CPS) thresholds, tumor proportion score thresholds, or tumor-cell and immune-cell staining classes. We therefore harmonize reported labels into two sides, \emph{high-expression} and \emph{low-expression}, with one guiding idea: preserve a monotone ranking of PD-L1 expression so that ``more expression'' maps to high-expression and ``less or no expression'' maps to low-expression, regardless of the original language. We provide concrete mappings and examples in the Online Appendix.

Second, studies often report multiple subgroup effects within the same trial record. This happens for three recurring reasons: multiple PD-L1 cutoff classes on the same side (for example, CPS $\ge 1$ and CPS $\ge 10$), multiple treatment arms, and repeated estimates under different follow-up periods or sample definitions. These within-study estimates are not independent because they are computed from overlapping participants; keeping all of them without taking into account their correlation would overweight those trials in prior estimation. We therefore enforce one estimate per study. If a study reports multiple cutoff-defined classes on one side, we keep the class closest to a cross-study reference cutoff within the same scoring system. In other cases, we randomly select one estimate. Full implementation details are reported in the Online Appendix.

Third, metric mixing is substantial across outcome families. Outside survival outcomes, reported effects combine incompatible scales (for example, odds ratios, least-square mean differences, and percentage-point differences), so direct pooling is not meaningful. We therefore narrow to the survival family (OS and PFS), where hazard ratios have a common interpretation across studies. A hazard ratio below one means treatment lowers the hazard rate and is therefore favorable for survival outcomes. 
We map hazard ratio to treatment effect $\mathrm{TE}=1-\mathrm{HR}$, so $\mathrm{TE}=0$ is the no-effect benchmark and larger $\mathrm{TE}$ implies better survival.

In total, the prior-study dataset for the OS outcome family includes 71 estimates from 64 studies (7 studies report both high and low, 52 report high only, and 5 report low only). We then represent the historical evidence as
\[
\hat\psi_i \mid \theta_i \ \stackrel{approx}{\sim}\ \mathcal N(R_i\theta_i,\Sigma_i),\qquad
\theta_i=\begin{bmatrix}\theta_{i,\text{high}}\\ \theta_{i,\text{low}}\end{bmatrix},
\]
with $R_i=I_2$ when both sides are reported and $R_i=[1\ 0]$ or $[0\ 1]$ when only one side is reported. We focus on the OS family; details and results for the PFS family are provided in the Online Appendix.

\runin{Results}
We estimate four EB priors: Gaussian-joint, NPMLE-joint, Gaussian-independent, and NPMLE-independent. ``Joint'' means estimating a bivariate prior on $(\theta_{\text{high}},\theta_{\text{low}})$ that allows dependence between the two subgroup effects. ``Independent'' means fitting each subgroup marginal prior separately and taking their product, which imposes cross-subgroup independence.
We report joint-prior results here; independent-prior results are reported in the Online Appendix.

Figure~\ref{fig:drug_prior_and_design} summarizes the estimated OS prior marginals under joint Gaussian EB and joint NPMLE EB. Both estimators imply that subgroup effects are centered above zero, with larger average effects in the high-expression subgroup than in the low-expression subgroup, consistent with the clinical motivation for PD-L1 stratification that higher expression predicts greater benefit from immunotherapy.
Quantitatively, under both EB approaches the estimated prior mean is about twice as large in the high-expression subgroup as in the low-expression subgroup, while the low-expression subgroup has the higher prior variance.

\begin{figure}[t!]
\caption{OS estimated prior marginals: Gaussian and NPMLE}
\begin{center}
\begin{tabular}{cc}
\includegraphics[width=0.42\textwidth]{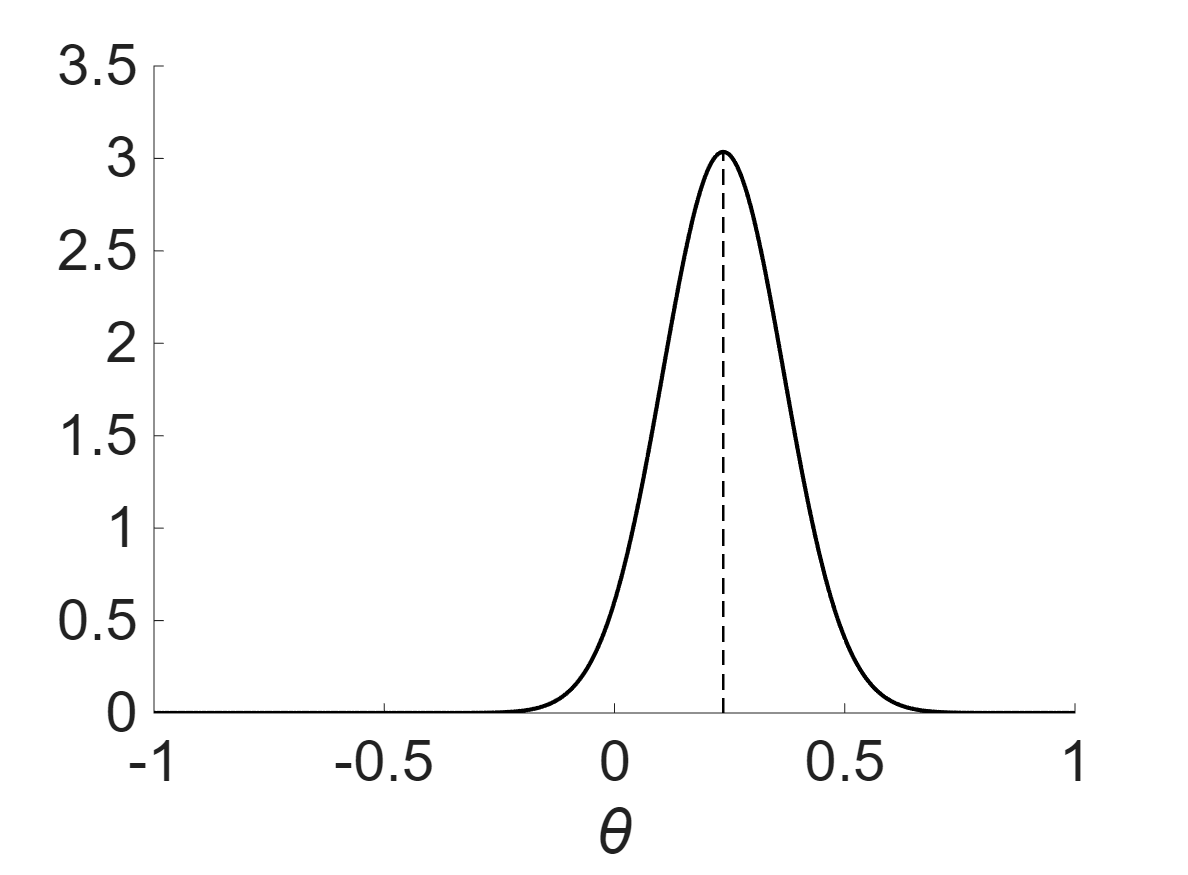} &
\includegraphics[width=0.42\textwidth]{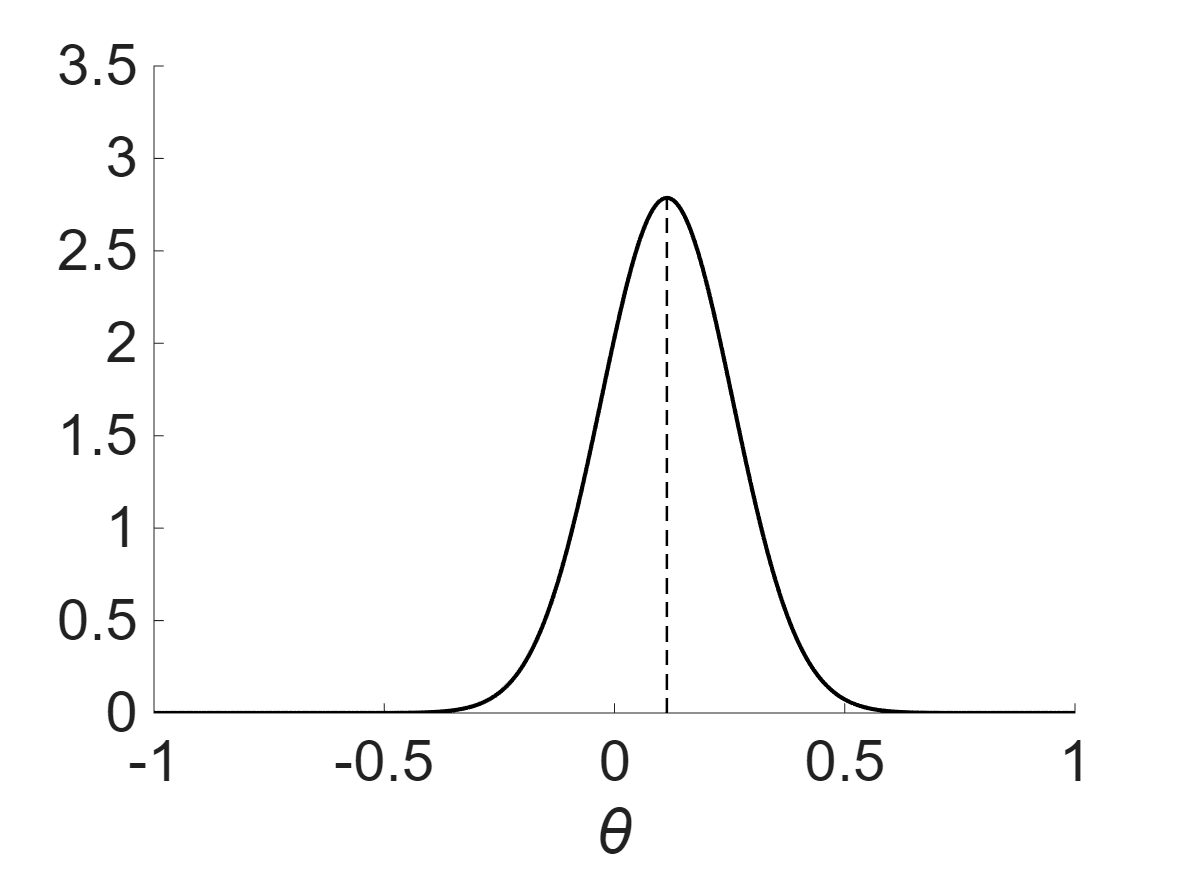} \\
\includegraphics[width=0.42\textwidth]{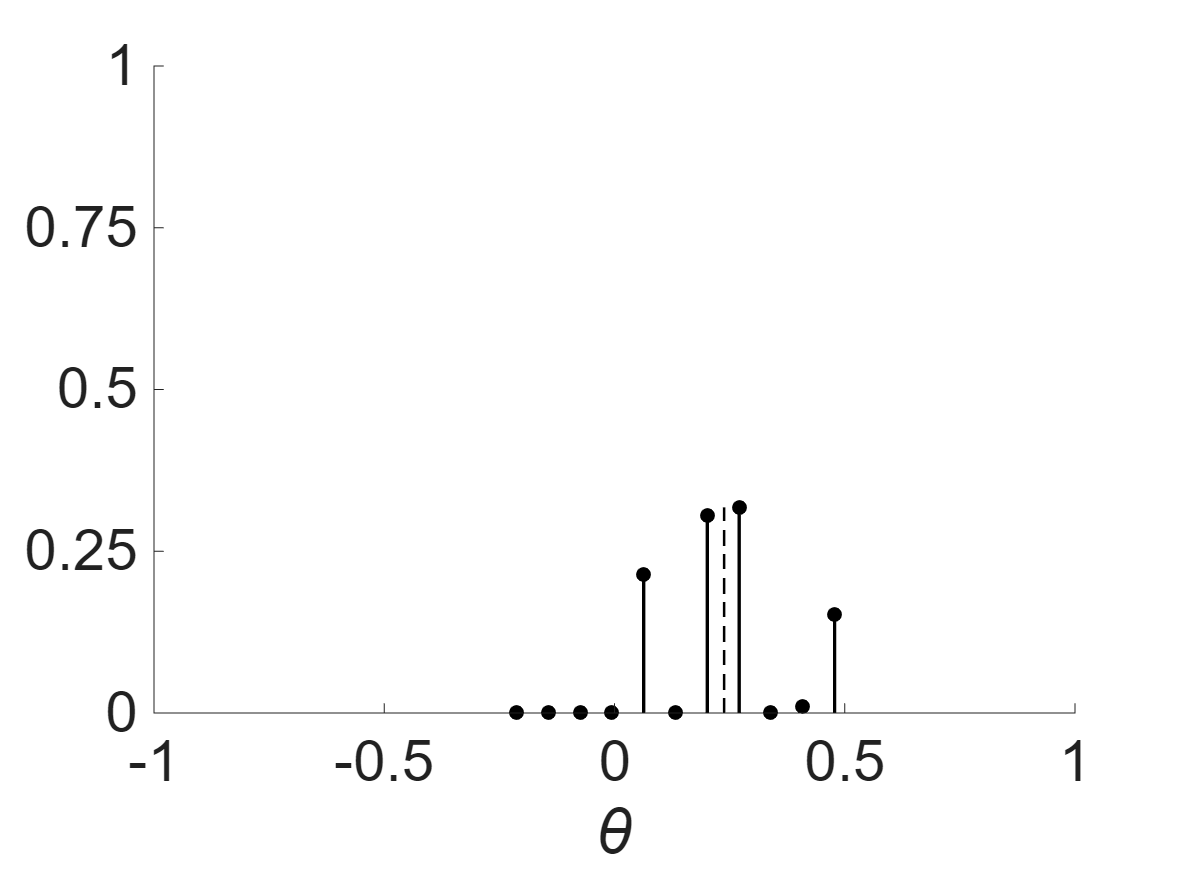} &
\includegraphics[width=0.42\textwidth]{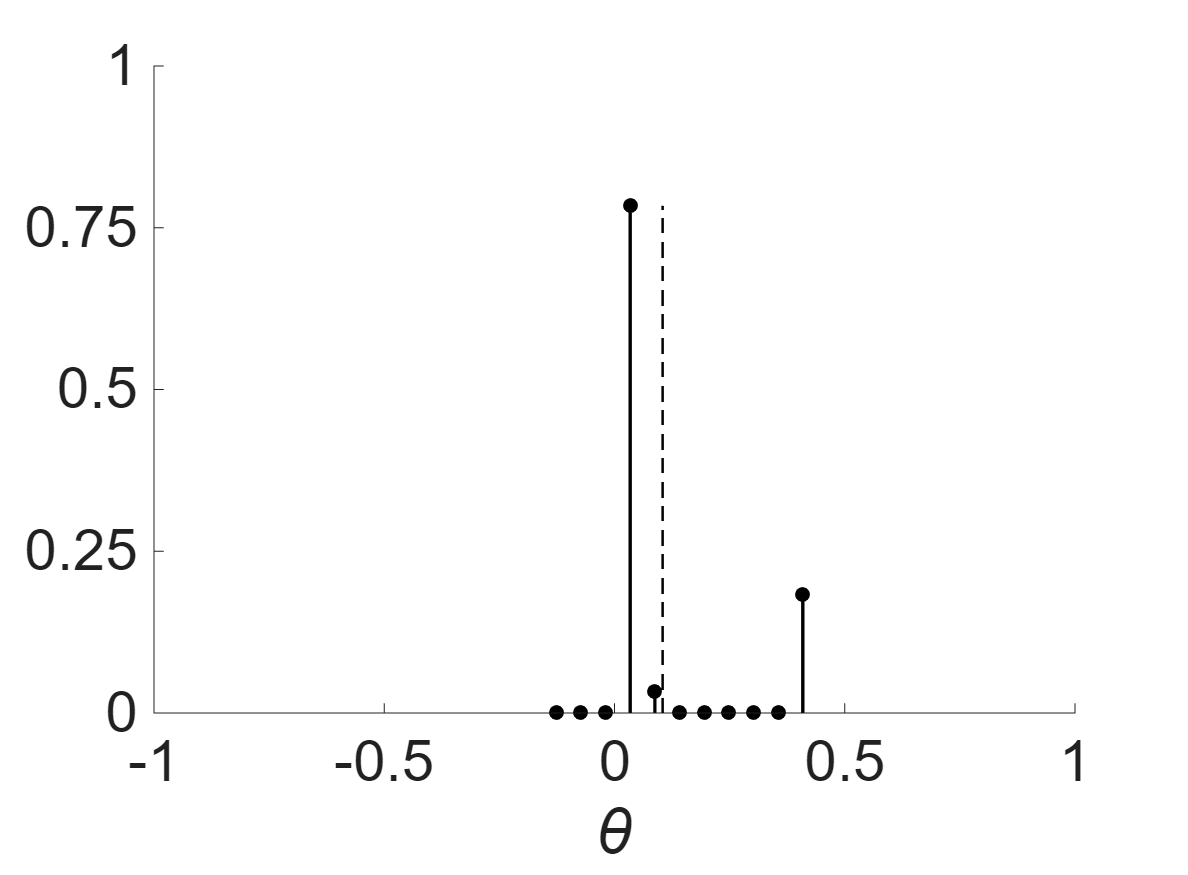}
\end{tabular}
\end{center}
\label{fig:drug_prior_and_design}
{\footnotesize {\em Notes}: Top row: Gaussian-joint prior marginals for PD-L1 high-expression and low-expression effects. Bottom row: NPMLE-joint prior marginals for the same two effects. Estimated prior means (high-expression, low-expression): Gaussian-joint EB $(0.236,0.114)$ and NPMLE-joint EB $(0.238,0.105)$; estimated prior variances (high-expression, low-expression): Gaussian-joint EB $(0.017,0.020)$ and NPMLE-joint EB $(0.016,0.021)$.}
\end{figure}

For the propensity-score design of the new drug trial, we focus on Objective~I, which minimizes posterior quadratic loss for $(\theta_{\text{new,high}},\theta_{\text{new,low}})$. For our numerical illustration, we set subgroup shares to $\pi=(0.5,0.5)$ and arm-specific outcome variances to $\sigma_{1,\text{high}}^2=\sigma_{0,\text{high}}^2=\sigma_{1,\text{low}}^2=\sigma_{0,\text{low}}^2=1$. We impose overlap bounds $e_g\in[0.05,0.95]$ for $g\in\{\text{high},\text{low}\}$, treatment costs $c_{\text{high}}=c_{\text{low}}=2$, and a budget constraint $\sum_g \pi_g c_g e_g\le B$ with $B=0.5$.
As a comparison, the no-information benchmark uses a diffuse Gaussian prior (formally $V^{-1}\to 0$), so the design is driven only by sampling precision and constraints. Under the symmetric inputs above and equal costs, this benchmark yields the balanced feasible design $(e_{\text{high}},e_{\text{low}})=(0.25,0.25)$.

Table~\ref{tab:drug_os_design_2x2} reports the optimal designs under the estimated EB priors. Both joint priors shift assignment toward the low-expression subgroup: $(0.229,0.271)$ under Gaussian EB and $(0.231,0.269)$ under NPMLE EB. 
Gaussian and NPMLE can imply different optimal designs because Gaussian imposes a single elliptical prior shape, while NPMLE allows discrete mass and non-Gaussian tail behavior; when those shape differences matter for posterior-risk gradients, the implied designs can diverge more. In OS, the two joint priors are close in means and second moments, so the resulting designs are also close.

Table~\ref{tab:drug_os_design_2x2} reports the Objective~I optimal designs under the estimated EB priors. The two joint-prior designs are similar: $(0.229,0.271)$ under Gaussian-joint EB and $(0.231,0.269)$ under NPMLE-joint EB. Both shift assignment toward the low-expression subgroup relative to the no-information benchmark because, under both priors, the low-expression effect has higher prior variance than the high-expression effect; the posterior-risk objective therefore places more value on reducing uncertainty about $\theta_{\text{new,low}}$. OS independent-subgroup robustness and all PFS results are reported in the Online Appendix.

\begin{table}[t!]
\caption{Objective-I optimal design for OS (joint priors)}
\begin{center}
\begin{tabular}{lcc}
\toprule
Prior & $e_{\text{high}}$ & $e_{\text{low}}$ \\
\midrule
Gaussian-joint EB & 0.229 & 0.271 \\
NPMLE-joint EB & 0.231 & 0.269 \\
No-information & 0.250 & 0.250 \\
\bottomrule
\end{tabular}
\end{center}
\label{tab:drug_os_design_2x2}
\end{table}

\subsection{Project STAR}
\label{subsec:star_empirical}
Our second application uses Project STAR (Student to Teacher Achievement Ratio), the Tennessee class-size RCT, to illustrate our EB design approach in a setting with rich historical evidence. We treat each STAR school as one prior ``study'' because random assignment to class types was carried out \emph{within} schools and the treatment contrast and outcome measurement follow a common protocol across sites \citep{Krueger1999,KruegerWhitmore2001}, so each school yields an internally valid but noisy vector of stratum-level treatment-effect estimates. Cross-school differences in student composition and implementation generate cross-site treatment-effect heterogeneity, and the resulting external-validity problem motivates borrowing strength across sites when forecasting effects in a new setting \citep{Menzel2025,AdjahoChristensen2025}. We use empirical Bayes to aggregate the collection of school-level estimates to learn the cross-school distribution and form a predictive prior for a future, STAR-like experiment. We then use this predictive prior to guide the design of the next experiment by choosing stratum-specific treatment propensities: in the new experiment, each student is randomized with a known assignment probability that can vary by stratum, subject to overlap and budget constraints.

\runin{Data}
Project STAR randomly assigned students (and teachers) within each participating school to one of three class types: small classes, regular-size classes, and regular-size classes with a full-time teacher aide. Let $D=1$ denote assignment to a small class and $D=0$ denote assignment to either a regular class or a regular+aide class. We pool the latter two arms because achievement outcomes are similar across them in STAR \citep{Krueger1999,KruegerWhitmore2001}. We focus on end-of-kindergarten reading scores as the outcome variable. For the new experiment, we define four randomization strata using the $2\times 2$ cross of race and free-lunch eligibility (a proxy for low family income):
\[
\begin{aligned}
S_1&=(\text{non-white},\ \text{free lunch}),\quad
S_2=(\text{non-white},\ \text{non-free lunch}),\\
S_3&=(\text{white},\ \text{free lunch}),\quad
S_4=(\text{white},\ \text{non-free lunch}).
\end{aligned}
\]
Among baseline covariates, race and free-lunch status are the two strongest single predictors of kindergarten reading in our sample, and the partition is directly policy-relevant for targeted class-size policies. 
We keep students with non-missing outcomes and covariates, leaving $N=5{,}772$ students in $79$ schools.
The resulting stratum shares in the STAR samples are
$(26.1\%,6.5\%,22.2\%,45.2\%)$ for $(S_1,S_2,S_3,S_4)$.

Each school is treated as one prior study $i$. For each stratum $g$, we construct the school-level stratum-specific treatment-effect estimate from the raw data as the within-school difference in mean outcomes,
\[
\hat\theta_{ig}=\bar Y_{ig}(1)-\bar Y_{ig}(0),\qquad
\widehat{\mathrm{SE}}_{ig}^2=\frac{s_{1ig}^2}{n_{1ig}}+\frac{s_{0ig}^2}{n_{0ig}},
\]
where $\bar Y_{ig}(d)$ is the average reading score among students in school $i$, stratum $g$, assigned to arm $D=d$, and $(n_{dig},s_{dig}^2)$ are the corresponding sample sizes and variances.
If a school has fewer than two observations in either arm for stratum $g$, we treat that component as missing and represent the observed components with a selector matrix $R_i$. Formally, prior information enters as $\hat\psi_i\sim \mathcal N(R_i\theta_i,\Sigma_i)$. This keeps all $79$ schools in prior estimation: $10$ schools contribute one stratum estimate, $53$ contribute two, $10$ contribute three, and $6$ contribute all four.
Figure~\ref{fig:star_data_rct} plots these school-by-stratum estimates with 95\% confidence intervals; the substantial dispersion across schools within each stratum provides direct evidence of cross-site treatment-effect heterogeneity, motivating EB aggregation across schools.

\begin{figure}[t!]
\caption{School-level treatment-effect estimates by stratum}
\begin{center}
\begin{tabular}{cc}
\includegraphics[width=0.42\textwidth]{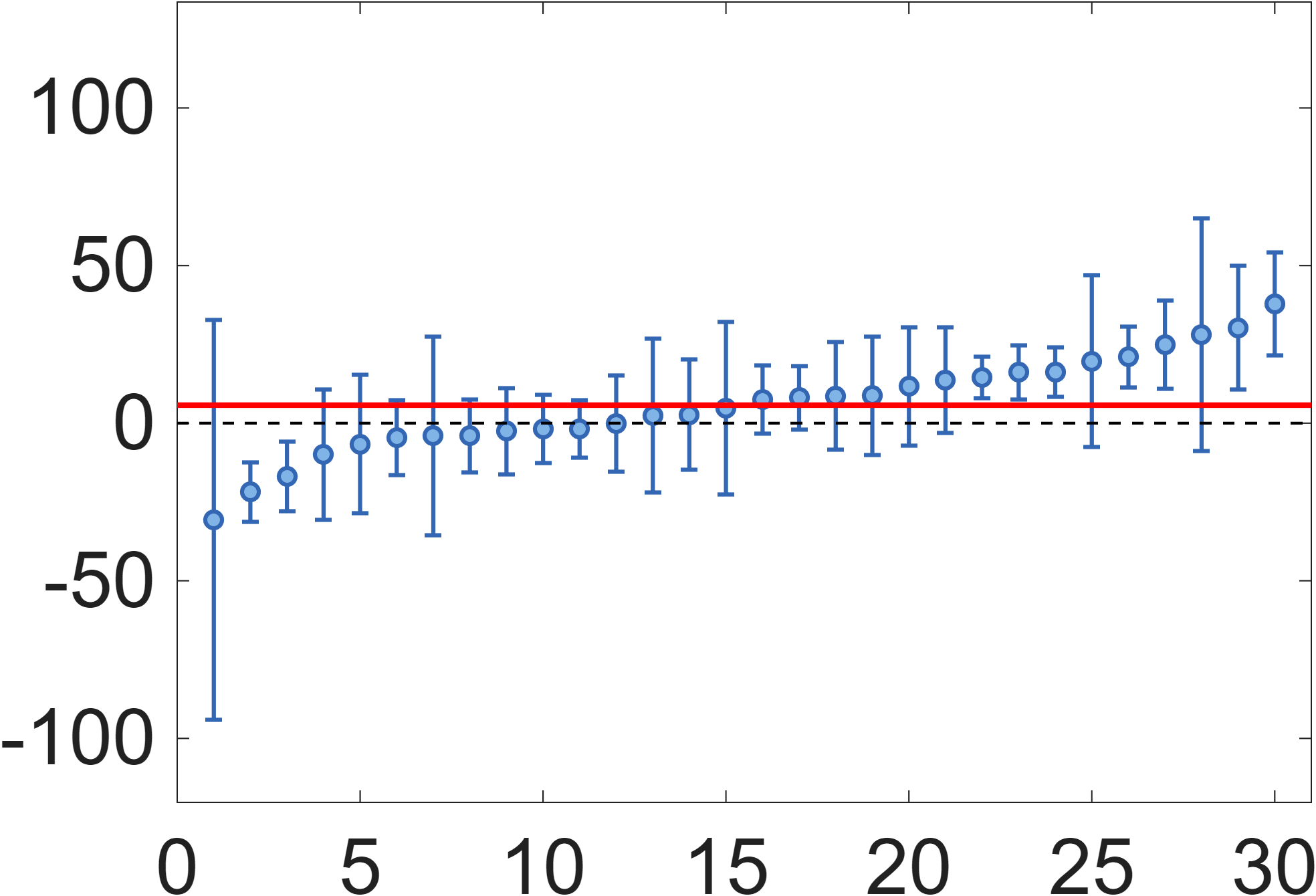} &
\includegraphics[width=0.42\textwidth]{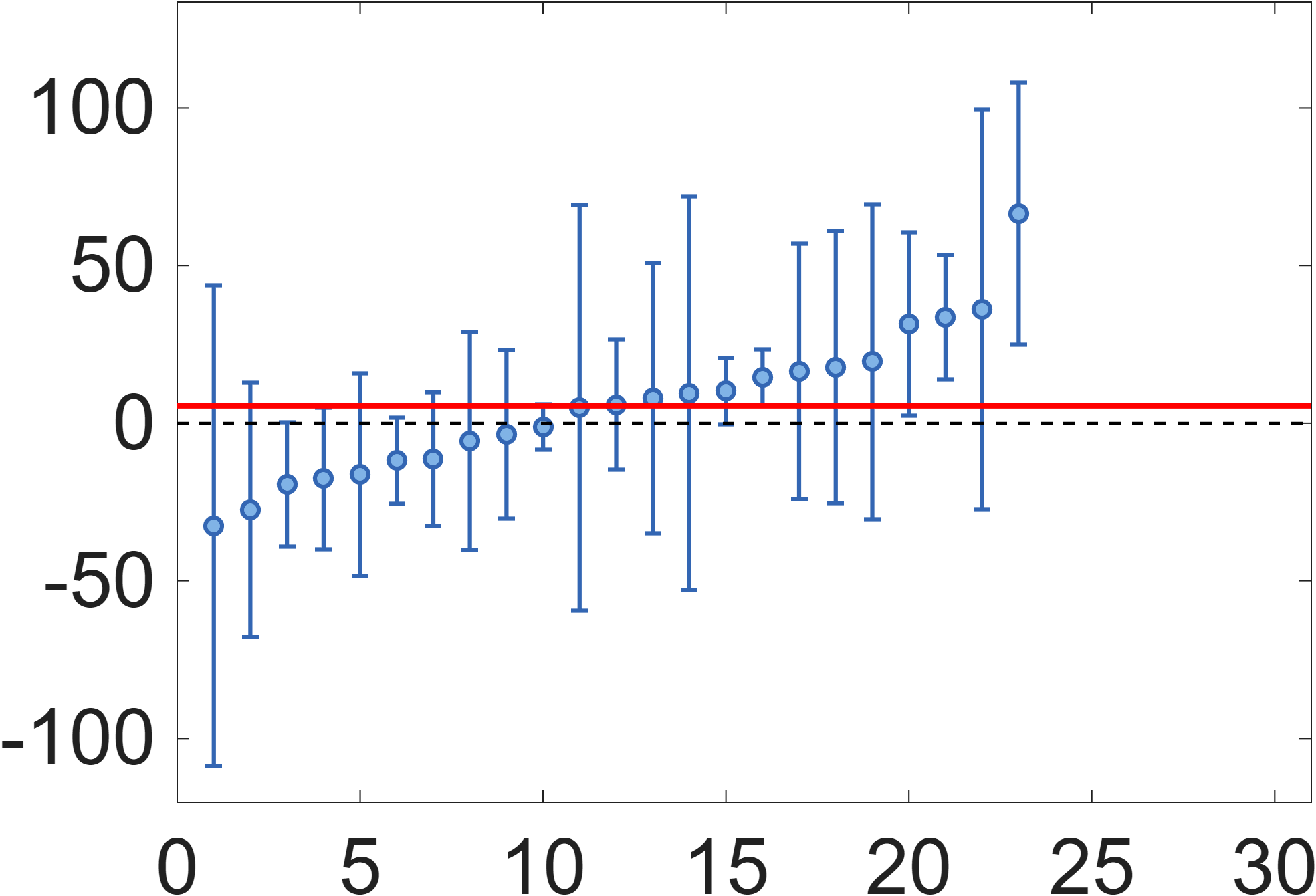} \\
\includegraphics[width=0.42\textwidth]{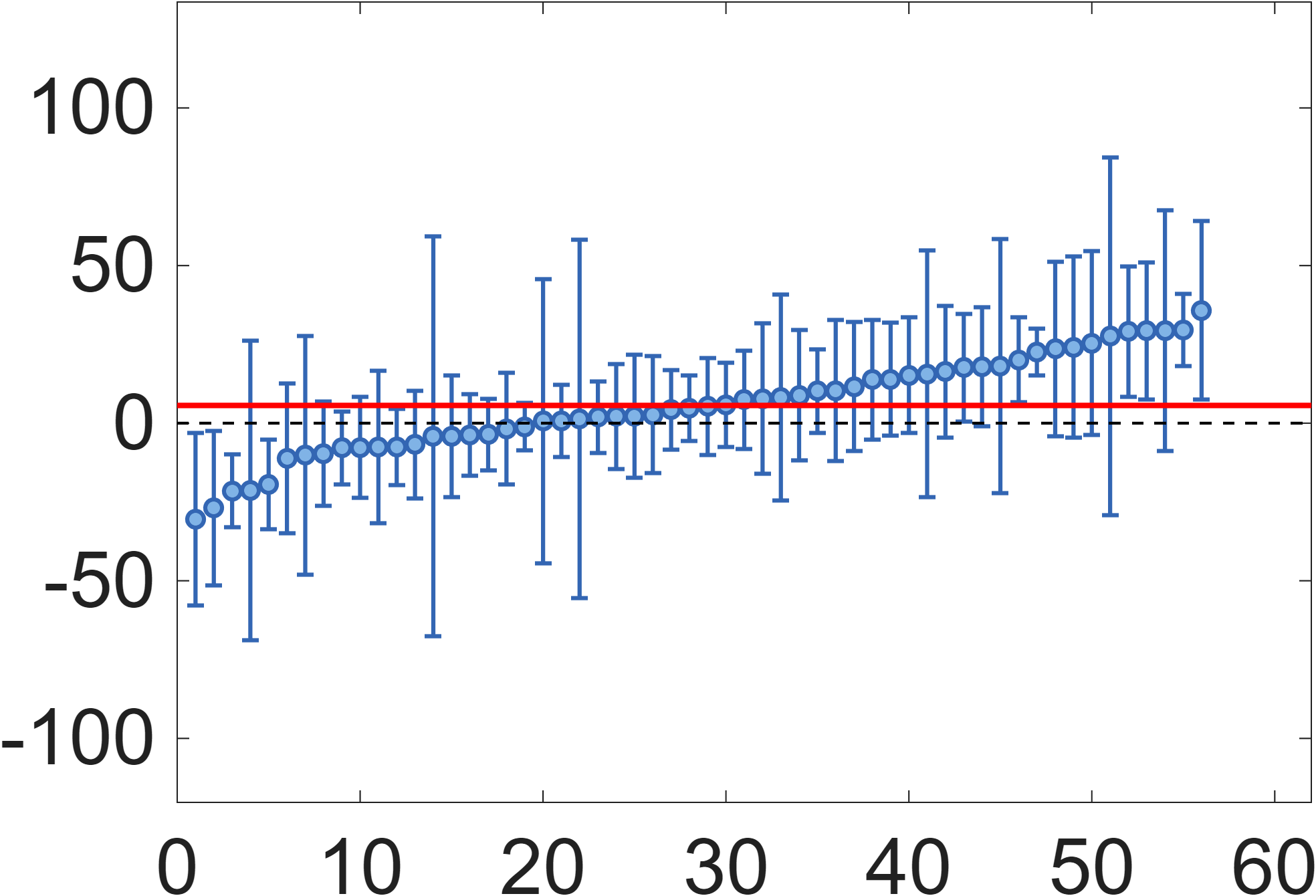} &
\includegraphics[width=0.42\textwidth]{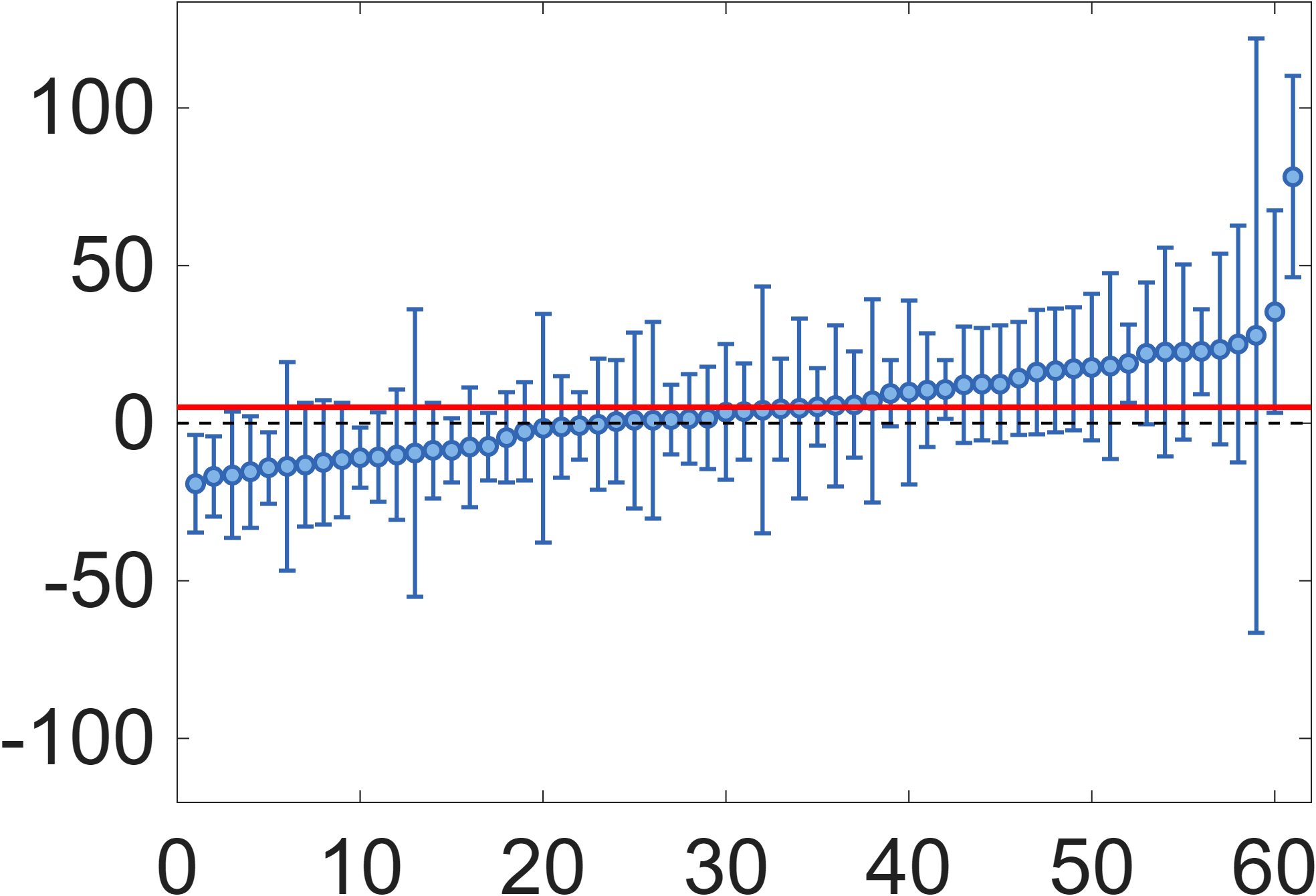}
\end{tabular}
\end{center}
\label{fig:star_data_rct}
{\footnotesize {\em Notes}: Top-left to bottom-right panels represent $S_1,S_2,S_3,S_4$. Each panel plots school-level stratum-specific estimates $\hat\theta_{ig}$ with 95\% confidence intervals ($\hat\theta_{ig}\pm 1.96\,\mathrm{SE}_{ig}$), sorted by $\hat\theta_{ig}$. The number of school-level estimates is $(30,23,56,61)$ for $(S_1,S_2,S_3,S_4)$.}
\end{figure}

\runin{Results}
As in the first application, we estimate four EB priors for the stratum-specific treatment-effect vector: Gaussian and NPMLE, each under independent-subgroup and joint specifications. Figure~\ref{fig:star_prior_dist} summarizes the \emph{joint} priors used in the main text. Despite the different functional forms, the two joint estimators tell a consistent story: expected gains from small classes are largest for the more disadvantaged strata $S1$ and smallest for the most advantaged stratum $S4$, and the amount of prior uncertainty is substantial---which mirrors the wide dispersion of school-level estimates shown in Figure~\ref{fig:star_data_rct}. Quantitatively, the implied prior means span roughly $3$--$6$ points across strata, while prior variances are on the order of $10^2$ and are notably larger for the disadvantaged groups (especially $S_2$), indicating that historical evidence is informative but far from decisive. The joint priors also imply positive cross-stratum dependence, so additional precision in one subgroup can spill over to others.

\begin{figure}[t!]
\caption{Estimated prior distributions and moments (joint priors)}
\begin{center}
\begin{tabular}{cc}
\includegraphics[width=0.42\textwidth]{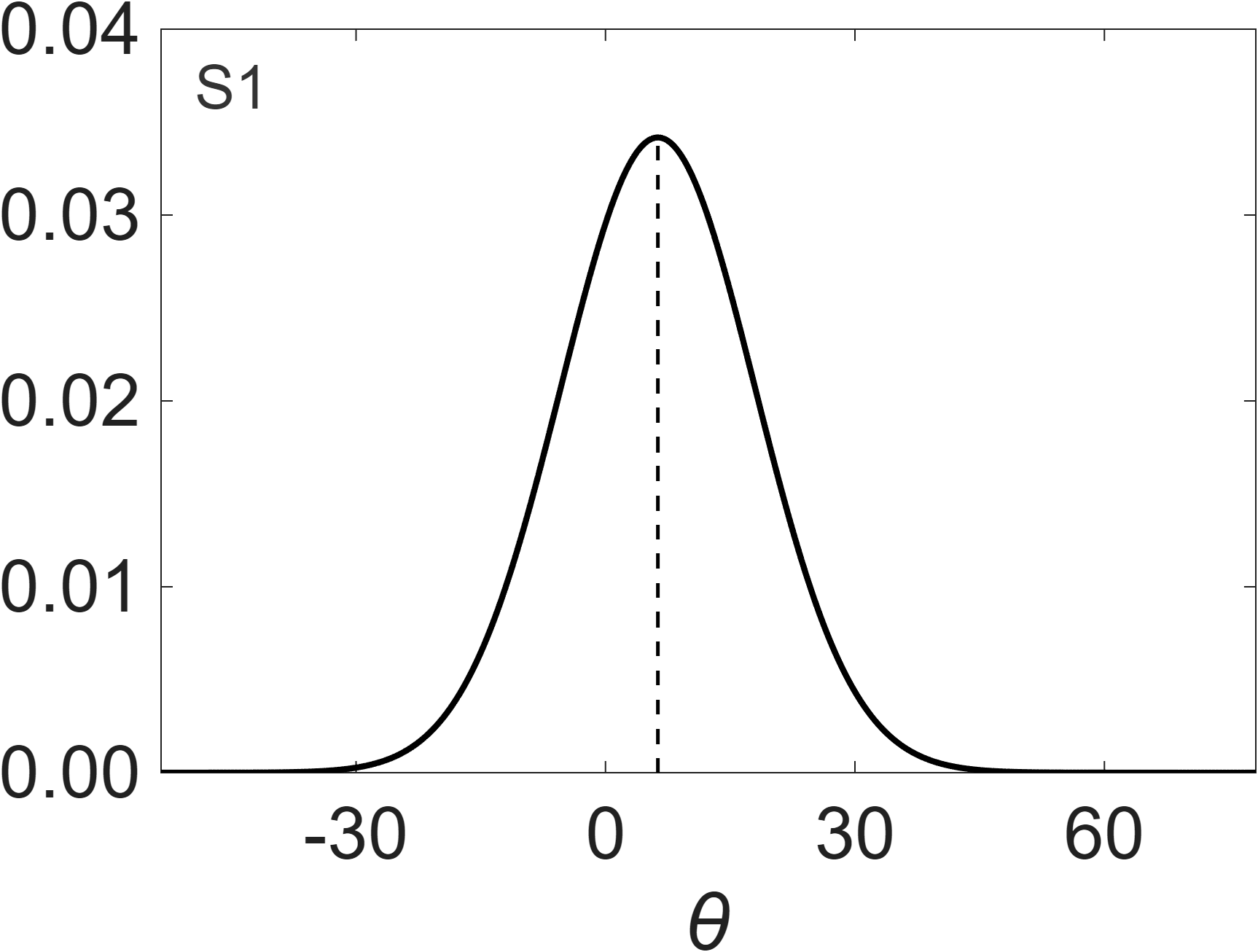} &
\includegraphics[width=0.42\textwidth]{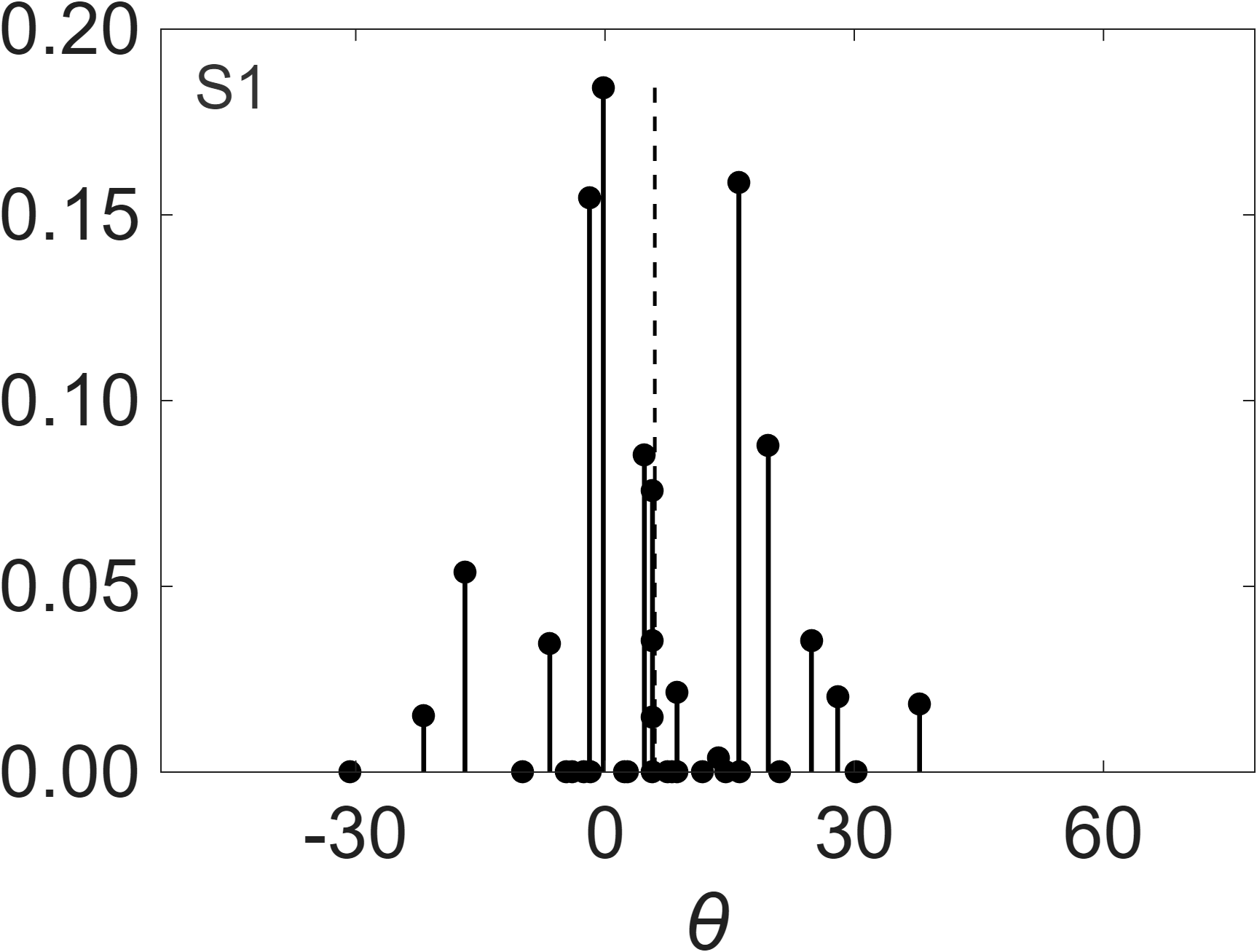} \\
\includegraphics[width=0.42\textwidth]{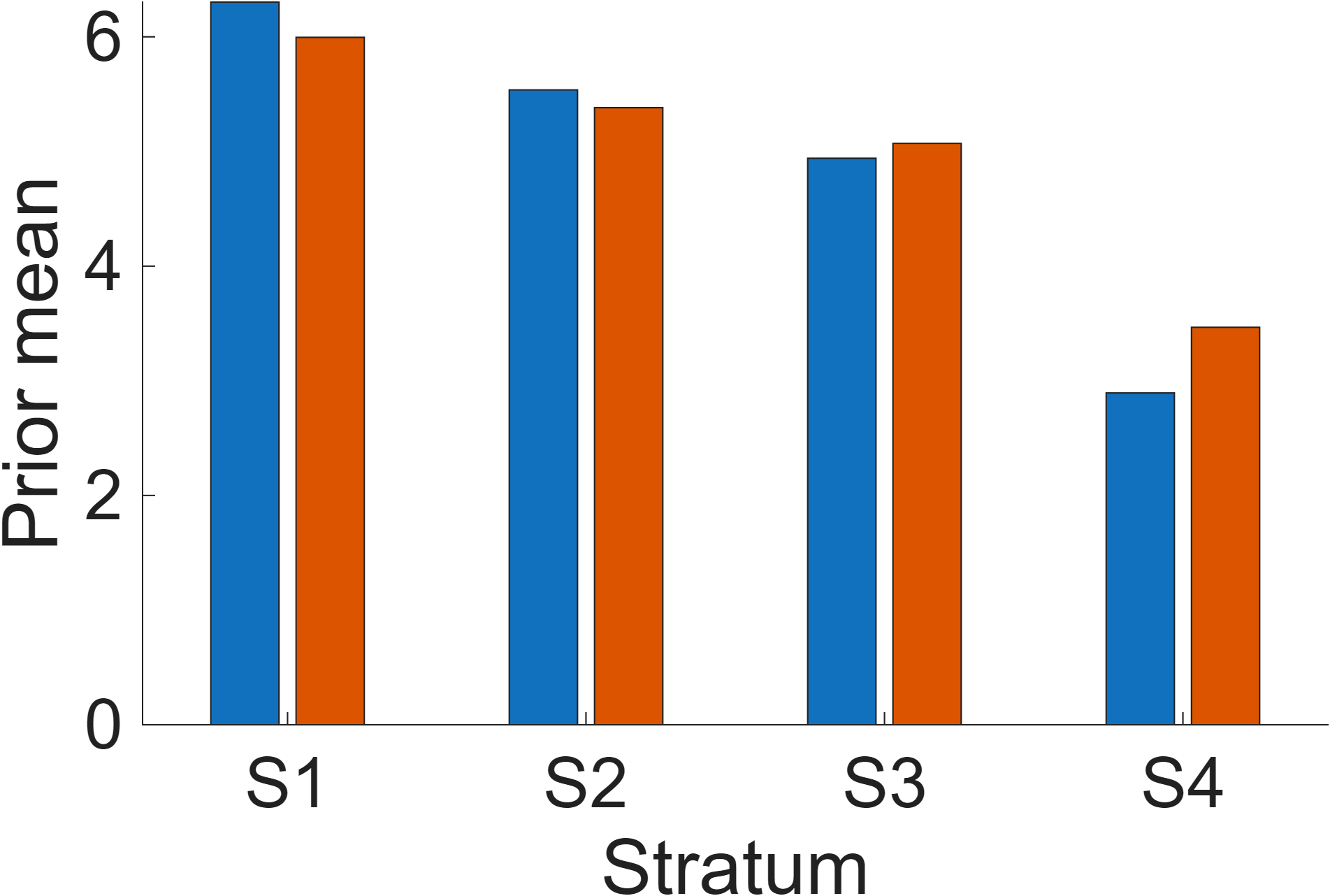} &
\includegraphics[width=0.42\textwidth]{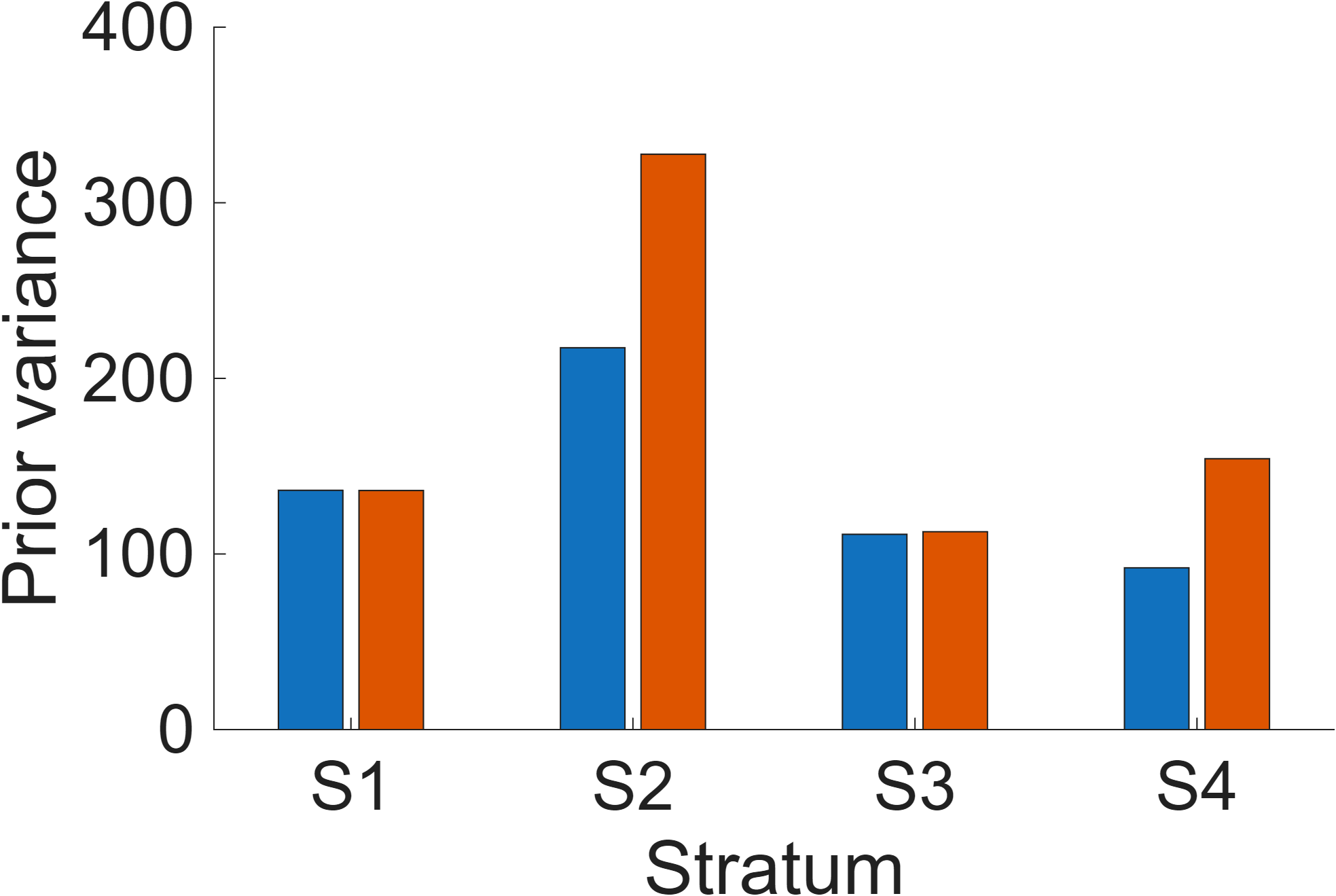}
\end{tabular}
\end{center}
\label{fig:star_prior_dist}
{\footnotesize {\em Notes}: 
Top-left is Gaussian-joint and top-right is NPMLE-joint for stratum $S_1$. Bottom-left compares prior means across strata, and bottom-right compares prior variances across strata. Blue bars represent the mean or variance under Gaussian-joint, and red bars represent the mean or variance under NPMLE-joint.
}
\end{figure}

We then solve for new-site treatment propensities $e=(e_1,\dots,e_4)$ under overlap bounds $e_g\in[0.05,0.95]$ and an average treatment budget $\sum_{g=1}^4 \pi_g e_g\le 0.30$, using the pooled STAR stratum shares $\pi$ from the data section. Because the new-site outcome variances $(\sigma_{1g}^2,\sigma_{0g}^2)$ are unknown, we use the estimated variance from pooled historical STAR data:
\[
\sigma_{1g}^2=\widehat{\Var}(Y\mid D=1,S=g),\qquad
\sigma_{0g}^2=\widehat{\Var}(Y\mid D=0,S=g).
\]
Figure~\ref{fig:star_main_alloc} reports the resulting optimal designs under Gaussian-joint and NPMLE-joint priors. The two EB estimators yield very similar allocations across all three objectives, so the qualitative design implications are robust to the choice of joint-prior family.

Under Objective~I (estimation), optimal propensities are interior. A useful benchmark is the within-stratum sampling-variance--minimizing split absent the budget constraint,
$
e_g^{\text{bal}}=\sigma_{1g}/(\sigma_{1g}+\sigma_{0g}),
$
which is close to one-half in all strata here (in particular, $e_2^{\text{bal}}= 0.48$). Consistent with this benchmark, the no-information design already places $e_2$ close to $e_2^{\text{bal}}$ (about $0.46$), and the joint-EB designs keep $e_2$ in the same range (about $0.42$--$0.45$). Thus, even though $S_2$ is the most prior-uncertain component, EB does not ``solve'' this by pushing $e_2$ further upward; instead it mainly reallocates probability mass away from $S_4$ and raises $e_1$ toward $e_2$. 
This matches Proposition~\ref{prop:quadratic_design}: posterior-risk reduction depends not only on each stratum's prior variance but also on the covariance between strata. Here, $S_2$ is a small-share stratum ($\pi_2=6.5\%$ versus $\pi_1=26.1\%$), so the total information available within $S_2$ is limited once the treated/control split is already near $e_2^{\text{bal}}$. 
At the same time, the joint EB prior links the two disadvantaged strata: $S_1$ is most correlated with $S_2$ (Corr$(S_1,S_2)=0.79$ under Gaussian-joint and $0.31$ under NPMLE-joint), so improving precision in the larger stratum $S_1$ also tightens inference about $S_2$.
This combination---small $\pi_2$ and positive cross-stratum correlation---explains why the Objective~I solution keeps $e_1$ close to $e_2$ rather than concentrating experimentation exclusively in $S_2$.

Under Objective~II (in-experiment welfare), the design becomes primarily mean-targeting. Both joint priors push the highest-mean stratum to the upper overlap bound ($e_1=0.95$) and the lowest-mean strata to the lower bound ($e_3=e_4=0.05$), illustrating the exploitation force characterized in Proposition~\ref{prop:insample_formal}; the remaining stratum is set at an interior value $0.28$ determined by the budget and overlap constraints.

Under Objective~III (post-policy choice), Proposition~\ref{prop:policy_voI_formal} implies that policy-relevant uncertainty is highest where the posterior mean is close to the adoption boundary and posterior uncertainty remains large. In our data, $S_2$ is the most policy-relevant uncertain stratum because it combines moderate prior mean with the largest prior variance. Moreover, its target share is small ($\pi_2=6.5\%$), so the design-stage sampling noise term $s^2_g(e)$ is large for a given $e$, making extra allocation especially valuable for policy learning. This is why $e_2$ rises sharply relative to no-information (from $0.149$ to $0.331$--$0.341$). By contrast, $S_4$ has a lower mean but also substantially lower uncertainty and much larger target share ($\pi_4=45.2\%$), so posterior uncertainty is reduced more quickly even without heavy reallocation; correspondingly, optimal designs reduce $e_4$ (from $0.317$ to $0.265$--$0.271$).

\begin{figure}[t!]
\caption{Optimal treatment propensities by objective under EB priors}
\begin{center}
\begin{tabular}{cc}
\includegraphics[width=0.47\textwidth]{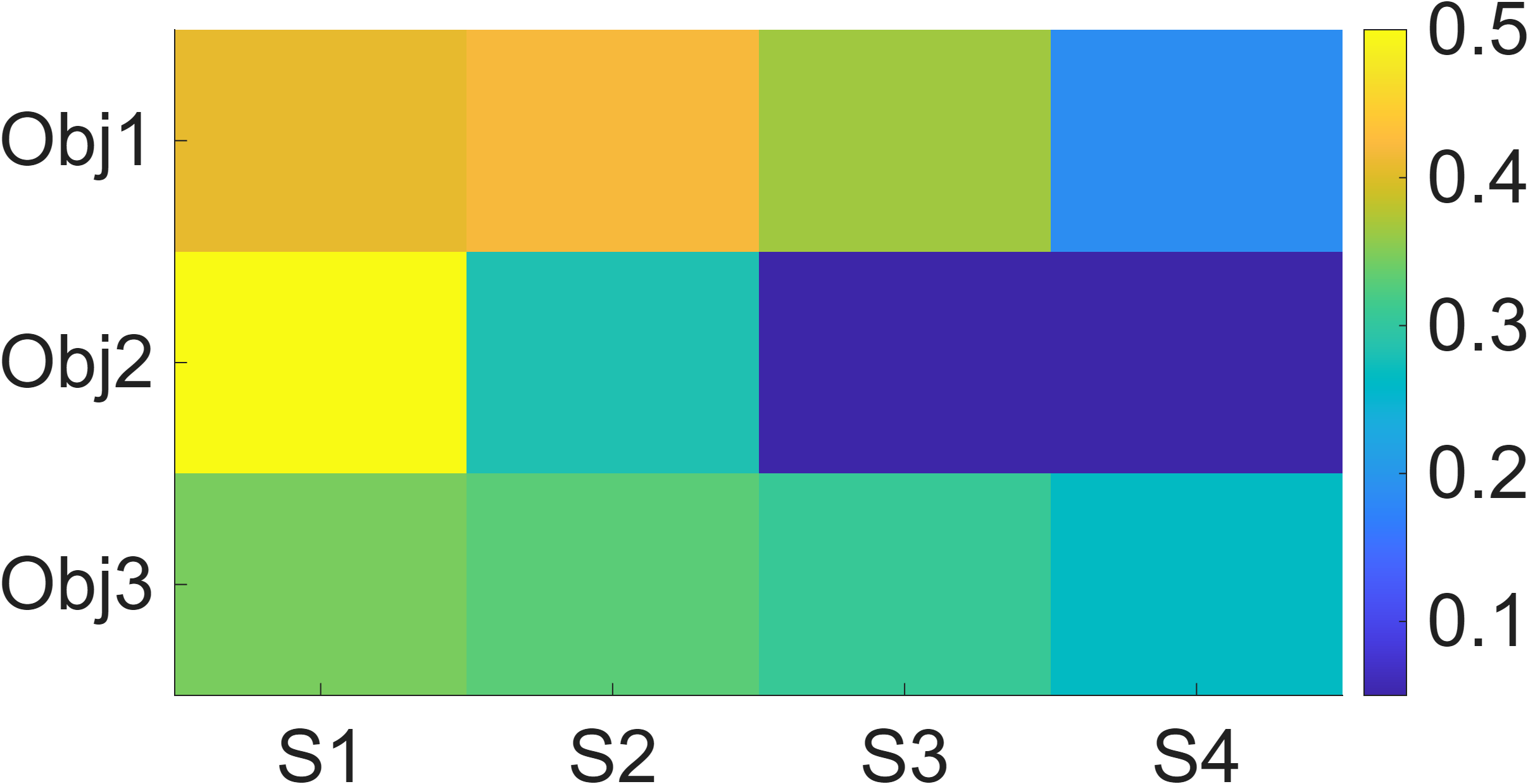} &
\includegraphics[width=0.47\textwidth]{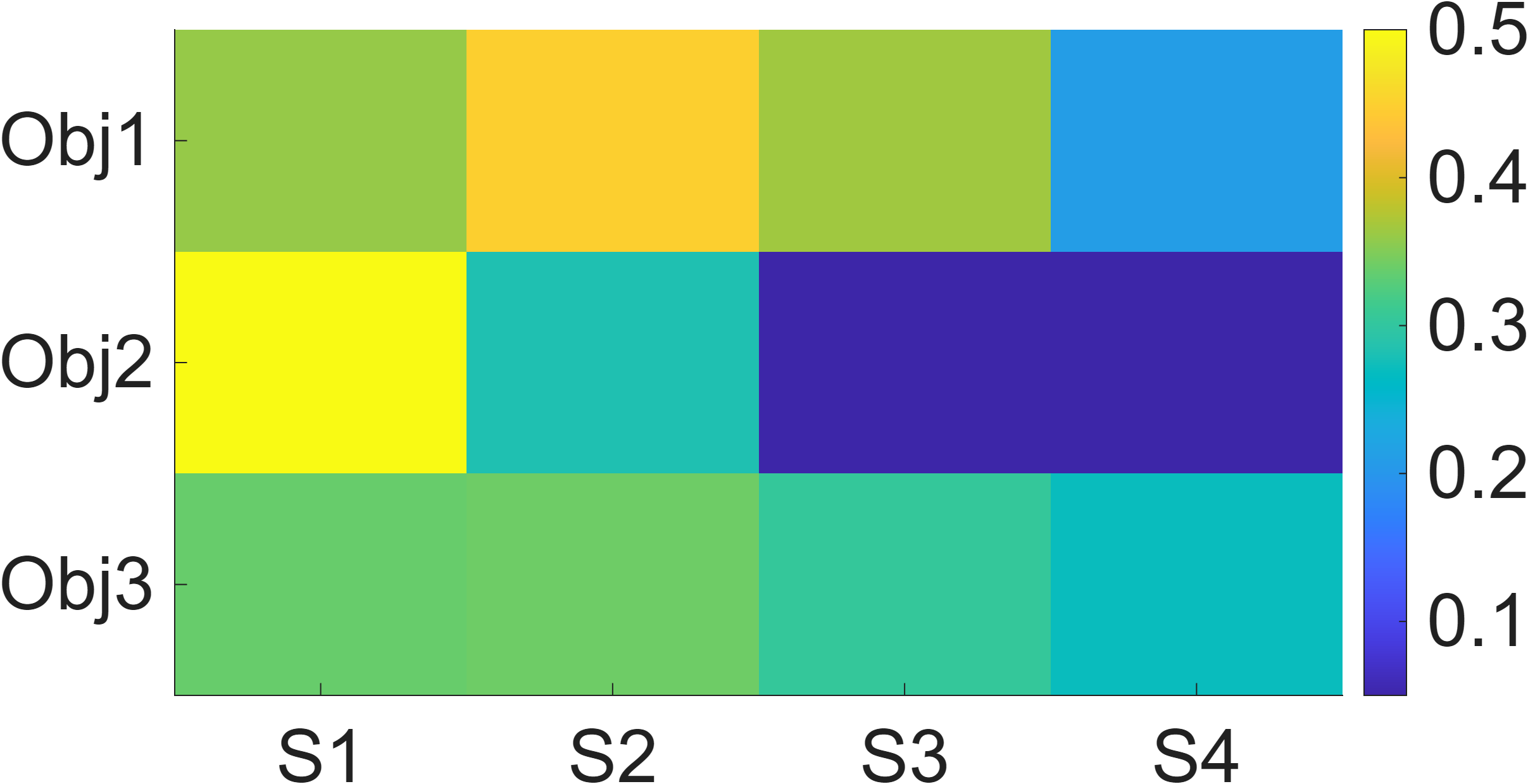}
\end{tabular}
\end{center}
\label{fig:star_main_alloc}
{\footnotesize {\em Notes}: Each panel corresponds to one joint-prior estimator. Rows are Objective I (estimation), Objective II (in-experiment welfare), and Objective III (post-policy choice); columns are strata $S_1$--$S_4$. Bars plot the optimal propensities under the stated budget and overlap constraints; the text discusses the key quantitative contrasts.}
\end{figure}

In the Online Appendix, we report additional STAR results, including independent-subgroup EB designs and comparisons of optimal designs across prior estimation schemes under each objective.

\section{Conclusion}
\label{sec:conclusion}

This paper develops an empirical Bayes approach to experimental design that leverages information from prior related studies. The framework combines (i) prior estimation from collections of related studies and (ii) a decision-theoretic criterion for choosing designs. The central message is that external evidence can be translated into a predictive prior for the next study, and that this prior can reshape optimal assignment decisions in constrained designs.

The analysis of propensity score design highlights three channels through which EB affects optimal designs: prior dispersion reshapes precision allocation under quadratic-loss estimation; prior means drive targeting under in-experiment welfare; and both means and variances matter under post-experiment policy choice, concentrating information where adoption decisions are most uncertain. The theoretical analysis links design performance to prior-estimation accuracy through a finite-sample oracle inequality and comparison bounds against a no-information benchmark that ignores external evidence; it also shows that plug-in EB designs are asymptotically oracle-optimal under weak convergence of the estimated prior, derives explicit first- and second-order regret rates for Gaussian EB and NPMLE EB, and clarifies when the prior is irrelevant in Gaussian design comparisons (only in the genuinely univariate case). The two empirical applications illustrate that our approach can deliver interpretable reallocations of experimental effort in practice.

Several extensions worth further investigation. First, while the current framework assumes exchangeability across studies, richer models could incorporate study-level covariates to allow for conditional exchangeability; one natural approach is to adopt covariate-powered empirical Bayes methods as in \cite{IgnatiadisWager2019}. Second, the framework could be extended to sequential experimental design, in which the empirical Bayes prior is updated as new experimental evidence accumulates over time. Third, complementary work could develop inference procedures for the design choice itself that properly account for uncertainty arising from empirical Bayes prior estimation, following the ideas in \cite{AndrewsChen2025}.

\bibliography{ref_design}

\clearpage
\renewcommand{\thepage}{A.\arabic{page}}
\setcounter{page}{1}

\begin{appendix}
	\markright{Online Appendix -- This Version: \today }
	
	\renewcommand{\theequation}{A.\arabic{equation}}
	\setcounter{equation}{0}	
	\renewcommand*\thetable{A-\arabic{table}}
	\setcounter{table}{0}
	\renewcommand*\thefigure{A-\arabic{figure}}
	\setcounter{figure}{0}
	\renewcommand*\thetheorem{A-\arabic{theorem}}
	\setcounter{theorem}{0}
    \renewcommand*\theproposition{A-\arabic{proposition}}
    \setcounter{proposition}{0}
	\begin{center}
		
		{\large {\bf Online Appendix: Using Prior Studies to Design Experiments: An Empirical Bayes Approach}}
		
		{\bf Zhiheng You}
	\end{center}
	
	\thispagestyle{empty} 
	
	\noindent This appendix consists of the following sections: 
	
	\begin{itemize}
		\item[A.] Details on Example \ref{emp:two_stratum_quadratic}
		\item[B.] Verification of Assumption~\ref{as:Psi_Lip} for Common Decision Problems
		\item[C.] Verification of Second-Order Assumptions in Example~\ref{emp:two_stratum_quadratic}
		\item[D.] Proofs
        \item[E.] Additional Results for the Empirical Applications
	\end{itemize}
	
	\newpage
	

\section{Details on Example \ref{emp:two_stratum_quadratic}}
\label{app:two_stratum_quadratic}

We make an extended statement for example \ref{emp:two_stratum_quadratic}.
\begin{proposition}
In the two-stratum specialization of the quadratic-loss Gaussian criterion in Proposition~\ref{prop:quadratic_design}, write $a_g:=s^2/v_g$ for the true prior pseudo-sample sizes and $\hat a_g:=s^2/\hat v_g$ for their EB plug-in estimates. Under the maintained interior condition ($|a_2-a_1|<N$ and $|\hat a_2-\hat a_1|<N$),
\begin{enumerate}[leftmargin=*]
\item \emph{No-information benchmark:} in the diffuse-prior limit ($a_1=a_2=0$), $\eta^{NI}=(N/2,N/2)$.
\item \emph{Oracle allocation and oracle gain:} the oracle allocation is $\eta^O=((N+a_2-a_1)/2,(N+a_1-a_2)/2)$ and
\[
\delta_G=U_G(\eta^O)-U_G(\eta^{NI})
=
s^2\cdot
\frac{(a_1-a_2)^2}{\big(\frac N2+a_1\big)\big(\frac N2+a_2\big)\big(N+a_1+a_2\big)}.
\]
\item \emph{EB allocation and ``EB beats NI'':} the EB allocation is $\eta^{EB}=((N+\hat a_2-\hat a_1)/2,(N+\hat a_1-\hat a_2)/2)$ and $U_G(\eta^{EB})\ge U_G(\eta^{NI})$ if and only if
\[
(\hat a_2-\hat a_1)(a_2-a_1)\ge 0
\quad\text{and}\quad
|\hat a_2-\hat a_1|\le 2|a_2-a_1|.
\]
\end{enumerate}
\end{proposition}

\begin{proof}
Under quadratic loss with $L=I_2$ and $\Lambda=I_2$, ex-ante value equals minus the Bayes risk, which equals the sum of posterior variances.
Given the prior--sampling pair
\[
\theta_g\sim \mathcal N(0,v_g),\qquad 
\hat\theta_g \mid \theta_g,N_g \sim \mathcal N\!\left(\theta_g,\frac{s^2}{N_g}\right),\qquad g=1,2,
\]
normal--normal conjugacy implies
\[
\Var(\theta_g\mid \hat\theta_g,N_g)=\left(v_g^{-1}+\frac{N_g}{s^2}\right)^{-1}=\frac{s^2}{N_g+a_g},
\qquad a_g:=\frac{s^2}{v_g}.
\]
Hence the Bayes risk for a design $\eta=(N_1,N_2)$ is
\[
\mathcal R_G(\eta)=\sum_{g=1}^2 \frac{s^2}{N_g+a_g},
\]
and maximizing $U_G(\eta)$ is equivalent to minimizing $\mathcal R_G(\eta)$.

\medskip
\noindent\emph{(i) No-information benchmark.}
In the diffuse-prior limit $v_g\to\infty$ (so $a_1=a_2=0$), the problem is
\[
\min_{0<N_1<N}\ \frac{s^2}{N_1}+\frac{s^2}{N-N_1}.
\]
The objective is strictly convex and symmetric around $N_1=N/2$, so the unique minimizer is $N_1^{NI}=N/2$ and $\eta^{NI}=(N/2,N/2)$.

\medskip
\noindent\emph{(ii) Oracle design and oracle gain.}
For general $(a_1,a_2)$, minimizing $\mathcal R_G$ is equivalent to minimizing
\[
f(N_1)=\frac{s^2}{N_1+a_1}+\frac{s^2}{N-N_1+a_2}
\quad\text{over }0<N_1<N.
\]
The first-order condition for an interior optimum is
\[
0=f'(N_1)=-\frac{s^2}{(N_1+a_1)^2}+\frac{s^2}{(N-N_1+a_2)^2}
\quad\Longleftrightarrow\quad
N_1+a_1 = N-N_1+a_2,
\]
which yields $N_1^O=(N+a_2-a_1)/2$ and $N_2^O=(N+a_1-a_2)/2$.
Evaluating the oracle risk gives
\[
\mathcal R_G(\eta^O)=\frac{4s^2}{N+a_1+a_2}.
\]
Evaluating the no-information design under the true prior gives
\[
\mathcal R_G(\eta^{NI})=\frac{s^2}{\frac N2+a_1}+\frac{s^2}{\frac N2+a_2}.
\]
Since $U_G(\eta)=-\mathcal R_G(\eta)$, the oracle gain is
\[
\delta_G
=
U_G(\eta^O)-U_G(\eta^{NI})
=
\mathcal R_G(\eta^{NI})-\mathcal R_G(\eta^O)
=
s^2\cdot
\frac{(a_1-a_2)^2}{\big(\frac N2+a_1\big)\big(\frac N2+a_2\big)\big(N+a_1+a_2\big)}.
\]

\medskip
\noindent\emph{(iii) EB design and the ``EB beats NI'' condition.}
Replacing $(a_1,a_2)$ by $(\hat a_1,\hat a_2)$ in the first-order condition yields the EB allocation formula.
To characterize when EB beats $\eta^{NI}$ under the true prior, compare risks.
Write $\Delta \hat a:=\hat a_2-\hat a_1$ and $\Delta a:=a_2-a_1$.
A direct algebraic simplification gives
\[
\mathcal R_G(\eta^{EB})-\mathcal R_G(\eta^{NI})
=
\frac{4s^2\,(N+a_1+a_2)\,\Delta \hat a\,(\Delta \hat a-2\Delta a)}{(N+2a_1)(N+2a_2)\big(N+2a_1+\Delta \hat a\big)\big(N+2a_2-\Delta \hat a\big)}.
\]
Under the maintained interior conditions, the denominator is positive and $N+a_1+a_2>0$, so $\mathcal R_G(\eta^{EB})\le \mathcal R_G(\eta^{NI})$ if and only if
\[
\Delta \hat a\,(\Delta \hat a-2\Delta a)\le 0,
\]
which is equivalent to $(\hat a_2-\hat a_1)(a_2-a_1)\ge 0$ and $|\hat a_2-\hat a_1|\le 2|a_2-a_1|$.
The precision version follows since $a_g=s^2/v_g$. 
\end{proof}

\section{Verification of Assumption~\ref{as:Psi_Lip} for Common Decision Problems}
\label{app:Psi_Lip_examples}

We verify Assumption~\ref{as:Psi_Lip} for the five decision problems listed in Table~\ref{tab:ql_examples}.
Throughout, let $m_1,m_2$ denote generic posterior-mean arguments (either in $\R^d$ or in $\R$ as appropriate).

\begin{proposition}
\label{prop:Psi_Lip_examples}
The continuation values $\Psi$ in Table~\ref{tab:ql_examples} satisfy Assumption~\ref{as:Psi_Lip}.
In particular:
\begin{enumerate}[leftmargin=*]
\item For quadratic loss, $\Psi(m)=m'\Lambda m$ satisfies Assumption~\ref{as:Psi_Lip} with $p=1$ and $L=\|\Lambda\|_{\mathrm F}$.
\item For the mean--variance / LQ objective, $\Psi(m)=\frac{1}{2\gamma}m'\Sigma^{-1}m$ satisfies Assumption~\ref{as:Psi_Lip} with $p=1$ and
$L=\frac{1}{2\gamma}\|\Sigma^{-1}\|_{\mathrm F}$.
\item For binary adoption, $\Psi(m)=\sum_{g=1}^G \pi_g \max\{0,m_g\}$ satisfies Assumption~\ref{as:Psi_Lip} with $p=0$ and $L=\|\pi\|_2$.
\item For best-alternative selection, $\Psi(m)=\max_{j\le K} m_j$ satisfies Assumption~\ref{as:Psi_Lip} with $p=0$ and $L=1$.
\item For weighted hypothesis testing, $\Psi(m)=\max\{-a_0(1-m),-a_1 m\}$ satisfies Assumption~\ref{as:Psi_Lip} with $p=0$ and
$L=\max\{a_0,a_1\}$.
\end{enumerate}
\end{proposition}

\begin{proof}
We verify each case.

\medskip
\noindent\emph{(1) Quadratic loss.}
Let $\Psi(m)=m'\Lambda m$ with $\Lambda\succeq 0$.
Then
\[
\Psi(m_1)-\Psi(m_2)=m_1'\Lambda m_1-m_2'\Lambda m_2=(m_1-m_2)'\Lambda(m_1+m_2).
\]
Using $x'Ay=\mathrm{tr}(A y x')$ and Cauchy--Schwarz under the Frobenius inner product,
\[
\big|(m_1-m_2)'\Lambda(m_1+m_2)\big|
=
\big|\mathrm{tr}\!\big(\Lambda (m_1+m_2)(m_1-m_2)'\big)\big|
\le
\|\Lambda\|_{\mathrm F}\,\|(m_1+m_2)(m_1-m_2)'\|_{\mathrm F}.
\]
Moreover, $\|(m_1+m_2)(m_1-m_2)'\|_{\mathrm F}=\|m_1+m_2\|_2\,\|m_1-m_2\|_2$, so
\[
|\Psi(m_1)-\Psi(m_2)|
\le
\|\Lambda\|_{\mathrm F}\,\|m_1+m_2\|_2\,\|m_1-m_2\|_2
\le
\|\Lambda\|_{\mathrm F}\,(\|m_1\|_2+\|m_2\|_2)\,\|m_1-m_2\|_2.
\]
This implies Assumption~\ref{as:Psi_Lip} with $p=1$ and $L=\|\Lambda\|_{\mathrm F}$ (since
$\|m_1\|_2+\|m_2\|_2 \le 1+\|m_1\|_2+\|m_2\|_2$).

\medskip
\noindent\emph{(2) Mean--variance / LQ objective.}
Let $\Psi(m)=\frac{1}{2\gamma}m'\Sigma^{-1}m$ with $\gamma>0$ and $\Sigma\succ 0$.
The same argument as in (1), with $\Lambda=\frac{1}{2\gamma}\Sigma^{-1}$, yields
\[
|\Psi(m_1)-\Psi(m_2)|
\le
\frac{1}{2\gamma}\|\Sigma^{-1}\|_{\mathrm F}\,(\|m_1\|_2+\|m_2\|_2)\,\|m_1-m_2\|_2,
\]
so Assumption~\ref{as:Psi_Lip} holds with $p=1$ and $L=\frac{1}{2\gamma}\|\Sigma^{-1}\|_{\mathrm F}$.

\medskip
\noindent\emph{(3) Binary adoption.}
Let $\Psi(m)=\sum_{g=1}^G \pi_g (m_g)_+$ where $(x)_+:=\max\{0,x\}$.
The map $x\mapsto (x)_+$ is $1$-Lipschitz on $\R$, so $|(x)_+-(y)_+|\le |x-y|$.
Thus,
\[
|\Psi(m_1)-\Psi(m_2)|
\le
\sum_{g=1}^G |\pi_g|\,|(m_{1g})_+-(m_{2g})_+|
\le
\sum_{g=1}^G |\pi_g|\,|m_{1g}-m_{2g}|.
\]
By Cauchy--Schwarz,
\[
\sum_{g=1}^G |\pi_g|\,|m_{1g}-m_{2g}|
\le
\|\pi\|_2\,\|m_1-m_2\|_2.
\]
Hence Assumption~\ref{as:Psi_Lip} holds with $p=0$ and $L=\|\pi\|_2$.

\medskip
\noindent\emph{(4) Best-alternative selection.}
Let $\Psi(m)=\max_{j\le K} m_j$.
Then
\[
|\Psi(m_1)-\Psi(m_2)|
=
\big|\max_j m_{1j}-\max_j m_{2j}\big|
\le
\max_j |m_{1j}-m_{2j}|
=
\|m_1-m_2\|_\infty
\le
\|m_1-m_2\|_2.
\]
Thus Assumption~\ref{as:Psi_Lip} holds with $p=0$ and $L=1$.

\medskip
\noindent\emph{(5) Weighted hypothesis testing.}
Let $\Psi(m)=\max\{-a_0(1-m),-a_1 m\}$ with $a_0,a_1\ge 0$ and scalar $m$.
Define the affine functions $\ell_0(m):=-a_0(1-m)=-a_0+a_0 m$ and $\ell_1(m):=-a_1 m$.
Each is Lipschitz, with constants $a_0$ and $a_1$ respectively.
The pointwise maximum of Lipschitz functions is Lipschitz with constant equal to the maximum of the individual constants, hence
\[
|\Psi(m_1)-\Psi(m_2)|
\le
\max\{a_0,a_1\}\,|m_1-m_2|
=
\max\{a_0,a_1\}\,\|m_1-m_2\|_2.
\]
Therefore Assumption~\ref{as:Psi_Lip} holds with $p=0$ and $L=\max\{a_0,a_1\}$.
\end{proof}

\section{Verification of Second-Order Assumptions in Example~\ref{emp:two_stratum_quadratic}}
\label{app:two_stratum_secondorder}

This subsection verifies Assumptions~\ref{as:strong_concavity} and~\ref{as:grad_approx} in the two-stratum quadratic-loss allocation setting of
Example~\ref{emp:two_stratum_quadratic}. Throughout, write $a_g:=s^2/v_g$ and $\hat a_g:=s^2/\hat v_g$.

\begin{proposition}[Strong concavity and uniform gradient approximation]
\label{prop:two_stratum_secondorder}
Consider the two-stratum Gaussian quadratic-loss problem in Example~\ref{emp:two_stratum_quadratic} with feasible set
$\mathcal H=\{(N_1,N_2)\in\R_+^2:\ N_1+N_2=N\}$ and $v_g>0$.
\begin{enumerate}[leftmargin=*]
\item \emph{Strong concavity.} The oracle objective $U_G$ is continuously differentiable on $\mathcal H$ and satisfies Assumption~\ref{as:strong_concavity}
(on the affine hull of $\mathcal H$) with a valid strong-concavity modulus
\[
m \;=\; s^2\Big((N+a_1)^{-3}+(N+a_2)^{-3}\Big) \;>\;0.
\]
Moreover, $U_G$ is strictly concave on $\mathcal H$, hence has a unique maximizer, which is interior under the maintained interior condition
$|a_2-a_1|<N$.
\item \emph{Uniform gradient approximation.} Define the plug-in objective $U_{\hat G_n}$ by replacing $(a_1,a_2)$ with $(\hat a_1,\hat a_2)$.
On any event where $\hat a_g\ge a_g/2$ for $g=1,2$, one has
\[
\sup_{\eta\in\mathcal H}\|\nabla U_{\hat G_n}(\eta)-\nabla U_G(\eta)\|_2
\;\le\;
C\sum_{g=1}^2 |\hat a_g-a_g|
\]
for a finite constant $C$ depending only on $(s^2,a_1,a_2)$.
In particular, if $|\hat v_g-v_g|=O_{p}(n^{-1/2})$ for $g=1,2$, then
\[
\sup_{\eta\in\mathcal H}\|\nabla U_{\hat G_n}(\eta)-\nabla U_G(\eta)\|_2
=
O_{p}(n^{-1/2}),
\]
so Assumption~\ref{as:grad_approx} holds with $r_n=O_{p}(n^{-1/2})$.
\end{enumerate}
\end{proposition}

\begin{proof}
First, we derive the closed form for $U_G(\eta)$.
As shown in Appendix~\ref{app:two_stratum_quadratic}, under quadratic loss with $L=I_2$ and $\Lambda=I_2$ the ex-ante value equals minus the Bayes risk,
and conjugacy implies
\[
\Var(\theta_g\mid \hat\theta_g,N_g)=\left(v_g^{-1}+\frac{N_g}{s^2}\right)^{-1}=\frac{s^2}{N_g+a_g},
\qquad a_g:=\frac{s^2}{v_g}.
\]
Hence, up to an additive constant that does not depend on $\eta$,
\[
U_G(\eta)= -\sum_{g=1}^2 \frac{s^2}{N_g+a_g}.
\]
This formula implies $U_G$ is continuously differentiable on $\mathcal H$ and
\[
\frac{\partial U_G(\eta)}{\partial N_g}=\frac{s^2}{(N_g+a_g)^2},
\qquad g=1,2.
\]

\medskip
Second, we prove the strong concavity of $U_G(\eta)$ on $\mathcal H$.
Parametrize $\eta$ by $t=N_1\in[0,N]$ (so $N_2=N-t$) and define
\[
f(t):=U_G\big((t,N-t)\big)= -\frac{s^2}{t+a_1}-\frac{s^2}{N-t+a_2}.
\]
Then
\[
f''(t)=-2s^2\left((t+a_1)^{-3}+(N-t+a_2)^{-3}\right)
\le
-2s^2\left((N+a_1)^{-3}+(N+a_2)^{-3}\right).
\]
Let $c:=2s^2((N+a_1)^{-3}+(N+a_2)^{-3})$. By Taylor's theorem, for any $t_1,t_2\in[0,N]$,
\[
f(t_2) \le f(t_1)+f'(t_1)(t_2-t_1)-\frac{c}{2}(t_2-t_1)^2.
\]
Now for $\eta=(t_1,N-t_1)$ and $\eta_2=(t_2,N-t_2)$, we have $\|\eta_2-\eta_1\|_2^2=2(t_2-t_1)^2$ and
$\nabla U_G(\eta_1)^\prime(\eta_2-\eta_1)=f'(t_1)(t_2-t_1)$. Therefore
\[
U_G(\eta_2) \le U_G(\eta_1)+\nabla U_G(\eta_1)^\prime(\eta_2-\eta_1)-\frac{m}{2}\|\eta_2-\eta_1\|_2^2
\]
holds with $m=c/2=s^2((N+a_1)^{-3}+(N+a_2)^{-3})>0$, verifying Assumption~\ref{as:strong_concavity} on the affine hull of $\mathcal H$.
Strict concavity follows from $f''(t)<0$ for all $t\in[0,N]$, and interiority of the maximizer holds under $|a_2-a_1|<N$ as in Appendix~\ref{app:two_stratum_quadratic}.

\medskip
Third, we show the uniform gradient approximation condition for $U_G(\eta)$.
Define $U_{\hat G_n}$ by replacing $a_g$ with $\hat a_g$. Then
\[
\frac{\partial U_{\hat G_n}(\eta)}{\partial N_g}=\frac{s^2}{(N_g+\hat a_g)^2}.
\]
Fix $g\in\{1,2\}$. By the mean value theorem applied to $x\mapsto s^2/(N_g+x)^2$, there exists $\tilde a_g$ between $a_g$ and $\hat a_g$ such that
\[
\left|\frac{s^2}{(N_g+\hat a_g)^2}-\frac{s^2}{(N_g+a_g)^2}\right|
=
\frac{2s^2}{(N_g+\tilde a_g)^3}\,|\hat a_g-a_g|.
\]
On the event $\{\hat a_g\ge a_g/2\}$, we have $N_g+\tilde a_g\ge a_g/2$ for all $N_g\ge 0$, hence
\[
\sup_{\eta\in\mathcal H}\left|\frac{\partial U_{\hat G_n}(\eta)}{\partial N_g}-\frac{\partial U_G(\eta)}{\partial N_g}\right|
\le
\frac{16s^2}{a_g^3}\,|\hat a_g-a_g|.
\]
Combining the two coordinates yields
\[
\sup_{\eta\in\mathcal H}\|\nabla U_{\hat G_n}(\eta)-\nabla U_G(\eta)\|_2
\le
\left(\sum_{g=1}^2 \Big(\frac{16s^2}{a_g^3}\,|\hat a_g-a_g|\Big)^2\right)^{1/2}
\le
C\sum_{g=1}^2 |\hat a_g-a_g|
\]
for a finite constant $C$ depending only on $(s^2,a_1,a_2)$.

Finally, since $\hat a_g-a_g=s^2(\hat v_g^{-1}-v_g^{-1})$, the delta method implies
$|\hat a_g-a_g|=O_{p}(n^{-1/2})$ whenever $|\hat v_g-v_g|=O_{p}(n^{-1/2})$ and $v_g>0$.
Therefore
\[
\sup_{\eta\in\mathcal H}\|\nabla U_{\hat G_n}(\eta)-\nabla U_G(\eta)\|_2
=
O_{p}(n^{-1/2}),
\]
which verifies Assumption~\ref{as:grad_approx} with $r_n=O_{p}(n^{-1/2})$.
\end{proof}

\section{Proofs}
\label{app:proofs}

\subsection{Proof of Proposition~\ref{prop:quadratic_design}}
Fix a feasible design $e\in\mathcal H$. Under the Gaussian prior and sampling model,
\[
\theta \sim \mathcal N(m,V),\qquad \hat\theta\mid \theta;\,e \sim \mathcal N(\theta,\Sigma(e)),
\]
the posterior is
\[
\theta\mid \hat\theta,e \sim \mathcal N\!\big(m^{post}(\hat\theta,e),\,V^{post}(e)\big),
\qquad
V^{post}(e)=(V^{-1}+\Sigma(e)^{-1})^{-1},
\]
with posterior mean
\[
m^{post}(\hat\theta,e)=V^{post}(e)\big(V^{-1}m+\Sigma(e)^{-1}\hat\theta\big).
\]
In particular, $V^{post}(e)$ depends on the design $e$ but not on the realized $\hat\theta$.

Let $Z:=L\theta$. Conditional on $(\hat\theta,e)$, $Z$ has mean $L m^{post}(\hat\theta,e)$ and covariance $L V^{post}(e)L'$. For any random vector $Z$ with mean $\mu$ and covariance $\Omega$,
\[
\E\big[(a-Z)'\Lambda(a-Z)\big]=(a-\mu)'\Lambda(a-\mu)+\tr(\Lambda\Omega),
\]
so the Bayes action is $a^*(\hat\theta,e)=\mu=\E[L\theta\mid \hat\theta,e]$ and the minimized posterior risk equals
\[
\inf_{a\in\R^K}\E\big[\ell(a,\theta)\mid \hat\theta,e\big]
=
\tr\!\Big(\Lambda\,\Var(L\theta\mid \hat\theta,e)\Big)
=
\tr\!\Big(\Lambda\,L V^{post}(e)L'\Big).
\]
Since this expression is deterministic given $e$, taking expectations over $\hat\theta$ leaves it unchanged, proving
\[
\mathcal R_Q(e)=\tr\!\Big(\Lambda\,L V^{post}(e)L'\Big).
\]
Finally, in the diffuse-prior limit $V^{-1}\to 0$, one has $V^{post}(e)=(V^{-1}+\Sigma(e)^{-1})^{-1}\to \Sigma(e)$, hence
$\mathcal R_Q(e)\to \tr(\Lambda\,L\Sigma(e)L')$. \qed

\subsection{Proof of Proposition~\ref{prop:insample_formal}}
Let $m_g=\E_Q[\tau_g]$ and consider
\[
\max_{e}\ \sum_{g=1}^G \pi_g e_g m_g
\quad\text{s.t.}\quad
\sum_{g=1}^G \pi_g c_g e_g\le B,\qquad
\underline e\le e_g\le 1-\underline e\ \ \forall g.
\]
This is a linear program over a nonempty compact polytope, so an optimizer exists. Let $e^*$ be an optimizer. The KKT conditions are necessary and sufficient. Introduce multipliers $\lambda\ge 0$, $\alpha_g\ge 0$ and $\beta_g\ge 0$ for the budget, lower, and upper constraints, respectively. Stationarity gives, for each $g$,
\[
0=\frac{\partial}{\partial e_g}\Big(\sum_{h}\pi_h e_h m_h-\lambda(\sum_h\pi_h c_h e_h-B)+\sum_h \alpha_h(e_h-\underline e)+\sum_h \beta_h((1-\underline e)-e_h)\Big)
=\pi_g(m_g-\lambda c_g)+\alpha_g-\beta_g.
\]
Complementary slackness yields $\lambda\big(\sum_h \pi_h c_h e_h^*-B\big)=0$, $\alpha_g(e_g^*-\underline e)=0$ and $\beta_g((1-\underline e)-e_g^*)=0$.

If $e_g^*\in(\underline e,1-\underline e)$, then $\alpha_g=\beta_g=0$ and the stationarity condition implies $m_g=\lambda c_g$.
If $m_g>\lambda c_g$, then $\pi_g(m_g-\lambda c_g)>0$ so stationarity forces $\beta_g>\alpha_g\ge 0$, hence $\beta_g>0$ and complementary slackness implies $e_g^*=1-\underline e$.
If $m_g<\lambda c_g$, then $\pi_g(m_g-\lambda c_g)<0$ so stationarity forces $\alpha_g>0$ and complementary slackness implies $e_g^*=\underline e$.
Therefore any optimizer is bang-bang except possibly on the non-generic set of strata with $m_g=\lambda c_g$, for which any $e_g^*\in[\underline e,1-\underline e]$ is compatible with KKT. The scalar $\lambda$ is selected so that the budget constraint binds when interior (and $\lambda=0$ if the constraint is slack). \qed

\subsection{Proof of Lemma~\ref{lem:postmean_distribution}} 
Fix a stratum $g$ and suppress the index. Under the Gaussian model
\[
\delta\sim\mathcal N(m,v),\qquad \hat\delta\mid \delta;\,e \sim \mathcal N(\delta,s^2),
\]
where $s^2=s_g^2(e_g)$. Normal--normal conjugacy gives
\[
\delta\mid \hat\delta \sim \mathcal N\!\left(\mu,\ \frac{vs^2}{v+s^2}\right),
\qquad
\mu=\E_{Q}[\delta\mid \hat\delta]=\frac{s^2}{v+s^2}m+\frac{v}{v+s^2}\hat\delta.
\]
Thus $\mu=m+\kappa(\hat\delta-m)$ with $\kappa:=\frac{v}{v+s^2}$. The marginal distribution of $\hat\delta$ is $\hat\delta\sim\mathcal N(m,v+s^2)$, so $\mu$ is Gaussian with
\[
\Var(\mu)=\kappa^2\Var(\hat\delta)=\Big(\frac{v}{v+s^2}\Big)^2(v+s^2)=\frac{v^2}{v+s^2},
\]
hence $\mu\sim\mathcal N\!\big(m,\frac{v^2}{v+s^2}\big)$, which matches the stated $\sigma_B^2$ after reinstating the stratum index.

Finally, $(\delta,\mu)$ is jointly Gaussian since $\mu$ is affine in $\hat\delta$ and $(\delta,\hat\delta)$ is jointly Gaussian. Writing $\hat\delta=\delta+\varepsilon$ with $\varepsilon\sim\mathcal N(0,s^2)$ independent of $\delta$,
\[
\mu-m=\kappa(\hat\delta-m)=\kappa\big((\delta-m)+\varepsilon\big).
\]
Therefore
\[
\Cov(\delta,\mu)=\Cov(\delta-m,\mu-m)=\kappa\,\Var(\delta-m)=\kappa v=\frac{v^2}{v+s^2}=\Var(\mu).
\]
For jointly Gaussian variables, $\E_{Q}[\delta\mid \mu]=m+\frac{\Cov(\delta,\mu)}{\Var(\mu)}(\mu-m)=m+(\mu-m)=\mu$. \qed

\subsection{Proof of Proposition~\ref{prop:policy_voI_formal}}
Fix a stratum $g$ and suppress the index when unambiguous.
Let the post-experiment adoption action be $a\in\{0,1\}$ with welfare $a\delta$.
Under the Gaussian prior--likelihood pair,
\[
\delta \sim \mathcal N(m,v),
\qquad
\hat\delta \mid \delta;\,e \sim \mathcal N(\delta,s^2(e)),
\]
where $e=e_g$ and $s^2(e)=s_g^2(e_g)$.
Write
\[
\mu:=\E_{Q}[\delta\mid \hat\delta;\,e],
\]
which coincides with $\E_{Q}[\delta\mid D(e)]$ since the experimental data enter only through $(\hat\delta,s^2(e))$.
Conditional on the observed statistic, posterior expected welfare from action $a$ equals
\[
\E_{Q}[a\delta\mid \hat\delta;\,e]=a\,\E_{Q}[\delta\mid \hat\delta;\,e]=a\,\mu.
\]
Thus the Bayes-optimal adoption rule under prior $Q$ is $a^*_{Q}(\hat\delta;\,e)=\mathbf 1\{\mu>0\}$ and the stratum-level ex-ante value is
\[
\E_{Q}[\delta\,a^*_{Q}(\hat\delta;\,e)]
=
\E_{Q}[\delta\,\mathbf 1\{\mu>0\}].
\]
By iterated expectations and Lemma~\ref{lem:postmean_distribution}, $\E_{Q}[\delta\mid \mu]=\mu$, hence
\[
\E_{Q}[\delta\,\mathbf 1\{\mu>0\}]
=
\E_{Q}\big[\E_{Q}[\delta\mid \mu]\mathbf 1\{\mu>0\}\big]
=
\E_{Q}[\mu\,\mathbf 1\{\mu>0\}].
\]
Lemma~\ref{lem:postmean_distribution} also gives $\mu\sim \mathcal N(m,\sigma_B^2(e))$, where $\sigma_B^2(e)=\frac{v^2}{v+s^2(e)}$.
For $X\sim\mathcal N(m,\sigma^2)$,
\[
\E[X\,\mathbf 1\{X>0\}]
=
m\,\Pr(X>0)+\sigma\,\E\Big[Z\,\mathbf 1\{Z>-m/\sigma\}\Big],
\quad Z:=\frac{X-m}{\sigma}\sim\mathcal N(0,1).
\]
Since $\Pr(X>0)=\Phi(m/\sigma)$ and
\[
\E\Big[Z\,\mathbf 1\{Z>-\tfrac{m}{\sigma}\}\Big]=\int_{-m/\sigma}^{\infty} z\phi(z)\,dz
=\big[-\phi(z)\big]_{-m/\sigma}^{\infty}=\phi(m/\sigma),
\]
it follows that
\[
\E[X\,\mathbf 1\{X>0\}]
=
m\,\Phi\!\Big(\frac{m}{\sigma}\Big)+\sigma\,\phi\!\Big(\frac{m}{\sigma}\Big).
\]
Applying this with $X=\mu$ and $\sigma=\sigma_B(e)$ yields
\[
\E_{Q}[\delta\,a^*_{Q}(\hat\delta;\,e)]
=
m\,\Phi\!\Big(\frac{m}{\sigma_B(e)}\Big)+\sigma_B(e)\,\phi\!\Big(\frac{m}{\sigma_B(e)}\Big)
=:\Gamma(e).
\]
Summing across strata with weights $\pi_g$ gives \eqref{eq:U_policy_sum}--\eqref{eq:Gamma_policy_formal}.

It remains to characterize an interior optimum under the binding cost constraint.
The design problem is
\[
\max_{e\in\mathcal H}\ \sum_{g=1}^G \pi_g\,\Gamma_g(e_g)
\quad\text{s.t.}\quad
\sum_{g=1}^G \pi_g c_g e_g\le B,\ \ e_g\in[\underline e,1-\underline e].
\]
Form the Lagrangian
\[
\mathcal L(e,\lambda,\alpha,\beta)
=
\sum_{g=1}^G \pi_g\,\Gamma_g(e_g)
+\lambda\Big(\sum_{g=1}^G \pi_g c_g e_g-B\Big)
+\sum_{g=1}^G \alpha_g(\underline e-e_g)
+\sum_{g=1}^G \beta_g(e_g-(1-\underline e)),
\]
with $\lambda\ge 0$, $\alpha_g\ge 0$, $\beta_g\ge 0$.
If the cost constraint binds and the optimum is interior (so $\alpha_g=\beta_g=0$ for all $g$), the first-order conditions are
\[
0=\frac{\partial \mathcal L}{\partial e_g}
=
\pi_g\,\Gamma_g'(e_g)+\lambda\,\pi_g c_g,
\qquad g=1,\dots,G.
\]
To compute $\Gamma_g'(e_g)$, note that $\Gamma_g(e_g)$ depends on $e_g$ only through $\sigma=\sigma_{B,g}(e_g)$:
\[
\Gamma_g(e_g)= m_g\,\Phi\!\Big(\frac{m_g}{\sigma}\Big)+\sigma\,\phi\!\Big(\frac{m_g}{\sigma}\Big),
\qquad \sigma=\sigma_{B,g}(e_g).
\]
Differentiating with respect to $\sigma$ yields
\[
\frac{d}{d\sigma}\left\{ m_g\,\Phi\!\Big(\frac{m_g}{\sigma}\Big)+\sigma\,\phi\!\Big(\frac{m_g}{\sigma}\Big)\right\}
=
\phi\!\Big(\frac{m_g}{\sigma}\Big),
\]
so by the chain rule,
\[
\Gamma_g'(e_g)
=
\phi\!\Big(\frac{m_g}{\sigma_{B,g}(e_g)}\Big)\cdot
\frac{\partial \sigma_{B,g}(e_g)}{\partial e_g}.
\]
Substituting into the first-order condition and canceling $\pi_g>0$ gives \eqref{eq:kkt_policy_formal}.
If an optimum lies on the overlap bounds, the complementary-slackness conditions imply truncation to $[\underline e,1-\underline e]$ as stated. \qed

\subsection{Proof of Theorem~\ref{thm:oracle-ineq}}
Let $\eta^{EB}\in\argmax_{\eta\in\mathcal H}U_{\hat G_n}(\eta)$. Add and subtract $U_{\hat G_n}(\eta^O)$ and $U_{\hat G_n}(\eta^{EB})$:
\begin{align*}
U_G(\eta^O)-U_G(\eta^{EB})
&=\big(U_G(\eta^O)-U_{\hat G_n}(\eta^O)\big)
+\big(U_{\hat G_n}(\eta^O)-U_{\hat G_n}(\eta^{EB})\big)
+\big(U_{\hat G_n}(\eta^{EB})-U_G(\eta^{EB})\big).
\end{align*}
The middle term is nonpositive by optimality of $\eta^{EB}$ for $U_{\hat G_n}$. The other two terms are each bounded above by $\Delta_n$. Nonnegativity holds because $\eta^O$ maximizes $U_G$. \qed

\subsection{Proof of Corollary~\ref{cor:eb-vs-np}}
Rearrange \eqref{eq:oracle-ineq}: $U_G(\eta^{EB}) \ge U_G(\eta^O)-2\Delta_n$. Subtracting $U_G(\eta^{NI})$ on both sides yields the inequality. \qed

\subsection{Proof of Lemma~\ref{lem:weak-to-Delta-unbounded}}
We begin with a deterministic statement. Let $(Q_m)_{m\ge1}\subset\mathcal Q$ and $Q\in\mathcal Q$ satisfy $Q_m\Rightarrow Q$. We show
\[
\sup_{\eta\in\mathcal H}\big|U_{Q_m}(\eta)-U_Q(\eta)\big|\to 0.
\]
Let $Z:=\mathcal A\times\mathcal H$, which is compact by Assumption~\ref{ass:cont-envelope}(i). For $\lambda$-a.e.\ $y$, define
\[
f_y(z,\theta):=W(a,\theta)\,p_\eta(y\mid\theta),\qquad z=(a,\eta)\in Z.
\]
By Assumption~\ref{ass:cont-envelope}(ii)--(iii), for $\lambda$-a.e.\ $y$ the map $(z,\theta)\mapsto f_y(z,\theta)$ is continuous and satisfies
\[
\sup_{z\in Z}|f_y(z,\theta)|\le w(\theta)\,\bar p(y).
\]

Fix $L>0$. Choose a continuous function $\psi:[0,\infty)\to[0,1]$ with $\psi(t)=1$ for $t\le 1$ and $\psi(t)=0$ for $t\ge 2$, and set $\chi_L(\theta):=\psi(w(\theta)/L)$. Let $f_{y,L}(z,\theta):=f_y(z,\theta)\chi_L(\theta)$. Then $f_{y,L}$ is bounded and continuous on $Z\times\R^d$ and satisfies $f_{y,L}(z,\theta)=0$ whenever $w(\theta)\ge 2L$. Moreover, for any probability measure $R$ on $\R^d$ and any $y$,
\begin{align*}
\sup_{z\in Z}\Big|\int \big(f_y(z,\theta)-f_{y,L}(z,\theta)\big)\,dR(\theta)\Big|
&\le
\bar p(y)\int w(\theta)\,\mathbf 1\{w(\theta)>L\}\,dR(\theta).
\end{align*}
Set $C:=\sup_{R\in\mathcal Q}\int w(\theta)^{1+\delta}\,dR(\theta)<\infty$
(Assumption~\ref{ass:cont-envelope}(iv)). Then for any $R\in\mathcal Q$,
\[
\int w(\theta)\,\mathbf 1\{w(\theta)>L\}\,dR(\theta)
=
\int w(\theta)^{1+\delta} w(\theta)^{-\delta}\,\mathbf 1\{w(\theta)>L\}\,dR(\theta)
\le
L^{-\delta}\int w(\theta)^{1+\delta}\,dR(\theta)
\le
C\,L^{-\delta},
\]
where we used $w(\theta)^{-\delta}\le L^{-\delta}$ on $\{w(\theta)>L\}$.
Hence we have 
$$
\sup_{R\in\mathcal Q}\sup_{z\in Z}\Big|\int \big(f_y(z,\theta)-f_{y,L}(z,\theta)\big)\,dR(\theta)\Big|\le \bar p(y)C\,L^{-\delta}.
$$

Now fix $L$ and a $y$ for which the continuity statements hold. Since $Z\times K_{2L}$ is compact with $K_{2L}:=\{\theta:w(\theta)\le 2L\}$ and $f_{y,L}$ is continuous, it is uniformly continuous in $z$ uniformly over $\theta\in K_{2L}$. Let $\varepsilon>0$ and choose a finite $\varepsilon$-net $\{z_1,\dots,z_J\}\subset Z$ such that for each $z\in Z$ there exists $j(z)$ with
\[
\sup_{\theta\in\R^d}\big|f_{y,L}(z,\theta)-f_{y,L}(z_{j(z)},\theta)\big|\le \varepsilon.
\]
For each fixed $j$, $\theta\mapsto f_{y,L}(z_j,\theta)$ is bounded and continuous, so weak convergence implies
\[
\int f_{y,L}(z_j,\theta)\,dQ_m(\theta)\to \int f_{y,L}(z_j,\theta)\,dQ(\theta).
\]
Since $J$ is finite,
\[
\max_{1\le j\le J}\Big|\int f_{y,L}(z_j,\theta)\,d(Q_m-Q)(\theta)\Big|\to 0.
\]
For any $z\in Z$, the triangle inequality and the net property give
\[
\Big|\int f_{y,L}(z,\theta)\,d(Q_m-Q)(\theta)\Big|
\le
\Big|\int f_{y,L}(z_{j(z)},\theta)\,d(Q_m-Q)(\theta)\Big|+2\varepsilon,
\]
hence $\sup_{z\in Z}\big|\int f_{y,L}(z,\theta)\,d(Q_m-Q)(\theta)\big|\to 0$ by letting $m\to\infty$ and then $\varepsilon\downarrow 0$.

Combining the last display with the truncation bound, for $\lambda$-a.e.\ $y$ we obtain
\begin{align*}
\limsup_{m\to\infty}\ \sup_{z\in Z}\Big|\int f_y(z,\theta)\,d(Q_m-Q)(\theta)\Big|
&\le 2C\,\bar p(y)\,L^{-\delta}.
\end{align*}
Letting $L\to\infty$ yields
\[
\sup_{z\in Z}\Big|\int f_y(z,\theta)\,d(Q_m-Q)(\theta)\Big|\to 0
\qquad\text{for }\lambda\text{-a.e.\ }y.
\]

For such $y$ and each $\eta$,
\[
|u_{Q_m}(y;\eta)-u_Q(y;\eta)|
\le
\sup_{a\in\mathcal A}\Big|\int W(a,\theta)p_\eta(y\mid\theta)\,d(Q_m-Q)(\theta)\Big|
\le
\sup_{z\in Z}\Big|\int f_y(z,\theta)\,d(Q_m-Q)(\theta)\Big|,
\]
so $\sup_{\eta\in\mathcal H}|u_{Q_m}(y;\eta)-u_Q(y;\eta)|\to 0$ for $\lambda$-a.e.\ $y$. Moreover,
\[
\sup_{\eta\in\mathcal H}|u_{Q_m}(y;\eta)-u_Q(y;\eta)|
\le
\sup_{\eta\in\mathcal H}|u_{Q_m}(y;\eta)|+\sup_{\eta\in\mathcal H}|u_Q(y;\eta)|,
\]
and the right-hand side is integrable uniformly in $m$ by Assumption~\ref{ass:value-envelope}(ii). Dominated convergence therefore gives
\[
\int \sup_{\eta\in\mathcal H}|u_{Q_m}(y;\eta)-u_Q(y;\eta)|\,d\lambda(y)\to 0.
\]
Finally,
\begin{align*}
\sup_{\eta\in\mathcal H}|U_{Q_m}(\eta)-U_Q(\eta)|
&=
\sup_{\eta\in\mathcal H}\Big|\int \big(u_{Q_m}(y;\eta)-u_Q(y;\eta)\big)\,d\lambda(y)\Big| \\
&\le
\int \sup_{\eta\in\mathcal H}|u_{Q_m}(y;\eta)-u_Q(y;\eta)|\,d\lambda(y)\ \longrightarrow\ 0,
\end{align*}
which proves the deterministic claim.

To return to random $\hat G_n$, let $(n_k)$ be an arbitrary subsequence. Since $\hat G_n\Rightarrow G$ in probability, there exists a further subsequence $(n_{k_\ell})$ such that $\hat G_{n_{k_\ell}}\Rightarrow G$ almost surely. Because $\Pr(\hat G_n\in\mathcal Q)\to 1$, we may additionally choose the subsequence so that $\sum_\ell \Pr(\hat G_{n_{k_\ell}}\notin\mathcal Q)<\infty$, and then Borel--Cantelli implies $\hat G_{n_{k_\ell}}\in\mathcal Q$ for all sufficiently large $\ell$ almost surely. On this probability-one event, the deterministic claim applies pathwise with $Q_m=\hat G_{n_{k_\ell}}$ and $Q=G$, yielding $\Delta_{n_{k_\ell}}\to 0$ almost surely. Since every subsequence admits such a further subsequence, it follows that $\Delta_n\to 0$ in probability. \qed

\subsection{Proof of Corollary~\ref{cor:gaussian-eb-oracle-unbounded}}
Parameter convergence implies $\hat G_n\Rightarrow G$. Choose $\mathcal Q$ to be a set of Gaussians with $(\tau,V)$ restricted to a compact set in which eigenvalues of $V$ are bounded away from $0$ and $\infty$. Then Assumption~\ref{ass:cont-envelope}(iv) holds for any coercive polynomial envelope $w(\theta)$ because Gaussian moments are finite and uniformly bounded on such a compact parameter set. Lemma~\ref{lem:weak-to-Delta-unbounded} gives $\Delta_n\to 0$, and Theorem~\ref{thm:regret-consistency} concludes. \qed

\subsection{Proof of Proposition~\ref{prop:quad-gaussian-Delta}}
Under quadratic welfare, the Bayes action is the posterior mean and the optimized posterior welfare equals $-\tr(\V(\theta\mid Y_\eta))$; under Gaussian conjugacy, $\V(\theta\mid Y_\eta)=(V^{-1}+\Sigma(\eta)^{-1})^{-1}$ and is deterministic given $\eta$. Uniform continuity in $V$ follows from continuity of the map $V\mapsto (V^{-1}+A)^{-1}$ on the set of positive-definite $V$ with eigenvalues bounded away from $0$ and $A=\Sigma(\eta)^{-1}$ bounded uniformly over $\eta$. \qed

\subsection{Proof of Lemma~\ref{lem:value_reduction}}
Fix $(Q,\eta)$. By iterated expectations and Assumption~\ref{as:ql_state}, for any action $a$,
\[
\mathbb E_Q[W(a,\theta)\mid Y=y,\eta]
=\mathbb E_Q[w_0(\theta)\mid Y=y,\eta] + u(a) + v(a)^\prime\mu_{Q,\eta}(y).
\]
Taking the supremum over $a\in\mathcal A$ yields
$\sup_{a\in\mathcal A}\mathbb E_Q[W(a,\theta)\mid Y=y,\eta]=\mathbb E_Q[w_0(\theta)\mid Y=y,\eta] + \Psi(\mu_{Q,\eta}(y))$.
Taking expectations over $Y\sim f_{Q,\eta}$ and using iterated expectations gives
$U_Q(\eta)=\mathbb E_{Y\sim f_{Q,\eta}}[\Psi(\mu_{Q,\eta}(Y))] + \mathbb E_Q[w_0(\theta)]$. \qed

\subsection{Proof of Proposition~\ref{prop:npmle-weak-consistency}}
\label{app:proof-npmle-weak-consistency}

We condition throughout on the (deterministic) sequence of known covariance matrices $\{\Sigma_i\}_{i\ge1}$.
Let $\Theta_R:=\{\theta\in\mathbb R^d:\ \|\theta\|_2\le R\}$ and let
$\mathcal G_R$ denote the set of all probability measures supported on $\Theta_R$.
Let $d_{\mathrm{BL}}$ denote the bounded--Lipschitz metric on $\mathcal G_R$:
\[
d_{\mathrm{BL}}(G_1,G_2)
:=
\sup\left\{\left|\int f\,dG_1-\int f\,dG_2\right|:\ \|f\|_\infty\le 1,\ \mathrm{Lip}(f)\le 1\right\}.
\]
Since $\Theta_R$ is compact, $d_{\mathrm{BL}}$ metrizes weak convergence on $\mathcal G_R$ and $\mathcal G_R$ is compact under $d_{\mathrm{BL}}$.

For any positive definite $\Sigma$, define the Gaussian location-mixture density
\[
f_{G,\Sigma}(x):=\int_{\Theta_R}\varphi_{\Sigma}(x-\theta)\,dG(\theta).
\]
Given data $(\hat\theta_i,\Sigma_i)_{i=1}^n$, define the average log-likelihood
\[
\ell_n(G):=\frac1n\sum_{i=1}^n \log f_{G,\Sigma_i}(\hat\theta_i).
\]
By definition of the sieve NPMLE in Proposition~\ref{prop:npmle-weak-consistency}, $\hat G_n$ maximizes $\ell_n(\cdot)$ over $\mathcal G_R$.
We first prove three useful lemmas.

\medskip
\begin{lemma}[Envelope bounds for Gaussian mixtures]
\label{lem:npmle-envelope}
Let $\mathcal S:=\{\Sigma:\underline\sigma^2 I_d\preceq \Sigma\preceq \overline\sigma^2 I_d\}$.
Under Assumption~\ref{ass:npmle-consistency}, there exist finite constants $C_0,C_1,C_2<\infty$ (depending only on $d,R,\underline\sigma^2,\overline\sigma^2$)
such that for all $\Sigma\in\mathcal S$, all $G\in\mathcal G_R$, and all $x\in\mathbb R^d$,
\[
-C_0-C_1\|x\|_2^2 \ \le\ \log f_{G,\Sigma}(x)\ \le\ C_2.
\]
Consequently, for any subset $\mathcal B\subseteq \mathcal G_R$ and $f_{\mathcal B,\Sigma}(x):=\sup_{G\in\mathcal B} f_{G,\Sigma}(x)$,
\[
\left|\log\frac{f_{\mathcal B,\Sigma}(x)}{f_{G,\Sigma}(x)}\right|\ \le\ C_0'+C_1'\|x\|_2^2
\]
for some finite $C_0',C_1'$ depending only on $d,R,\underline\sigma^2,\overline\sigma^2$.
\end{lemma}

\begin{proof}
Fix $\Sigma\in\mathcal S$ and $u\in\mathbb R^d$. Since $|\Sigma|\ge (\underline\sigma^2)^d$ and $\exp(-\tfrac12 u^\prime\Sigma^{-1}u)\le 1$,
\[
\varphi_{\Sigma}(u)\le (2\pi)^{-d/2}(\underline\sigma^2)^{-d/2}=:M,
\]
so $f_{G,\Sigma}(x)\le M$ and hence $\log f_{G,\Sigma}(x)\le \log M=:C_2$.

Next, since $|\Sigma|\le (\overline\sigma^2)^d$ and $\Sigma^{-1}\preceq (\underline\sigma^2)^{-1}I_d$,
\[
\varphi_{\Sigma}(u)\ge (2\pi)^{-d/2}(\overline\sigma^2)^{-d/2}\exp\!\left(-\frac{\|u\|_2^2}{2\underline\sigma^2}\right)=:m\exp\!\left(-\frac{\|u\|_2^2}{2\underline\sigma^2}\right).
\]
For any $x$ and any $\theta\in\Theta_R$, we have $\|x-\theta\|_2\le \|x\|_2+R$, hence
\[
\varphi_{\Sigma}(x-\theta)\ge m\exp\!\left(-\frac{(\|x\|_2+R)^2}{2\underline\sigma^2}\right)
\ge m\exp\!\left(-\frac{\|x\|_2^2}{\underline\sigma^2}-\frac{R^2}{\underline\sigma^2}\right),
\]
where we used $(a+b)^2\le 2a^2+2b^2$.
Integrating over $G$ (a probability measure on $\Theta_R$) yields the same lower bound for $f_{G,\Sigma}(x)$, and taking logs gives
$\log f_{G,\Sigma}(x)\ge -C_0-C_1\|x\|_2^2$ for suitable constants $C_0,C_1$.

Finally, for any $\mathcal B\subseteq\mathcal G_R$, we have $f_{\mathcal B,\Sigma}(x)\le M$ and the same lower bound for $f_{G,\Sigma}(x)$, so
$\left|\log\{f_{\mathcal B,\Sigma}(x)/f_{G,\Sigma}(x)\}\right|$ is bounded by an affine function of $\|x\|_2^2$.
\end{proof}

\medskip 

\begin{lemma}[Identifiability of multivariate Gaussian location mixtures]
\label{lem:npmle-identifiability}
Fix any $\Sigma\succ 0$. If $f_{G_1,\Sigma}(x)=f_{G_2,\Sigma}(x)$ for Lebesgue-a.e.\ $x$, then $G_1=G_2$.
\end{lemma}

\begin{proof}
Let $\phi_{G}(t):=\int e^{it^\prime\theta}\,dG(\theta)$ be the characteristic function of $G$.
The Fourier transform of $f_{G,\Sigma}$ equals $\phi_G(t)\exp(-\tfrac12 t^\prime\Sigma t)$.
If $f_{G_1,\Sigma}=f_{G_2,\Sigma}$ a.e., then their Fourier transforms coincide, so
$\phi_{G_1}(t)=\phi_{G_2}(t)$ for all $t$, hence $G_1=G_2$. 
\end{proof}

\medskip

\begin{lemma}[Local uniform negativity]
\label{lem:npmle-local-negativity}
Fix $G_0\in\mathcal G_R$ with $G_0\neq G$. Then there exist $r_0>0$ and $c_0>0$ such that, for the open ball
$\mathcal B:=\{G'\in\mathcal G_R:\ d_{\mathrm{BL}}(G',G_0)<r_0\}$,
\[
\sup_{\Sigma\in\mathcal S}\ \mathbb E_{X\sim f_{G,\Sigma}}
\left[\log\frac{f_{\mathcal B,\Sigma}(X)}{f_{G,\Sigma}(X)}\right]
\ \le\ -c_0\ <\ 0.
\]
\end{lemma}

\begin{proof}
For each $\Sigma\in\mathcal S$,
\[
\mathbb E_{X\sim f_{G,\Sigma}}\!\left[\log\frac{f_{G_0,\Sigma}(X)}{f_{G,\Sigma}(X)}\right]
= -\mathrm{KL}\!\left(f_{G,\Sigma}\,\|\,f_{G_0,\Sigma}\right)
<0,
\]
where strict negativity follows from Lemma~\ref{lem:npmle-identifiability}.
Define the uniform KL-gap
\[
\kappa_0:=\inf_{\Sigma\in\mathcal S}\mathrm{KL}\!\left(f_{G,\Sigma}\,\|\,f_{G_0,\Sigma}\right).
\]
We claim $\kappa_0>0$.
If not, there exists $\Sigma_m\in\mathcal S$ with $\mathrm{KL}(f_{G,\Sigma_m}\|f_{G_0,\Sigma_m})\downarrow 0$.
By compactness of $\mathcal S$, along a subsequence $\Sigma_m\to\Sigma_\infty\in\mathcal S$.
By Lemma~\ref{lem:npmle-envelope}, the log-likelihood ratios admit a uniform square-integrable envelope under $X\sim f_{G,\Sigma_m}$, and
$f_{G,\Sigma_m}(x)\to f_{G,\Sigma_\infty}(x)$ and $f_{G_0,\Sigma_m}(x)\to f_{G_0,\Sigma_\infty}(x)$ pointwise in $x$.
Therefore, by dominated convergence and lower semicontinuity of KL,
$\mathrm{KL}(f_{G,\Sigma_\infty}\|f_{G_0,\Sigma_\infty})=0$, hence
$f_{G,\Sigma_\infty}=f_{G_0,\Sigma_\infty}$ a.e., contradicting Lemma~\ref{lem:npmle-identifiability}.
Thus $\kappa_0>0$.

Now, for $r>0$ let $\mathcal B_r:=\{G':d_{\mathrm{BL}}(G',G_0)<r\}$ and define
\[
H(r,\Sigma):=\mathbb E_{X\sim f_{G,\Sigma}}\!\left[\log\frac{f_{\mathcal B_r,\Sigma}(X)}{f_{G,\Sigma}(X)}\right].
\]
Since $\mathcal B_r\downarrow\{G_0\}$ as $r\downarrow 0$ and $G'\mapsto f_{G',\Sigma}(x)$ is continuous for fixed $(\Sigma,x)$,
we have $f_{\mathcal B_r,\Sigma}(x)\downarrow f_{G_0,\Sigma}(x)$ pointwise in $x$.
Lemma~\ref{lem:npmle-envelope} provides an integrable envelope for the log-ratio under $X\sim f_{G,\Sigma}$, uniformly over $\Sigma\in\mathcal S$,
so dominated convergence implies that for each fixed $\Sigma$,
\[
\lim_{r\downarrow 0} H(r,\Sigma)
=
\mathbb E_{X\sim f_{G,\Sigma}}\!\left[\log\frac{f_{G_0,\Sigma}(X)}{f_{G,\Sigma}(X)}\right]
\le -\kappa_0.
\]
If there were no $r_0$ with $\sup_{\Sigma\in\mathcal S}H(r_0,\Sigma)\le -\kappa_0/2$, then we could find $r_m\downarrow 0$ and $\Sigma_m\in\mathcal S$
such that $H(r_m,\Sigma_m)>-\kappa_0/2$ for all $m$.
By compactness, along a subsequence $\Sigma_m\to\Sigma_\infty$, and dominated convergence would give
$H(r_m,\Sigma_m)\to \mathbb E_{X\sim f_{G,\Sigma_\infty}}[\log(f_{G_0,\Sigma_\infty}(X)/f_{G,\Sigma_\infty}(X))]\le -\kappa_0$,
a contradiction.
Hence there exists $r_0>0$ such that $\sup_{\Sigma\in\mathcal S}H(r_0,\Sigma)\le -\kappa_0/2$.
Setting $\mathcal B:=\mathcal B_{r_0}$ and $c_0:=\kappa_0/2$, we complete the proof. 
\end{proof}

\noindent \textbf{Proof of Proposition~\ref{prop:npmle-weak-consistency}}

For any subset $\mathcal B\subseteq\mathcal G_R$ and any $\Sigma\in\mathcal S$, define
$f_{\mathcal B,\Sigma}(x):=\sup_{G\in\mathcal B} f_{G,\Sigma}(x)$, and for each $i$ define
\[
Z_{i}(\mathcal B):=\log\frac{f_{\mathcal B,\Sigma_i}(\hat\theta_i)}{f_{G,\Sigma_i}(\hat\theta_i)}.
\]
Then, for any $\mathcal B$,
\[
\sup_{G'\in\mathcal B}\big(\ell_n(G')-\ell_n(G)\big)
=
\sup_{G'\in\mathcal B}\frac1n\sum_{i=1}^n \log\frac{f_{G',\Sigma_i}(\hat\theta_i)}{f_{G,\Sigma_i}(\hat\theta_i)}
\ \le\
\frac1n\sum_{i=1}^n Z_i(\mathcal B),
\]
because $\sup_{G'\in\mathcal B}\log f_{G',\Sigma_i}(\hat\theta_i)=\log f_{\mathcal B,\Sigma_i}(\hat\theta_i)$.

Then fix $\varepsilon>0$ and define the closed set
\[
\mathcal A_\varepsilon:=\{G'\in\mathcal G_R:\ d_{\mathrm{BL}}(G',G)\ge \varepsilon\}.
\]
Since $\mathcal G_R$ is compact, $\mathcal A_\varepsilon$ is compact.
For each $G_0\in\mathcal A_\varepsilon$, Lemma~\ref{lem:npmle-local-negativity} provides an open ball $\mathcal B(G_0)$ and a constant $c(G_0)>0$ such that
\[
\sup_{\Sigma\in\mathcal S}\ \mathbb E_{X\sim f_{G,\Sigma}}
\left[\log\frac{f_{\mathcal B(G_0),\Sigma}(X)}{f_{G,\Sigma}(X)}\right]\le -c(G_0).
\]
The collection $\{\mathcal B(G_0):G_0\in\mathcal A_\varepsilon\}$ is an open cover of $\mathcal A_\varepsilon$, so by compactness there exist
$G_1,\dots,G_J\in\mathcal A_\varepsilon$ such that $\mathcal A_\varepsilon \subseteq \bigcup_{j=1}^J \mathcal B_j$ with $\mathcal B_j:=\mathcal B(G_j)$.
For each $j$ there is $c_j:=c(G_j)>0$ with
\begin{equation}
\label{eq:npmle-negativity-uniform}
\sup_{\Sigma\in\mathcal S}\ \mathbb E_{X\sim f_{G,\Sigma}}
\left[\log\frac{f_{\mathcal B_j,\Sigma}(X)}{f_{G,\Sigma}(X)}\right]\le -c_j.
\end{equation}

Fix $j\in\{1,\dots,J\}$ and define $Z_{i,j}:=Z_i(\mathcal B_j)$.
By \eqref{eq:npmle-negativity-uniform} and $\Sigma_i\in\mathcal S$, $\mathbb E[Z_{i,j}]\le -c_j$ for all $i$.
Moreover, since $\theta_i\in\Theta_R$ a.s.\ and $\underline\sigma^2 I_d\preceq \Sigma_i\preceq \overline\sigma^2 I_d$, we have
$\sup_i \mathbb E[\|\hat\theta_i\|_2^4]<\infty$, and Lemma~\ref{lem:npmle-envelope} implies $\sup_i\mathbb E[Z_{i,j}^2]<\infty$.
Therefore $\sup_i\mathrm{Var}(Z_{i,j})<\infty$, so $\sum_{i=1}^\infty \mathrm{Var}(Z_{i,j})/i^2<\infty$ and Kolmogorov's strong law yields
\[
\frac1n\sum_{i=1}^n \big(Z_{i,j}-\mathbb E[Z_{i,j}]\big)\ \to\ 0
\quad\text{almost surely.}
\]
Consequently, almost surely for all sufficiently large $n$,
\[
\frac1n\sum_{i=1}^n Z_{i,j}
\le
\frac1n\sum_{i=1}^n \mathbb E[Z_{i,j}] + \frac{c_j}{2}
\le -\frac{c_j}{2}.
\]
It follows that, almost surely for all large $n$,
\[
\sup_{G'\in\mathcal B_j}\big(\ell_n(G')-\ell_n(G)\big)
\le \frac1n\sum_{i=1}^n Z_{i,j}
\le -\frac{c_j}{2}
<0,
\]
so no maximizer of $\ell_n(\cdot)$ over $\mathcal G_R$ can lie in $\mathcal B_j$.
Since the cover is finite, this holds simultaneously for $j=1,\dots,J$ almost surely for all large $n$.
Therefore, almost surely for all large $n$, $\hat G_n\notin \bigcup_{j=1}^J \mathcal B_j\supseteq \mathcal A_\varepsilon$,
so $d_{\mathrm{BL}}(\hat G_n,G)<\varepsilon$ eventually almost surely.
Since $\varepsilon>0$ is arbitrary, $d_{\mathrm{BL}}(\hat G_n,G)\to 0$ almost surely, which is equivalent to $\hat G_n\Rightarrow G$. \qed

\subsection{Proof of Theorem~\ref{thm:gaussian_rootn}}
Fix $\eta\in\mathcal H$.
Under Assumption~\ref{as:gaussian_summary_lik}, $Y\mid \theta;\eta\sim \mathcal N(\theta,\Sigma(\eta))$.
By Lemma~\ref{lem:value_reduction},
\[
U_{\hat G_n}(\eta)-U_G(\eta)
=
\mathbb E_{Y\sim f_{\hat G_n,\eta}}\big[\Psi(\mu_{\hat G_n,\eta}(Y))\big]
-
\mathbb E_{Y\sim f_{G,\eta}}\big[\Psi(\mu_{G,\eta}(Y))\big],
\]
since the $w_0$ terms cancel. Add and subtract $\mathbb E_{Y\sim f_{G,\eta}}[\Psi(\mu_{\hat G_n,\eta}(Y))]$ and bound by $(\mathrm I)+(\mathrm{II})$ where
\[
(\mathrm I):=\Big|\mathbb E_{Y\sim f_{G,\eta}}\big[\Psi(\mu_{\hat G_n,\eta}(Y))-\Psi(\mu_{G,\eta}(Y))\big]\Big|,
\quad
(\mathrm{II}):=\Big|\mathbb E_{Y\sim f_{\hat G_n,\eta}}[g_n(Y)]-\mathbb E_{Y\sim f_{G,\eta}}[g_n(Y)]\Big|
\]
with $g_n(y):=\Psi(\mu_{\hat G_n,\eta}(y))$.

For $(\mathrm I)$, Assumption~\ref{as:Psi_Lip} and Cauchy--Schwarz yield
\[
(\mathrm I)\le L\Big(\E_{f_{G,\eta}}\big[(1+\|\mu_{\hat G_n,\eta}(Y)\|_2^p+\|\mu_{G,\eta}(Y)\|_2^p)^2\big]\Big)^{1/2}
\Big(\E_{f_{G,\eta}}\big[\|\mu_{\hat G_n,\eta}(Y)-\mu_{G,\eta}(Y)\|_2^2\big]\Big)^{1/2}.
\]
Under Gaussian conjugacy, for $Q=\mathcal N(\tau,V)$ and $\Sigma=\Sigma(\eta)$,
\[
\mu_{Q,\eta}(y)=\tau+V(V+\Sigma)^{-1}(y-\tau)=:A_{V,\eta}y+(I-A_{V,\eta})\tau,
\qquad A_{V,\eta}:=V(V+\Sigma)^{-1}.
\]
Because $\underline\sigma^2 I_d\preceq \Sigma(\eta)\preceq \overline\sigma^2 I_d$, the predictive covariance under $G$ is
$V_0+\Sigma(\eta)\preceq (\lambda_{\max}(V_0)+\overline\sigma^2)I_d$, so $Y\sim f_{G,\eta}$ has Gaussian moments of all orders uniformly in $\eta$.
Since $\mu_{Q,\eta}(y)$ is affine in $y$ and $\|A_{V,\eta}\|_{\mathrm{op}}\le 1$, these moment bounds carry over to
$\mu_{G,\eta}(Y)$ and $\mu_{\hat G_n,\eta}(Y)$ (as $\hat\tau_n,\hat V_n$ stay in a neighborhood of $\tau_0,V_0$).
Hence the first factor above is $O_p(1)$ uniformly in $\eta$.

Moreover, writing $A_0:=A_{V_0,\eta}$ and $\hat A:=A_{\hat V_n,\eta}$,
\[
\mu_{\hat G_n,\eta}(y)-\mu_{G,\eta}(y)=(I-\hat A)(\hat\tau_n-\tau_0)+(\hat A-A_0)(y-\tau_0).
\]
Since $\|(V+\Sigma(\eta))^{-1}\|_{\mathrm{op}}\le 1/\underline\sigma^2$ uniformly in $\eta$,
one can obtain
$
\sup_{\eta\in\mathcal H}\|\hat A-A_0\|_{\mathrm{op}}\ \le\ C_1\|\hat V_n-V_0\|_{\mathrm F}.
$
Combining with $\|I_d-A_{\hat V_n,\eta}\|_{\mathrm{op}}\le \|\Sigma(\eta)\|_{\mathrm{op}}\cdot\|(\hat V_n+\Sigma(\eta))^{-1}\|_{\mathrm{op}}\le C_2$, we bound
\[
\|\mu_{\hat G_n,\eta}(y)-\mu_{G,\eta}(y)\|_2
\le
C(1+\|y\|_2)\big(\|\hat\tau_n-\tau_0\|_2+\|\hat V_n-V_0\|_{\mathrm F}\big),
\]
uniformly over $\eta\in\mathcal H$, and thus $(\mathrm I)=O_p(n^{-1/2})$ by Assumption~\ref{as:gauss_est_rate}.

For $(\mathrm{II})$, the Hellinger inequality gives
\[
(\mathrm{II})\le 2\,h(f_{\hat G_n,\eta},f_{G,\eta})\,
\Big(\E_{f_{\hat G_n,\eta}}[g_n(Y)^2]+\E_{f_{G,\eta}}[g_n(Y)^2]\Big)^{1/2},
\qquad g_n(y)=\Psi(\mu_{\hat G_n,\eta}(y)).
\]
To keep the squared-root term $O_p(1)$ uniformly in $\eta$, we need
$\E_{f_{\hat G_n,\eta}}[g_n(Y)^2]$ and $\E_{f_{G,\eta}}[g_n(Y)^2]$ uniformly bounded.
Assumption~\ref{as:Psi_Lip} implies $|\Psi(m)|\le C(1+\|m\|_2^{p+1})$, and $\mu_{\hat G_n,\eta}(y)$ is affine in $y$ with uniformly bounded coefficients,
so $g_n(Y)^2\le C(1+\|Y\|_2^{2(p+1)})$. Since $Y$ is Gaussian under both $f_{G,\eta}=\mathcal N(\tau_0,V_0+\Sigma(\eta))$ and
$f_{\hat G_n,\eta}=\mathcal N(\hat\tau_n,\hat V_n+\Sigma(\eta))$, and the covariance eigenvalues are uniformly bounded over $\eta$,
these expectations are $O_p(1)$ uniformly in $\eta$.

Finally, the Hellinger distance between Gaussians is a smooth (hence locally Lipschitz) function of their mean/covariance parameters on this uniformly
nondegenerate covariance class, so
\[
\sup_{\eta\in\mathcal H} h(f_{\hat G_n,\eta},f_{G,\eta})
\le C\big(\|\hat\tau_n-\tau_0\|_2+\|\hat V_n-V_0\|_{\mathrm F}\big)
=
O_p(n^{-1/2})
\]
by Assumption~\ref{as:gauss_est_rate}. Consequently $(\mathrm{II})=O_p(n^{-1/2})$.

Combining (I) and (II) yields $\sup_{\eta\in\mathcal H}|U_{\hat G_n}(\eta)-U_G(\eta)|=O_{p}(n^{-1/2})$.
The regret rate follows from Theorem~\ref{thm:oracle-ineq}. \qed

\subsection{Proof of Theorem~\ref{thm:npmle_rootn}}
Fix $\eta\in\mathcal H$ and write $\Sigma=\Sigma(\eta)$.
By Lemma~\ref{lem:value_reduction},
\[
U_{\hat G_n}(\eta)-U_G(\eta)
=
\mathbb E_{Z\sim f_{\hat G_n,\Sigma}}[\Psi(\mu_{\hat G_n,\Sigma}(Z))]
-
\mathbb E_{Z\sim f_{G,\Sigma}}[\Psi(\mu_{G,\Sigma}(Z))].
\]
Add and subtract $\mathbb E_{Z\sim f_{G,\Sigma}}[\Psi(\mu_{\hat G_n,\Sigma}(Z))]$ and decompose into $(\mathrm I)+(\mathrm{II})$.

Since $G$ is supported on $\left\{\theta\in\mathbb R^d:\ \|\theta\|_2\le R\right\}$, posterior means satisfy $$\|\mu_{G,\Sigma}(z)\|_2,\,\|\mu_{\hat G_n,\Sigma}(z)\|_2\le R.$$
On this compact set, Assumption~\ref{as:Psi_Lip} implies a finite global Lipschitz constant $L_R$, so
\[
(\mathrm I)
\le
L_R\sqrt{\E\big[\|\mu_{\hat G_n,\Sigma}(Z)-\mu_{G,\Sigma}(Z)\|_2^2\big]}
\le
L_R\sqrt{C_2}\,\frac{(\log n)^{(d+\max(d/2,4))/2}}{\sqrt n}
\]
by Assumption~\ref{as:npmle_inputs}(ii).

For $(\mathrm{II})$, $\Psi(\mu_{\hat G_n,\Sigma}(Z))$ is bounded (compact support), so the Hellinger inequality and Assumption~\ref{as:npmle_inputs}(i) yield
$(\mathrm{II})=O_{p}\left((\log n)^{(d+1)/2}/\sqrt n\right)$.
Thus
\[
\sup_{\eta\in\mathcal H}|U_{\hat G_n}(\eta)-U_G(\eta)|
=
O_{p}\left(\frac{(\log n)^{(d+\max(d/2,4))/2}}{\sqrt n}\right),
\]
where the uniformity over $\eta$ follows from the uniform-over-$\Sigma$ statements in Assumption~\ref{as:npmle_inputs} together with the eigenvalue bounds in Assumption~\ref{as:gaussian_summary_lik}.
The regret rate follows from Theorem~\ref{thm:oracle-ineq}. \qed

\subsection{Proof of Theorem~\ref{thm:second_order_general}}
Under Assumption~\ref{as:strong_concavity}, for any $\eta$ and $\eta^O$,
\[
U_G(\eta^O) \le U_G(\eta) + \langle \nabla U_G(\eta),\eta^O-\eta\rangle
-\frac{m}{2}\|\eta^O-\eta\|_2^2.
\]
Rearranging and applying Cauchy--Schwarz yields
\[
U_G(\eta^O)-U_G(\eta)
\le \|\nabla U_G(\eta)\|_2\,\, \|\eta^O-\eta\|_2
-\frac{m}{2}\|\eta^O-\eta\|_2^2.
\]
Let $t=\|\eta^O-\eta\|_2$ and $a=\|\nabla U_G(\eta)\|_2$. Then
\[
U_G(\eta^O)-U_G(\eta)\le a t-\frac{m}{2}t^2
= -\frac{m}{2}\Bigl(t-\frac{a}{m}\Bigr)^2+\frac{a^2}{2m}
\le \frac{a^2}{2m},
\]
i.e. $U_G(\eta^O)-U_G(\eta)\le \|\nabla U_G(\eta)\|_2^2/(2m)$.

Set $\eta=\widehat\eta$. Since $\widehat\eta$ is an interior maximizer of $U_{\widehat G_n}$, $\nabla U_{\widehat G_n}(\widehat\eta)=0$, hence
$$\|\nabla U_G(\widehat\eta)\|_2=\|\nabla U_G(\widehat\eta)-\nabla U_{\widehat G_n}(\widehat\eta)\|_2\le r_n.$$
Substitute to get $\Reg_n\le r_n^2/(2m)$. \qed

\subsection{Proof of Theorem~\ref{thm:blackwell-gaussian}}
If $\Sigma(\eta_1)\preceq \Sigma(\eta_2)$, then $\Sigma(\eta_2)-\Sigma(\eta_1)$ is positive semidefinite. Let
$Z\sim \mathcal N(0,\Sigma(\eta_2)-\Sigma(\eta_1))$ independent of $(\theta,\varepsilon_{\eta_1})$. Then
$Y_{\eta_2}\overset{d}{=}Y_{\eta_1}+Z$, so $\eta_2$ is a garbling of $\eta_1$. Blackwell's theorem \citep{Blackwell1951} implies
$U_Q(\eta_1)\ge U_Q(\eta_2)$ for every decision problem and prior. Gaussian conjugacy yields the posterior covariance formula. \qed

\subsection{Proof of Lemma~\ref{lem:reduce-to-scalar}}
Fix $\eta$ and write $\Sigma=\Sigma(\eta)$. Let $\theta\sim\mathcal N(\mu,V)$ and $Y_\eta=\theta+\varepsilon$
with $\varepsilon\sim\mathcal N(0,\Sigma)$ independent of $\theta$. For a given $\alpha\in\R^d$, define the
payoff-relevant index $\omega=\alpha'\theta$, with prior mean $\mu_\omega=\alpha'\mu$ and variance
$\sigma_\omega^2=\Var(\omega)=\alpha'V\alpha$.

Because $(\omega,Y_\eta)$ is jointly Gaussian, the posterior is Gaussian and can be written as
\[
\omega\mid Y_\eta \sim \mathcal N\!\big(Z_\eta^0,\ s_\eta^2\big),
\qquad
Z_\eta^0:=\E[\omega\mid Y_\eta],
\qquad
s_\eta^2:=\Var(\omega\mid Y_\eta)=\alpha'(V^{-1}+\Sigma^{-1})^{-1}\alpha.
\]
In particular, $s_\eta^2$ depends on $\eta$ only through $\Sigma(\eta)$ and does not depend on the realization
of $Y_\eta$.

Moreover, for jointly Gaussian variables the prediction error is independent of the signal: letting
$e_\eta:=\omega-Z_\eta^0$, we have $e_\eta\sim\mathcal N(0,s_\eta^2)$ and $e_\eta\perp Y_\eta$ (hence also
$e_\eta\perp Z_\eta^0$). This implies $\omega\perp Y_\eta\mid Z_\eta^0$, so $Z_\eta^0$ is sufficient for $\omega$.
Therefore $Y_\eta$ and $Z_\eta^0$ are Blackwell equivalent for the state $\omega$.

To obtain the canonical ``signal plus independent noise'' form, consider the jointly Gaussian pair
$(\omega,Z_\eta^0)$ and write the linear regression of $Z_\eta^0$ on $\omega$:
\[
Z_\eta^0=\mu_\omega+\kappa_\eta(\omega-\mu_\omega)+\zeta_\eta,
\qquad
\zeta_\eta\perp \omega,\ \ \zeta_\eta\ \text{Gaussian}.
\]
Using $\Cov(Z_\eta^0,\omega)=\Var(Z_\eta^0)$ (since $Z_\eta^0=\E[\omega\mid Y_\eta]$) and the law of total
variance $\Var(\omega)=\Var(Z_\eta^0)+s_\eta^2$, we obtain
\[
\kappa_\eta=\frac{\Var(Z_\eta^0)}{\Var(\omega)}=1-\frac{s_\eta^2}{\sigma_\omega^2}.
\]
If $\kappa_\eta=0$, then $Z_\eta^0$ is constant and the experiment is uninformative about $\omega$.

Now apply the invertible affine map $h_\eta(z):=\mu_\omega+(z-\mu_\omega)/\kappa_\eta$ and define the scalar
statistic
\[
Z_\eta := h_\eta(Z_\eta^0)=\mu_\omega+\frac{Z_\eta^0-\mu_\omega}{\kappa_\eta}.
\]
Substituting the regression yields
\[
Z_\eta=\omega+\xi_\eta,
\qquad
\xi_\eta:=\zeta_\eta/\kappa_\eta\ \perp\ \omega,
\qquad
\xi_\eta\sim\mathcal N\!\big(0,\tau^2(\eta)\big),
\]
with
\[
\tau^2(\eta)=\frac{\Var(\zeta_\eta)}{\kappa_\eta^2}
=\frac{1-\kappa_\eta}{\kappa_\eta}\,\sigma_\omega^2
=\frac{\sigma_\omega^2\,s_\eta^2}{\sigma_\omega^2-s_\eta^2}.
\]
Thus $\tau^2(\eta)$ depends on $\eta$ only through $\Sigma(\eta)$ (via $s_\eta^2$). Since $h_\eta$ is bijective,
$Z_\eta$ is Blackwell equivalent to $Z_\eta^0$, hence also to $Y_\eta$, for the state $\omega$. \qed

\subsection{Proof of Corollary~\ref{cor:scalar-blackwell}}

By Lemma~4, each feasible design $\eta\in\mathcal H$ induces a Blackwell-equivalent scalar normal experiment
$Z_\eta=\omega+\xi_\eta$ with $\xi_\eta\sim\mathcal N(0,\tau^2(\eta))$ independent of $\omega$.
From Theorem \ref{thm:blackwell-gaussian}, for a scalar normal experiment, the optimal design should minimize $\tau^2(\eta).$ Since
$
\tau^2(\eta)= \sigma_\omega^2\,s_\eta^2/(\sigma_\omega^2-s_\eta^2)
$
is increasing in $s_\eta^2$, it is equivalent to minimize $s_\eta^2=\Var(\omega\mid Y_\eta).$ \qed

\section{Additional Results for the Empirical Applications}
\subsection{Drug Trials}
\label{app:drug_more}

\runin{Detailed outcome and subgroup harmonization rules}
Outcome families are assigned by outcome-title keywords using deterministic rules: \emph{overall survival} if the title contains ``overall survival'' or ``OS''; \emph{progression-free survival} if it contains ``progression-free survival'' or ``PFS''; \emph{tumor response} if it contains ``objective response,'' ``overall response,'' or ``ORR''; and \emph{other efficacy} otherwise. 

PD-L1 subgroup labels are harmonized into two sides, high-expression and low-expression, using three sequential rules. First, if a row already has an explicit coarse label, we map positive/high-like labels to the high-expression side and negative/low-like labels to the low-expression side. Second, if a row is coded by a structured threshold (for example, ``CPS\(\ge\) 10" or ``TPS\(\le\) 1"), we map it by the inequality direction: \(\ge\) or \(>\) implies high-expression, while \(\le\) or \(<\) implies low-expression. Third, we map pathology staining classes: TC1/2/3 and IC1/2/3 (detectable  tumor-cell or immune-cell PD-L1 staining) to high-expression, and IC0 (no detectable immune-cell PD-L1 staining) to low-expression. If a label still cannot be assigned after these rules, we drop it. Table \ref{tab:drug_harmonization_examples} provides several examples in the study pool.

To address within-study dependence, we retain at most one estimate for each study--subgroup pair. When only one eligible estimate is available, it is retained. When multiple eligible estimates are reported and at least two distinct cutoff classes appear for the same subgroup, we use a consensus-cutoff rule: for each subgroup and scoring-system (CPS, TPS, or percent threshold), we construct a cross-study reference cutoff as the median of study-level median cutoffs, and retain the estimate whose cutoff is closest (in absolute distance) to that reference; ties are broken at random. In all other multiple-estimate cases (for example, multiple treatment-control comparisons in multi-arm trials, repeated estimates under different follow-up periods or sample definitions), we randomly retain one estimate using a fixed seed.

The resulting prior-study dataset for the PFS outcome family includes 57 estimates from 53 studies (4 studies report both high and low, 44 report high only, and 5 report low only).

\begin{table}[t!]
\caption{Examples of raw PD-L1 subgroup labels and harmonized sides}
\begin{center}
\begin{tabular}{lp{0.4\textwidth}l}
\toprule
NCT ID & Example raw subgroup label & Harmonized side \\
\midrule
NCT01704287 & PD-L1 Positive & High-expression \\
NCT01721746 & PD-L1 Negative & Low-expression \\
NCT01844505 & PD-L1 $<1\%$ & Low-expression \\
NCT02041533 & PD-L1 expression $\ge 5\%$ & High-expression \\
NCT02555657 & PD-L1 CPS $\ge 1$ & High-expression \\
NCT01984242 & IC1/2/3 population & High-expression \\
\bottomrule
\end{tabular}
\end{center}
\label{tab:drug_harmonization_examples}
{\footnotesize {\em Notes}: Examples are drawn from the study pool used in this application.}
\end{table}

Figures \ref{fig:drug_os_estimates} and \ref{fig:drug_pfs_estimates} plot the selected prior-study estimates under the OS and PFS outcome families, respectively.
\begin{figure}[t!]
\caption{OS prior-study estimates by PD-L1 subgroup}
\begin{center}
\begin{tabular}{cc}
\includegraphics[width=0.47\textwidth]{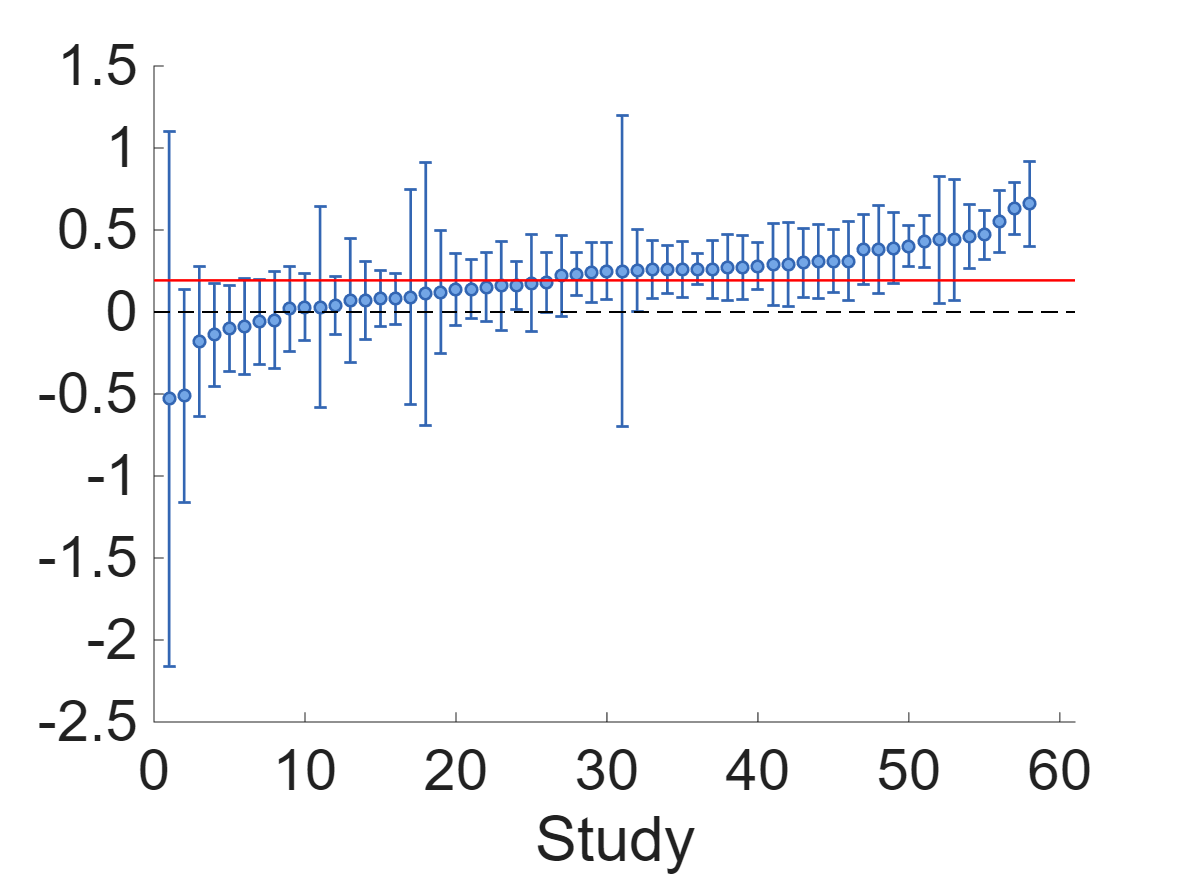} &
\includegraphics[width=0.47\textwidth]{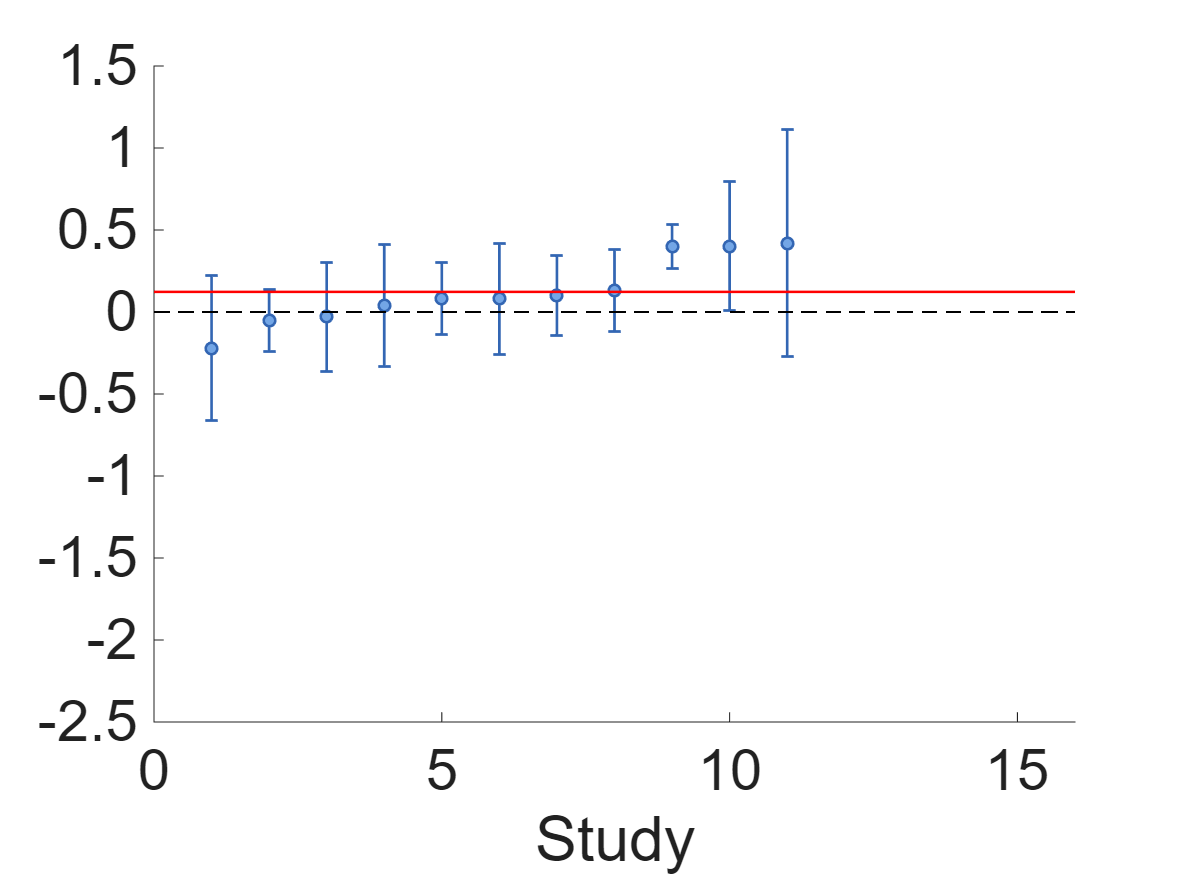}
\end{tabular}
\end{center}
\label{fig:drug_os_estimates}
{\footnotesize {\em Notes}: Each panel plots study-level estimates with 95\% confidence intervals ($\pm 1.96\times\mathrm{SE}$). Left panel: high-expression subgroup; right panel: low-expression subgroup. Studies are sorted by estimates within subgroup. The red horizontal line is the subgroup sample mean estimate, and the dashed black line is the zero benchmark.}
\end{figure}

\begin{figure}[htbp]
\caption{PFS prior-study estimates by PD-L1 subgroup}
\begin{center}
\begin{tabular}{cc}
\includegraphics[width=0.47\textwidth]{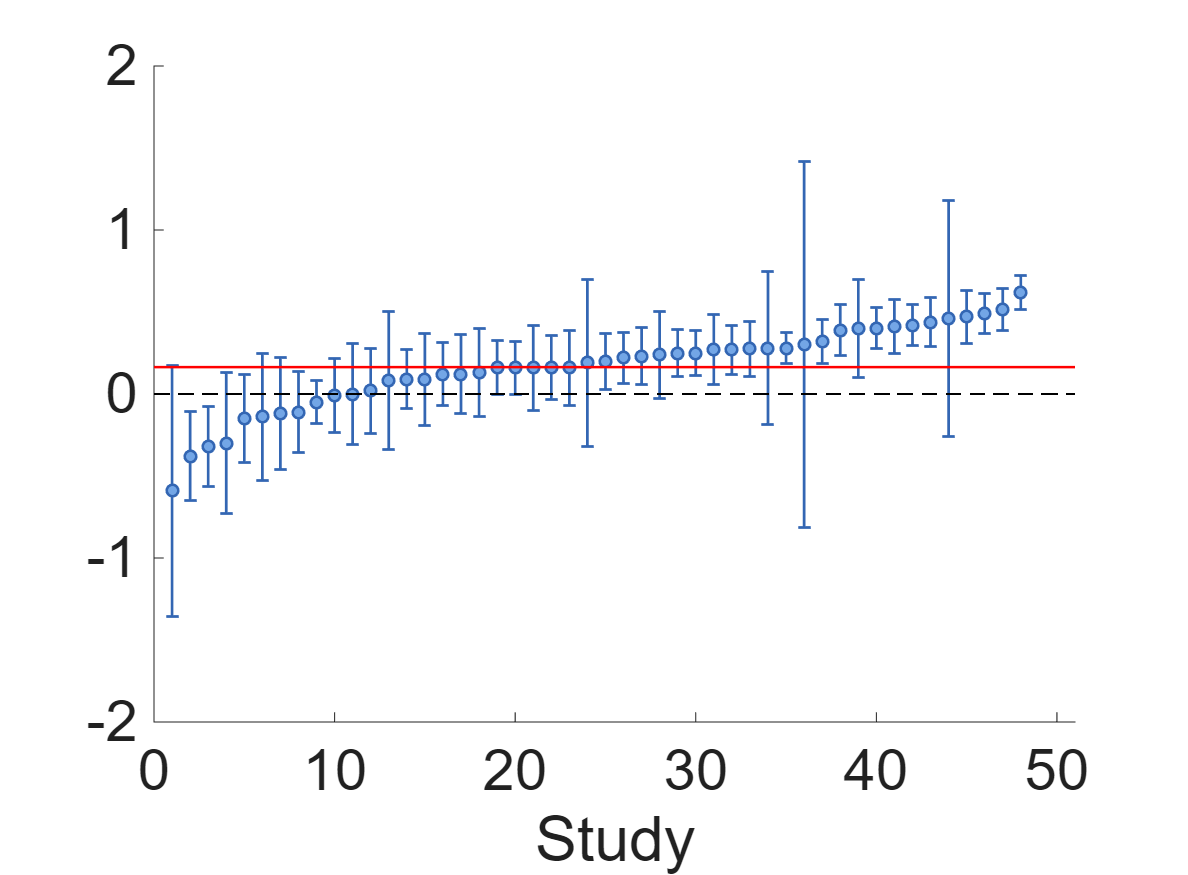} &
\includegraphics[width=0.47\textwidth]{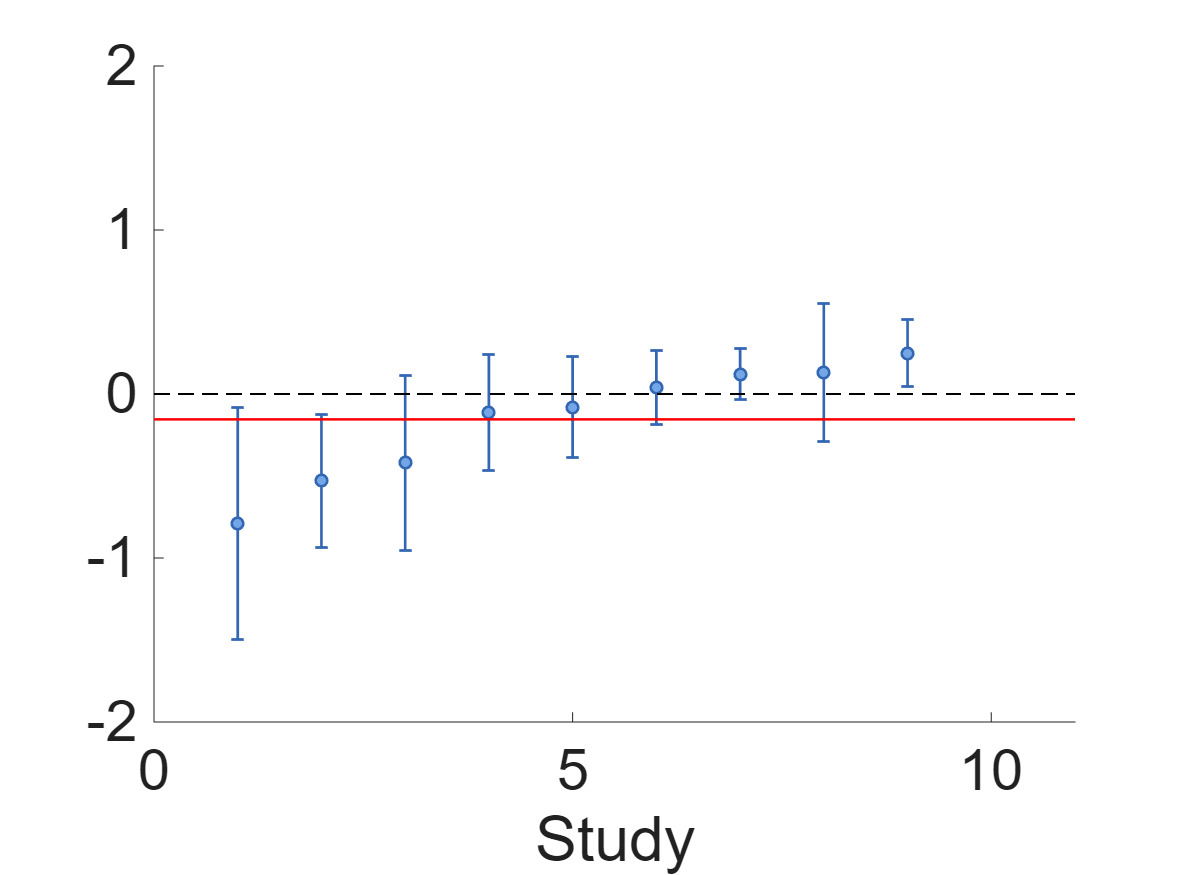}
\end{tabular}
\end{center}
\label{fig:drug_pfs_estimates}
{\footnotesize {\em Notes}: Each panel plots study-level estimates with 95\% confidence intervals ($\pm 1.96\times\mathrm{SE}$). Left panel: high-expression subgroup; right panel: low-expression subgroup. Studies are sorted by estimates within subgroup. The red horizontal line is the subgroup sample mean estimate, and the dashed black line is the zero benchmark.}
\end{figure}

\clearpage
\runin{Additional results}
\begin{figure}[htbp]
\caption{OS estimated prior marginals (independent)}
\begin{center}
\begin{tabular}{cc}
\includegraphics[width=0.42\textwidth]{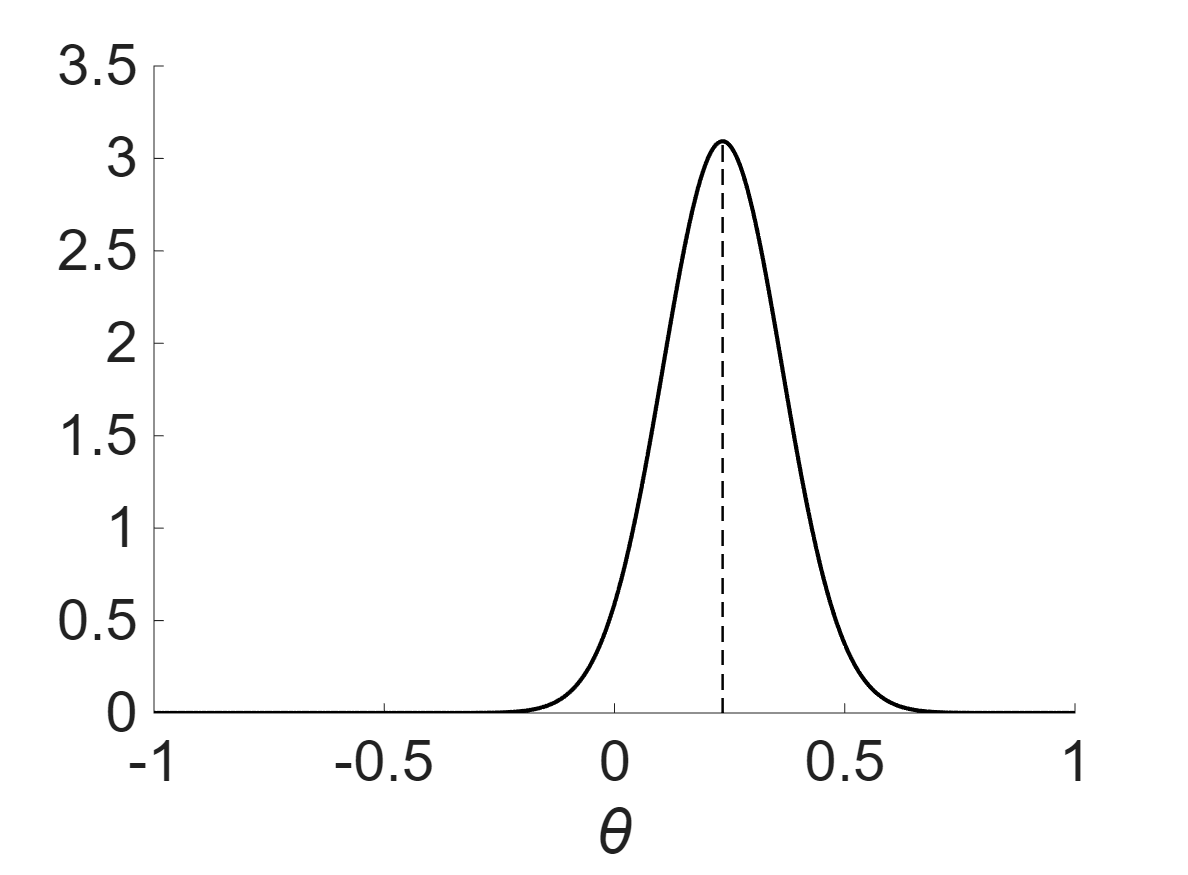} &
\includegraphics[width=0.42\textwidth]{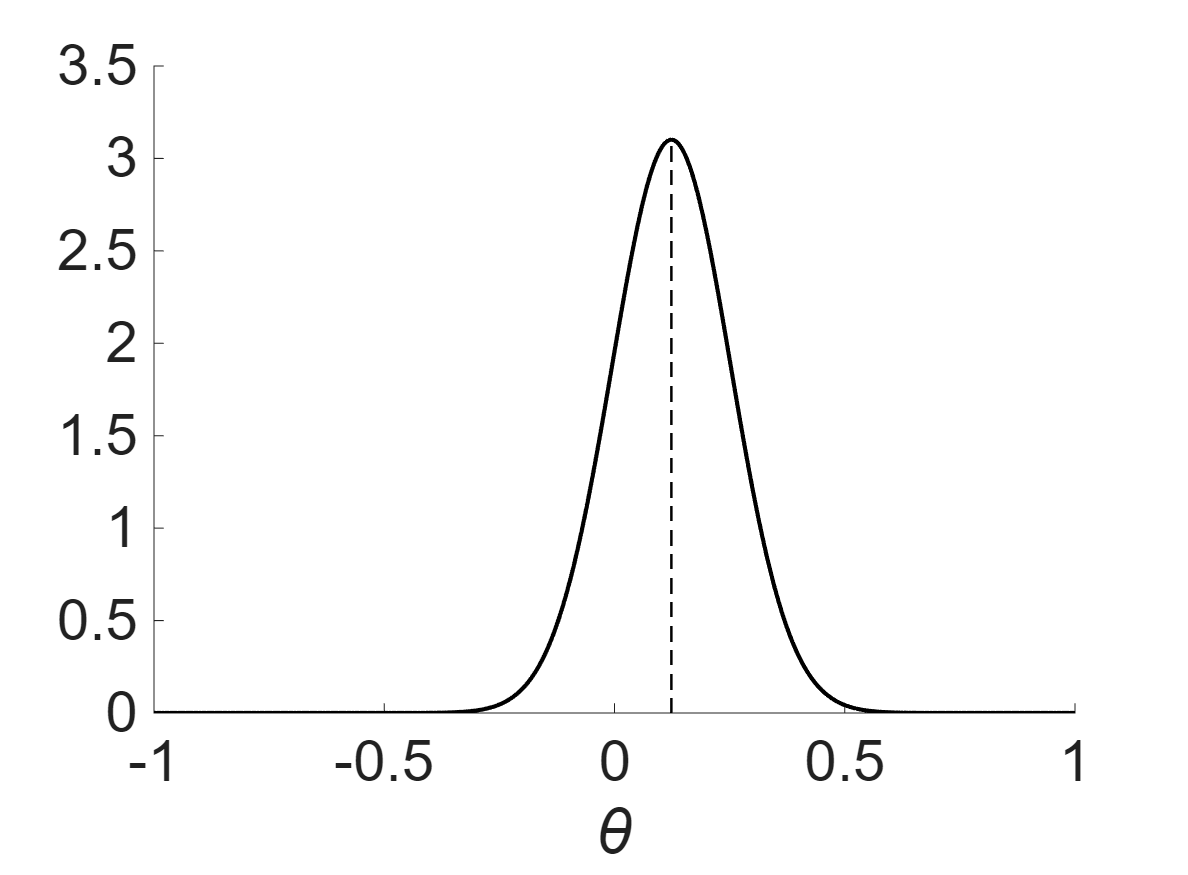} \\
\includegraphics[width=0.42\textwidth]{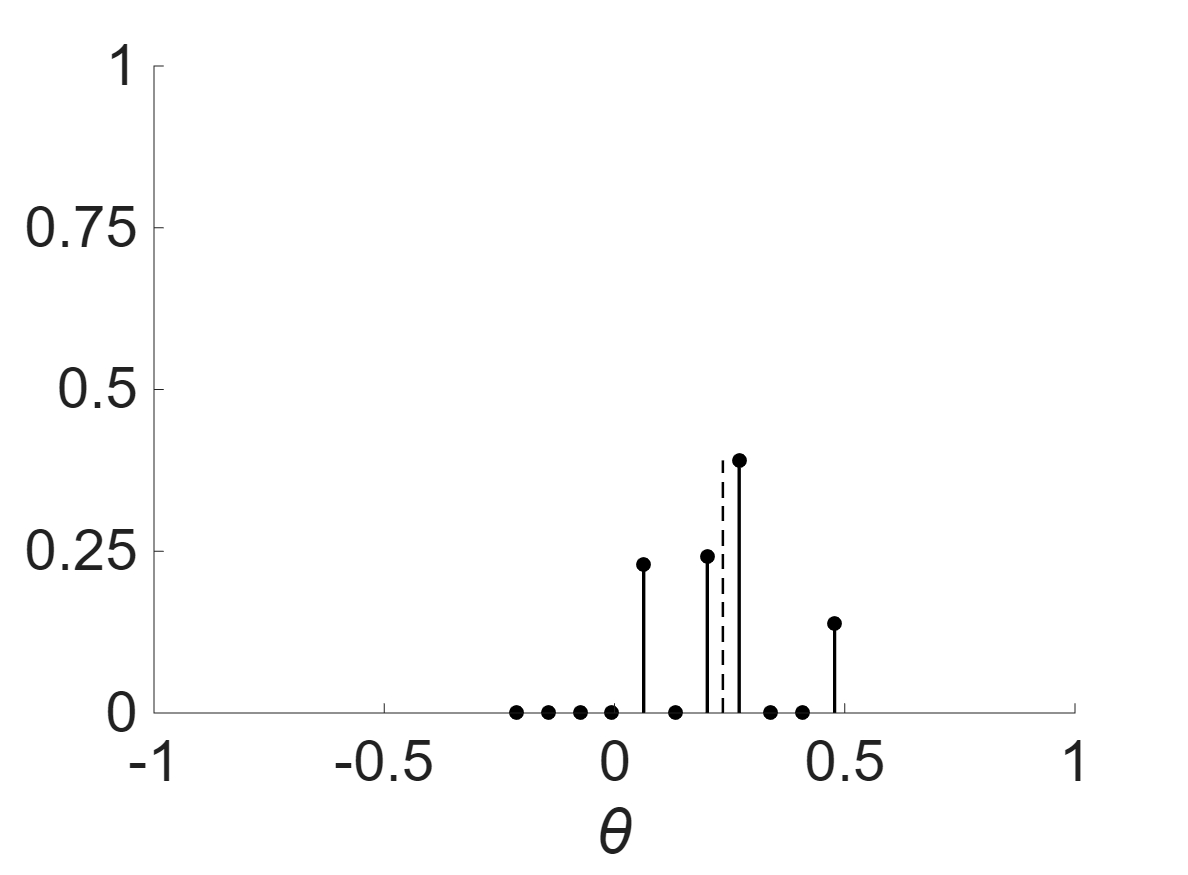} &
\includegraphics[width=0.42\textwidth]{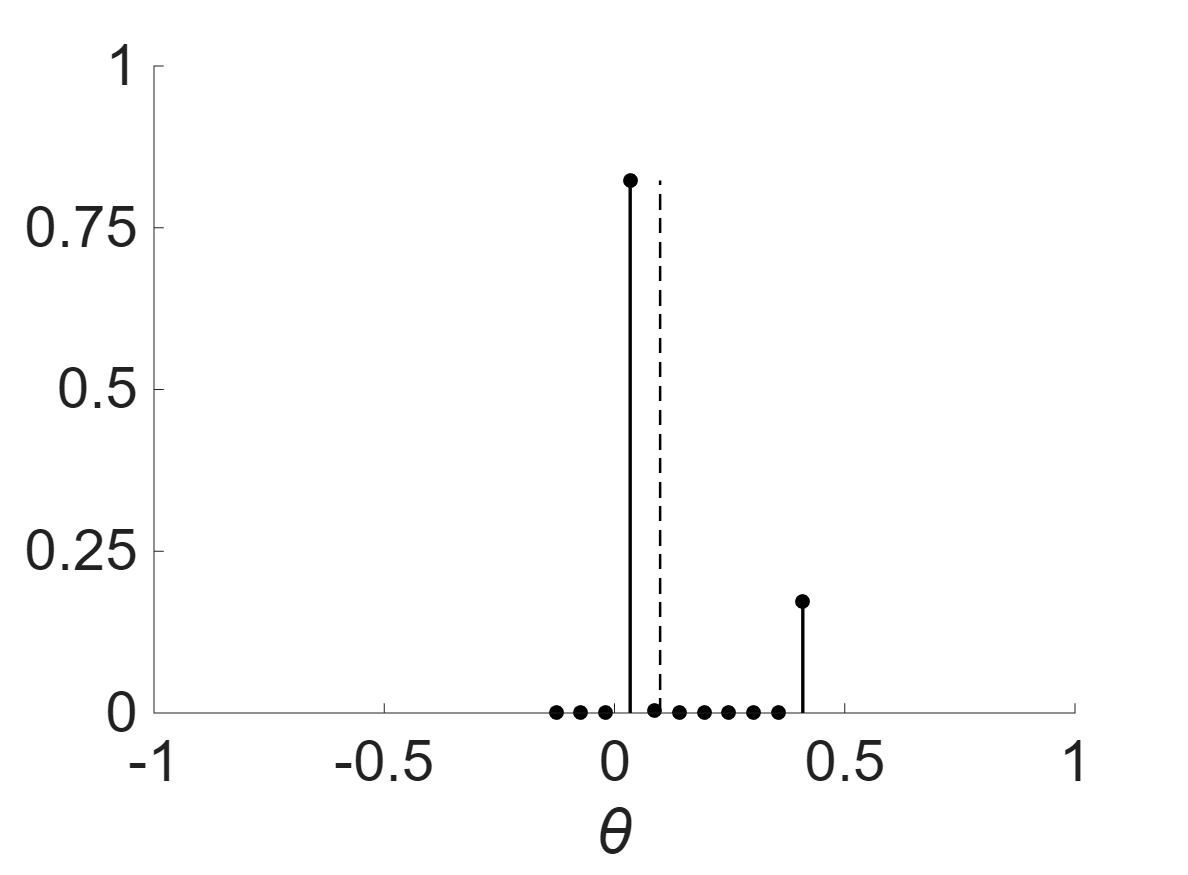}
\end{tabular}
\end{center}
\label{fig:drug_os_priors_diag}
{\footnotesize {\em Notes}: Top row: Gaussian independent-subgroup prior marginals for PD-L1 high-expression and low-expression effects. Bottom row: NPMLE independent-subgroup prior marginals for the same two effects.}
\end{figure}

\begin{table}[htbp]
\caption{Objective-I optimal design for OS (independent priors)}
\begin{center}
\begin{tabular}{lcc}
\toprule
Prior & $e_{\text{high}}$ & $e_{\text{low}}$ \\
\midrule
Gaussian-independent EB & 0.250 & 0.250 \\
NPMLE-independent EB & 0.233 & 0.267 \\
No-information & 0.250 & 0.250 \\
\bottomrule
\end{tabular}
\end{center}
\label{tab:drug_os_design_2x2_diag}
\end{table}

\begin{figure}[htbp]
\caption{PFS estimated prior marginals (joint)}
\begin{center}
\begin{tabular}{cc}
\includegraphics[width=0.42\textwidth]{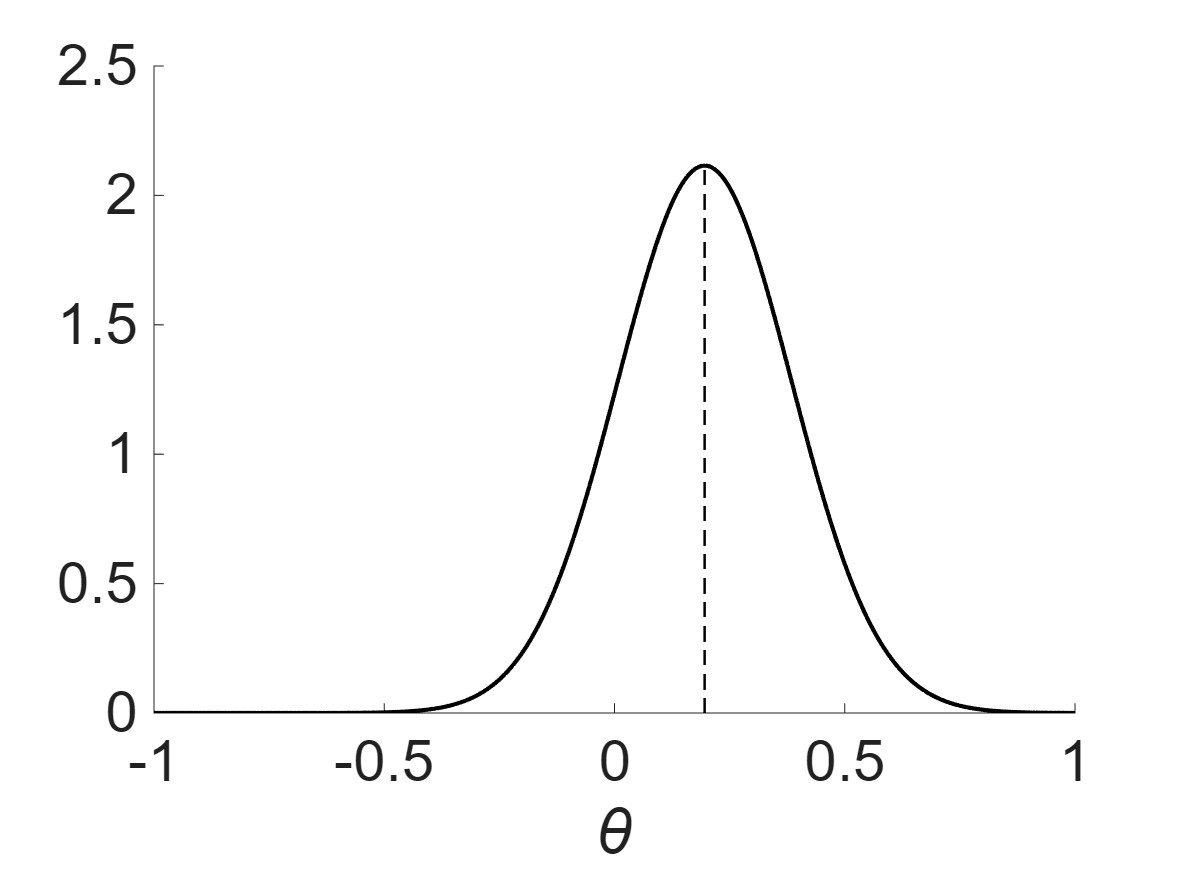} &
\includegraphics[width=0.42\textwidth]{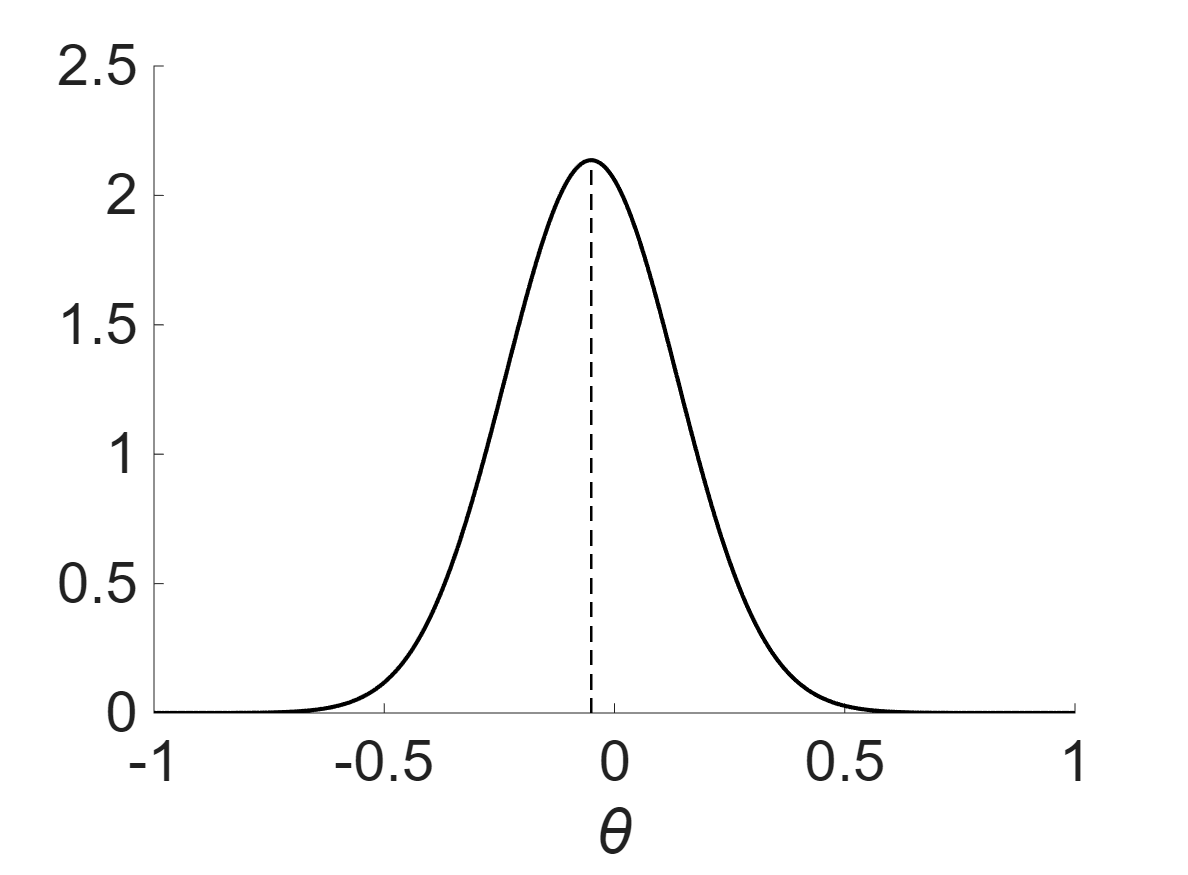} \\
\includegraphics[width=0.42\textwidth]{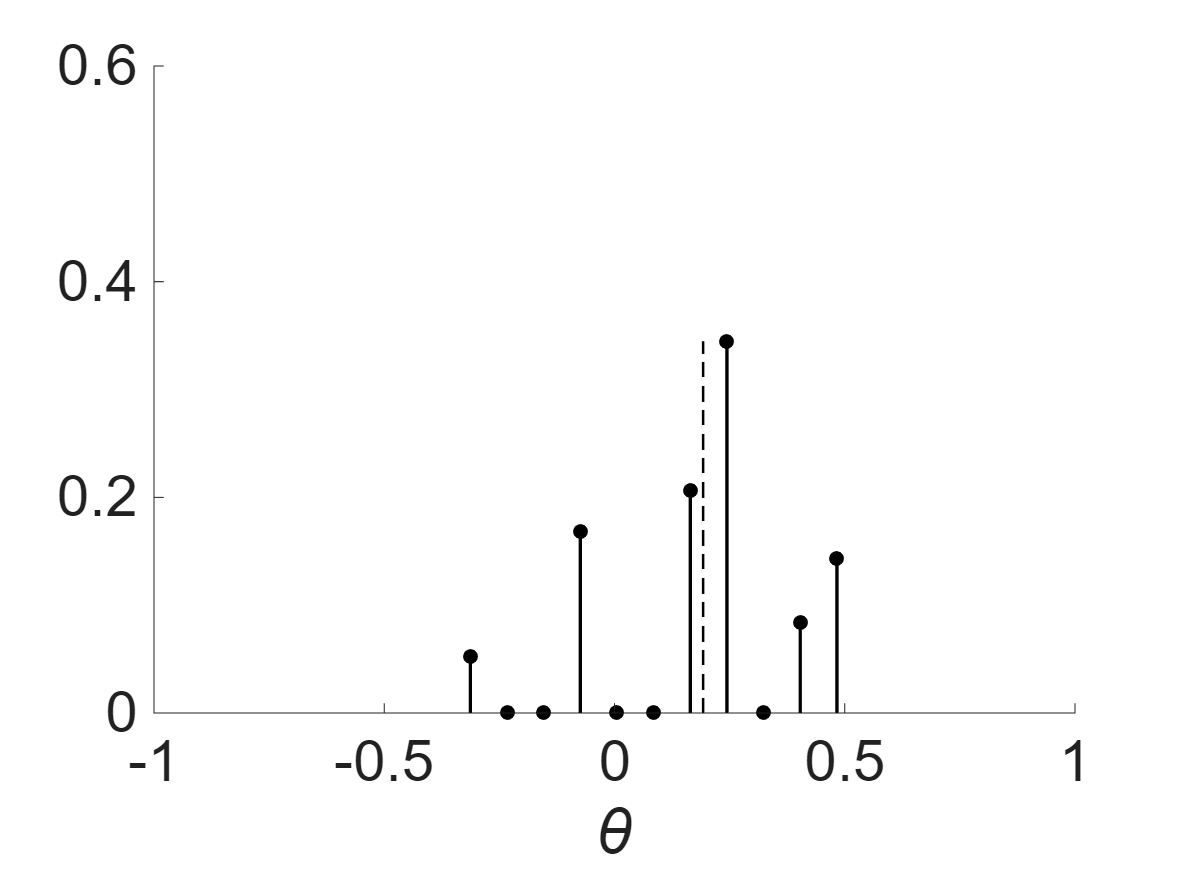} &
\includegraphics[width=0.42\textwidth]{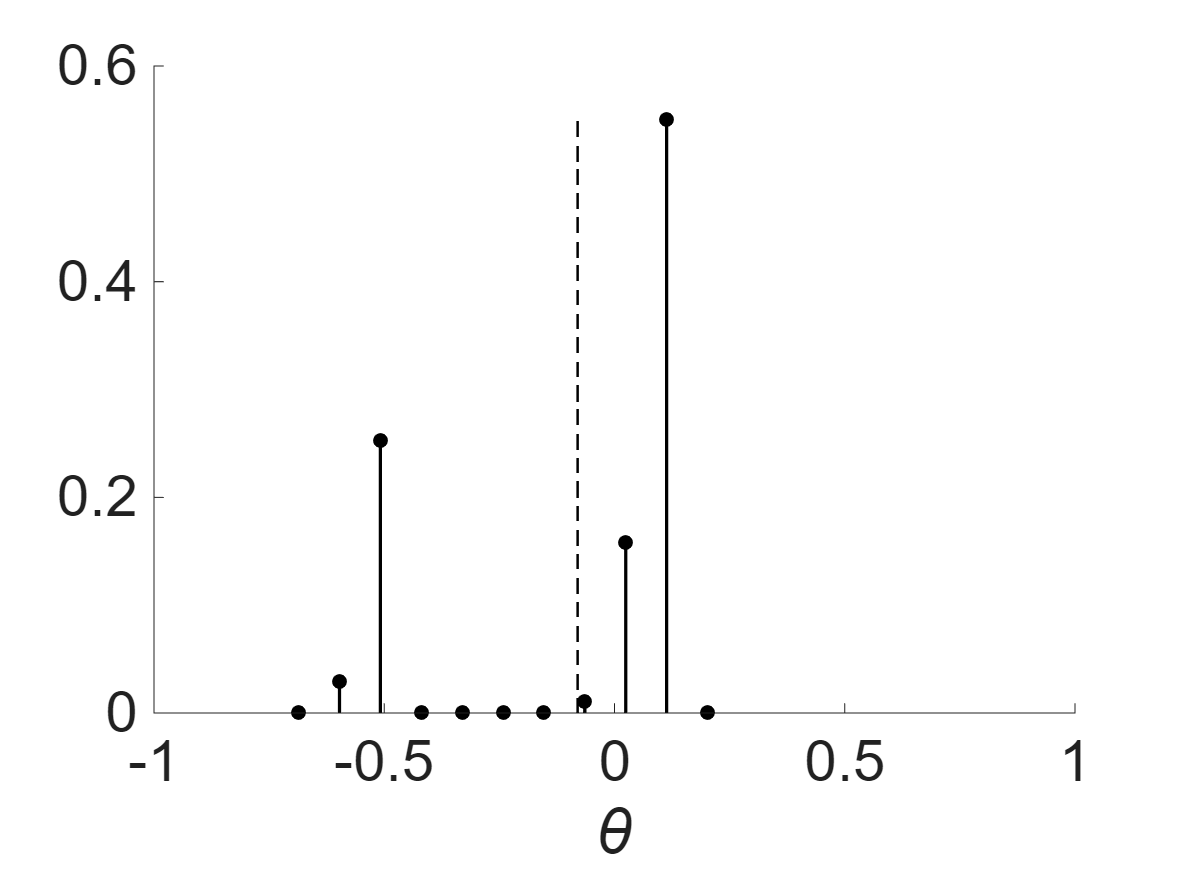}
\end{tabular}
\end{center}
\label{fig:drug_pfs_priors_joint}
{\footnotesize {\em Notes}: Top row: Gaussian-joint prior marginals for PD-L1 high-expression and low-expression effects. Bottom row: NPMLE-joint prior marginals for the same two effects.}
\end{figure}

\begin{figure}[htbp]
\caption{PFS estimated prior marginals (independent)}
\begin{center}
\begin{tabular}{cc}
\includegraphics[width=0.42\textwidth]{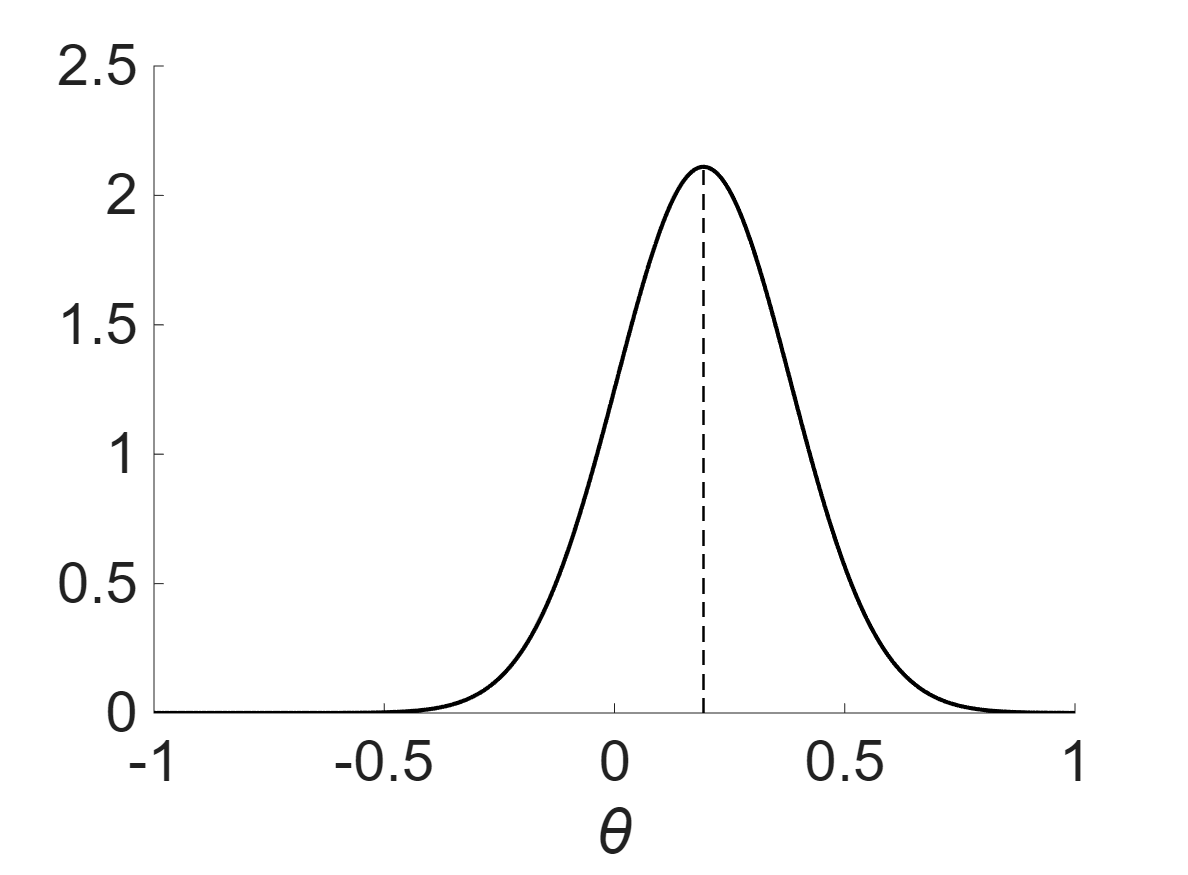} &
\includegraphics[width=0.42\textwidth]{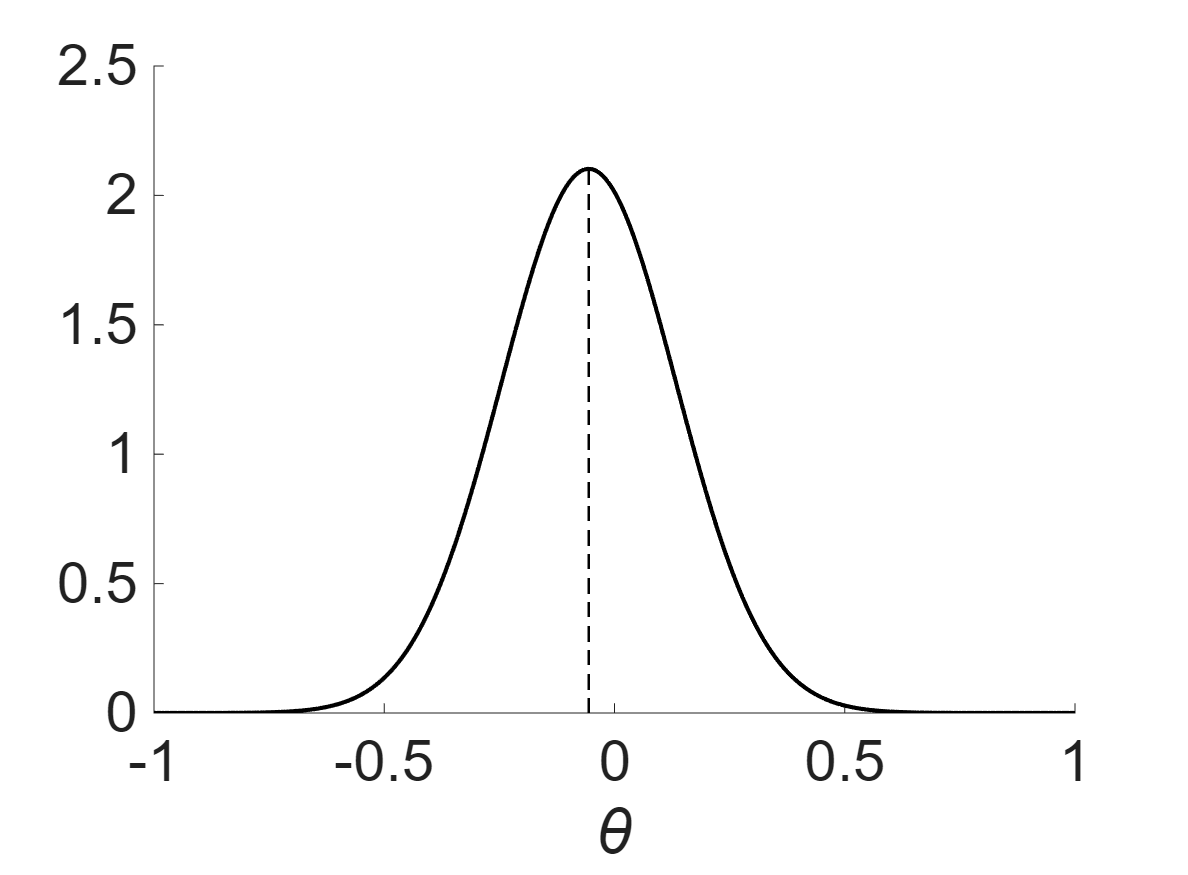} \\
\includegraphics[width=0.42\textwidth]{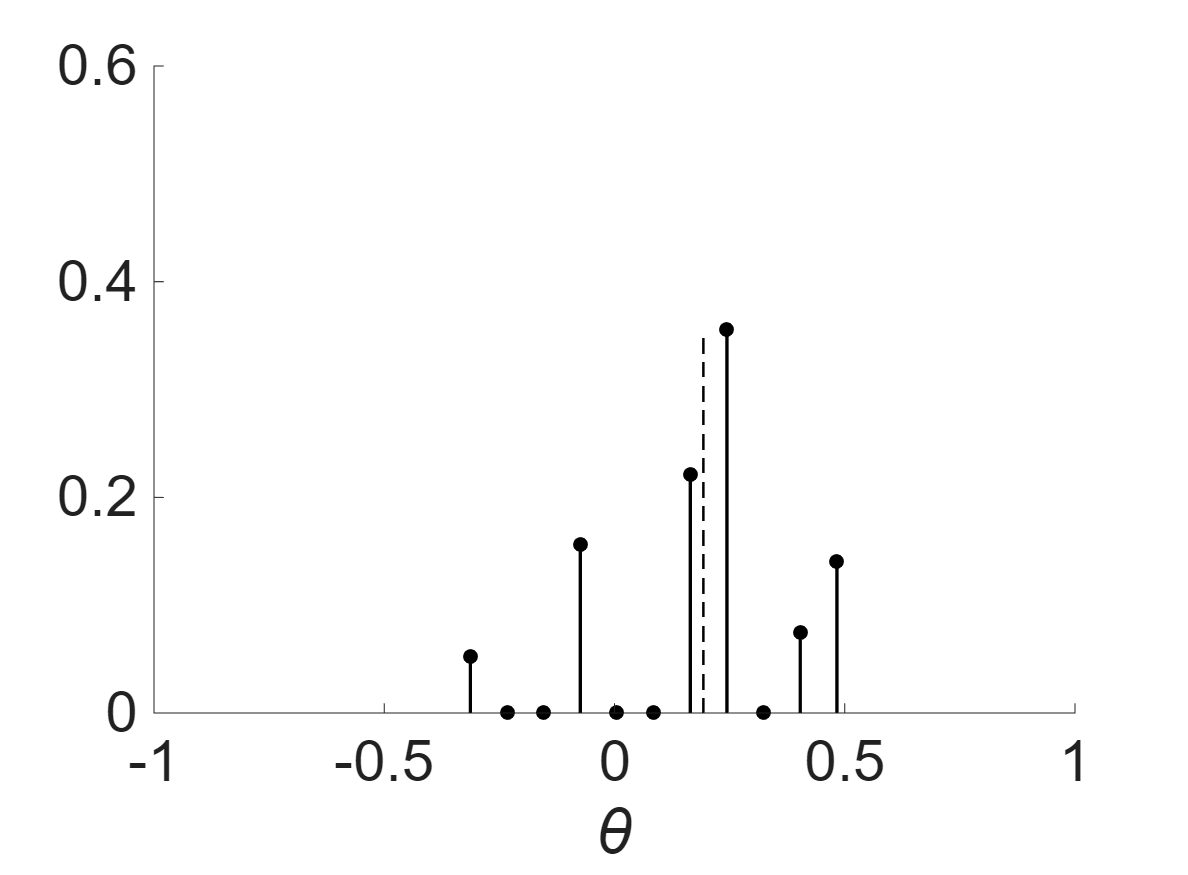} &
\includegraphics[width=0.42\textwidth]{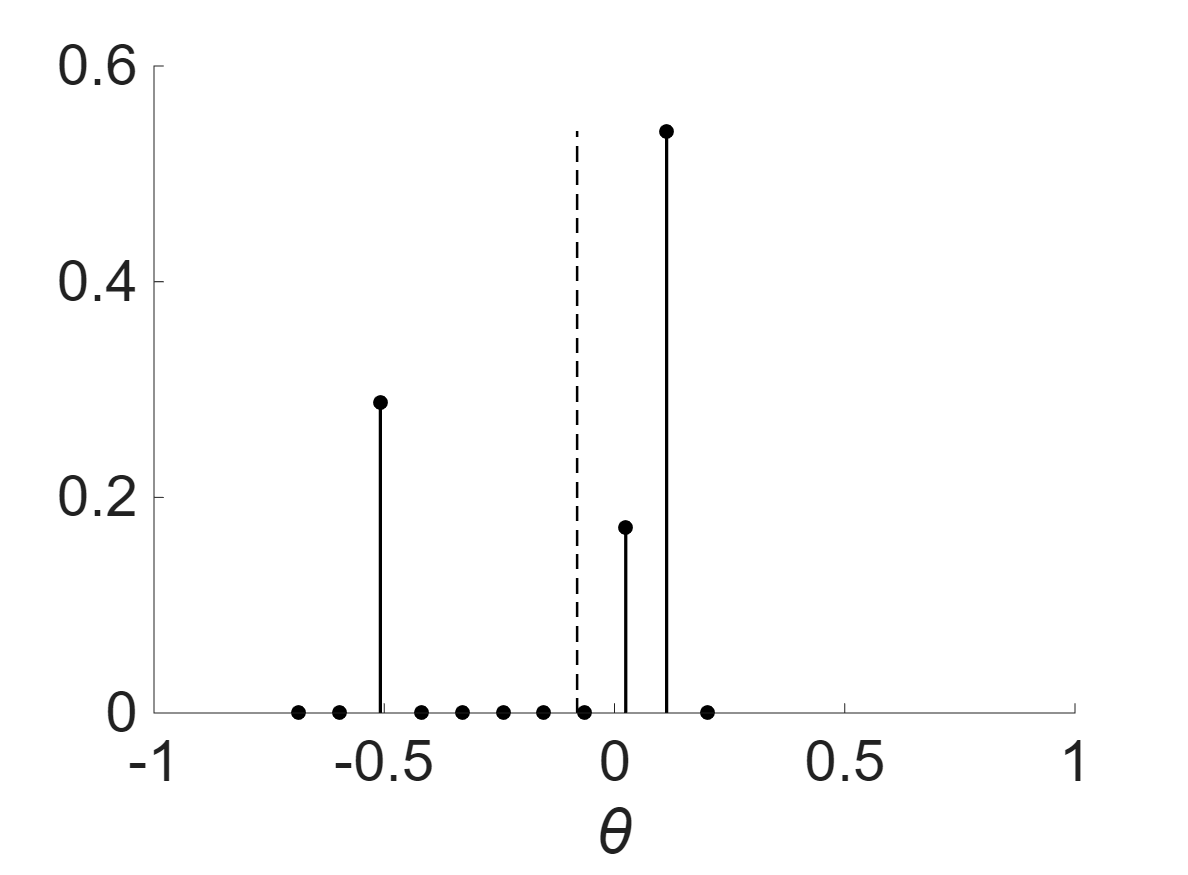}
\end{tabular}
\end{center}
\label{fig:drug_pfs_priors_diag}
{\footnotesize {\em Notes}: Top row: Gaussian-independent prior marginals for PD-L1 high-expression and low-expression effects. Bottom row: NPMLE-independent prior marginals for the same two effects.}
\end{figure}

\begin{table}[htbp]
\caption{Objective-I optimal design for PFS}
\begin{center}
\begin{tabular}{lcc}
\toprule
Prior & $e_{\text{high}}$ & $e_{\text{low}}$ \\
\midrule
Gaussian-independent EB & 0.250 & 0.250 \\
NPMLE-independent EB & 0.296 & 0.204 \\
Gaussian-joint EB & 0.251 & 0.249 \\
NPMLE-joint EB & 0.297 & 0.203 \\
No-information & 0.250 & 0.250 \\
\bottomrule
\end{tabular}
\end{center}
\label{tab:drug_pfs_design_2x2}
\end{table}

\clearpage
\subsection{Project STAR}
\label{app:star_more}

Figure~\ref{fig:star_app_prior_joint_rest} reports the joint-prior marginals for strata $S_2$--$S_4$, complementing the stratum-$1$ panels in Figure~\ref{fig:star_prior_dist}.

\begin{figure}[htbp]
\caption{STAR reading: joint-prior marginals for $S_2$--$S_4$}
\begin{center}
\begin{tabular}{cc}
\includegraphics[width=0.42\textwidth]{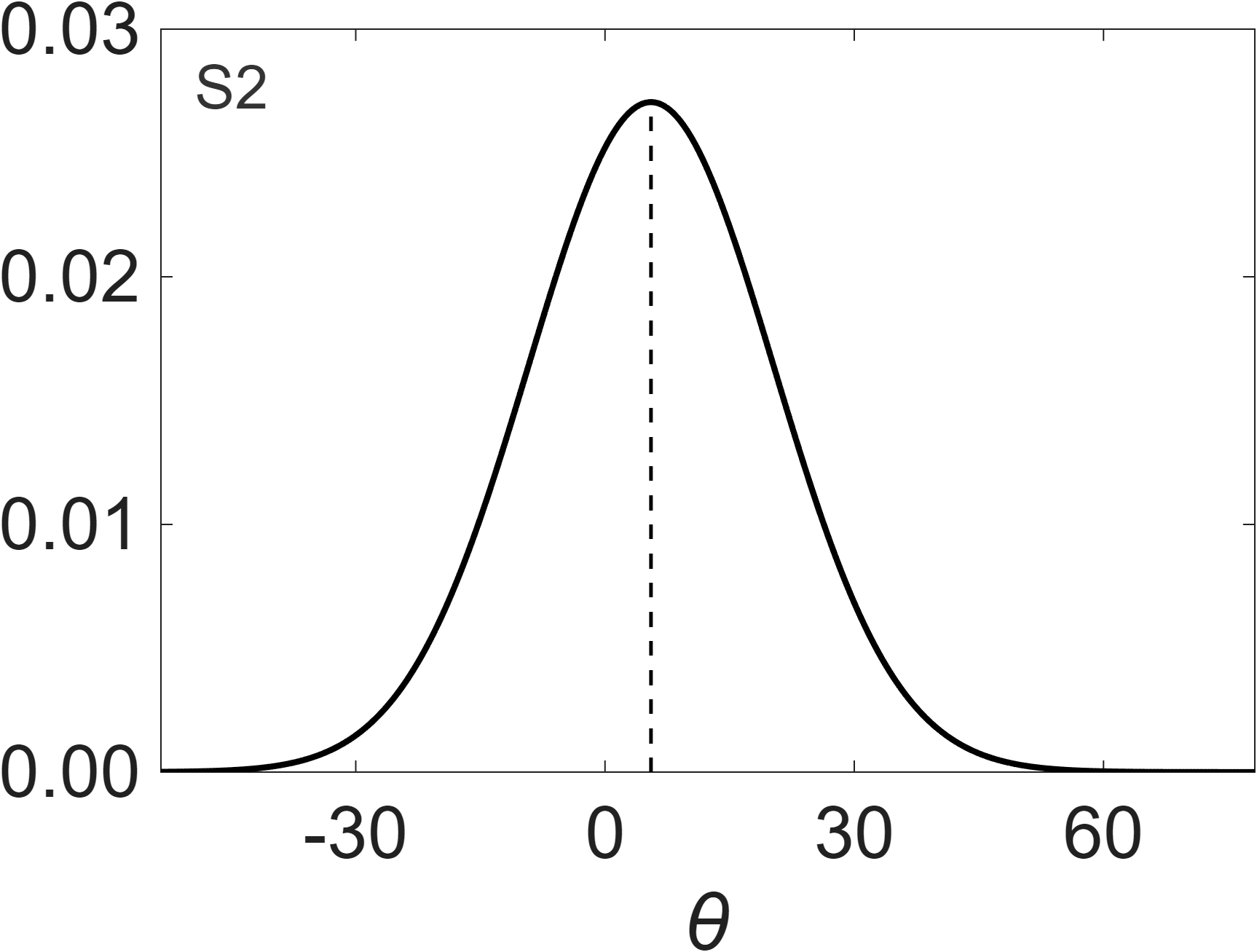} &
\includegraphics[width=0.42\textwidth]{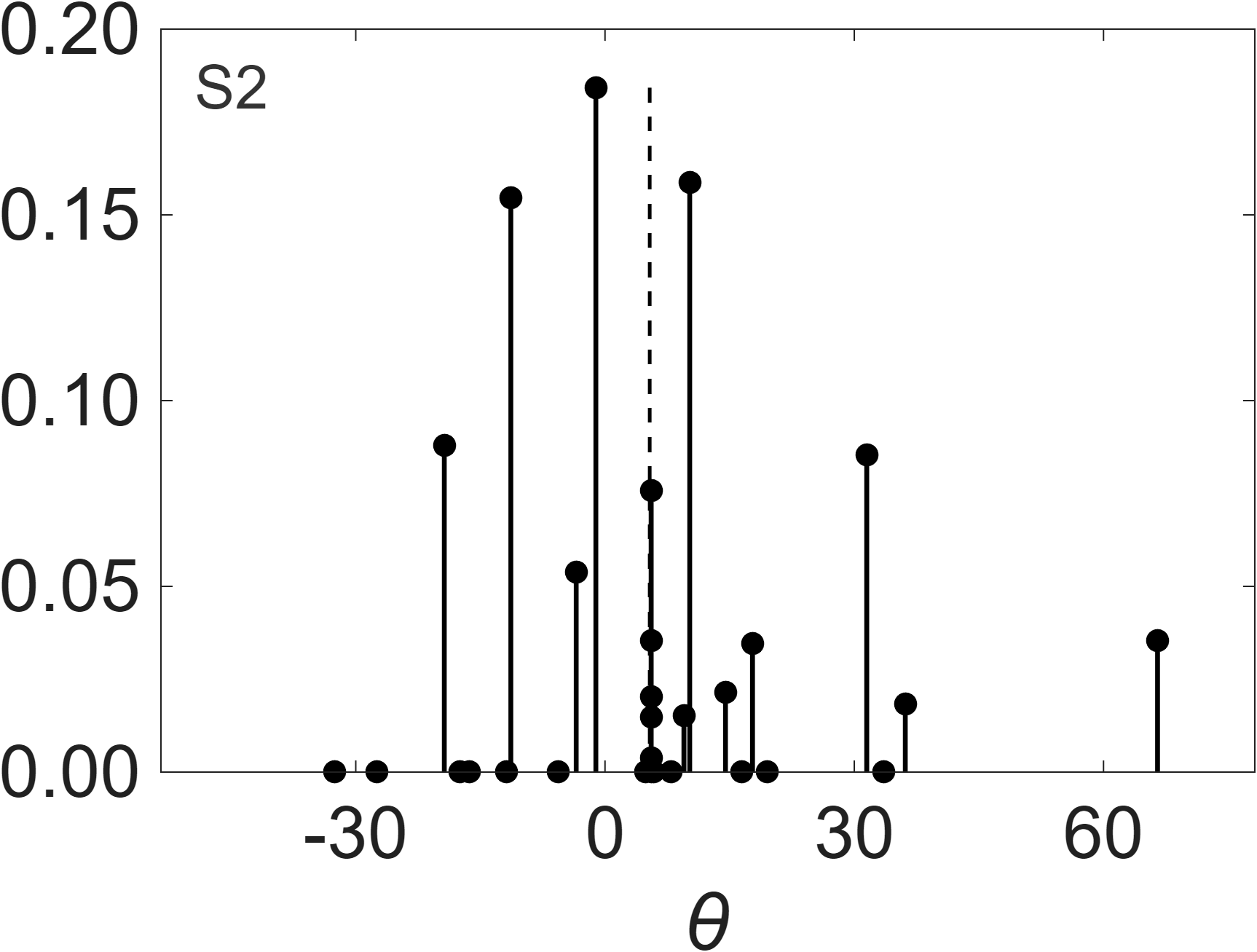} \\
\includegraphics[width=0.42\textwidth]{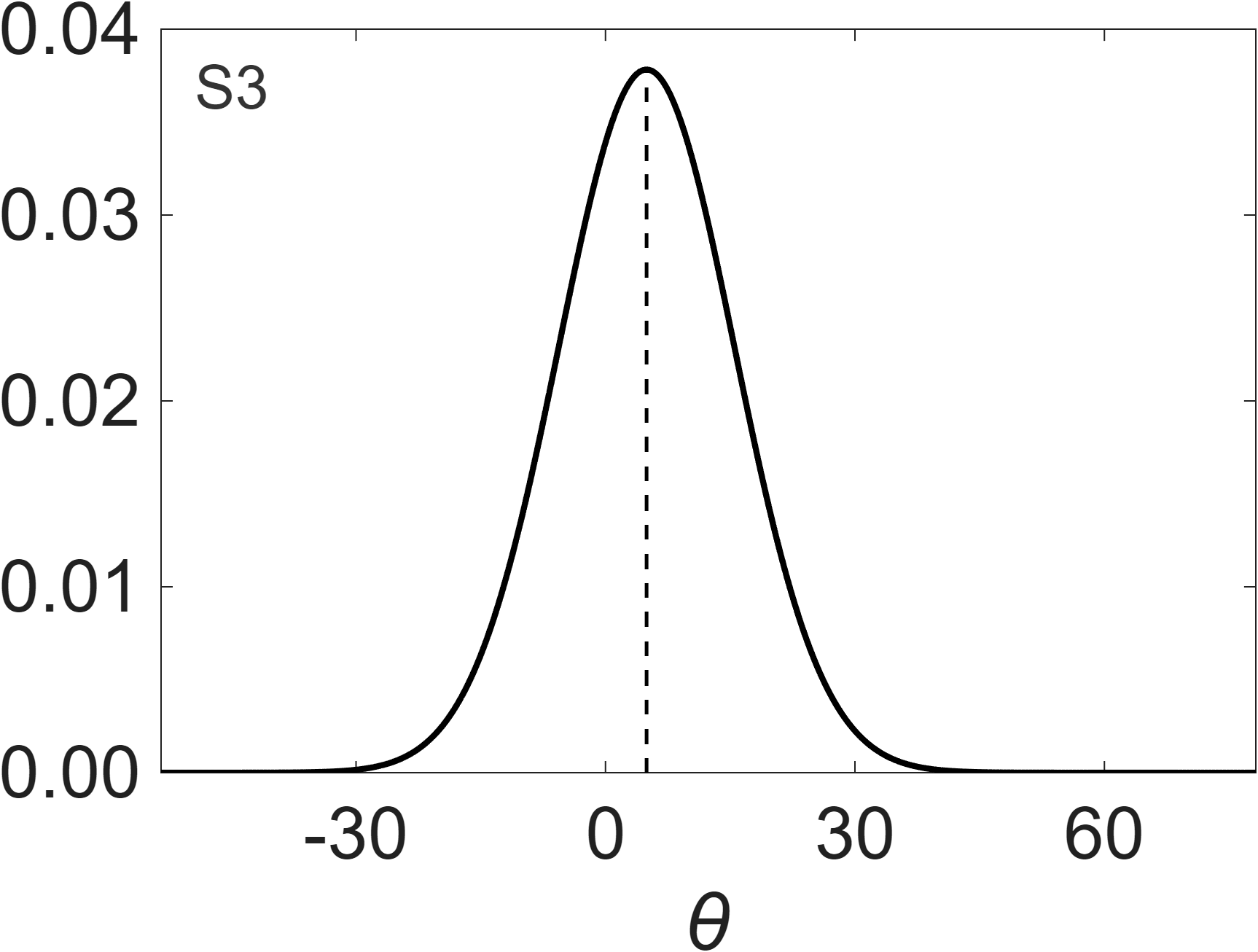} &
\includegraphics[width=0.42\textwidth]{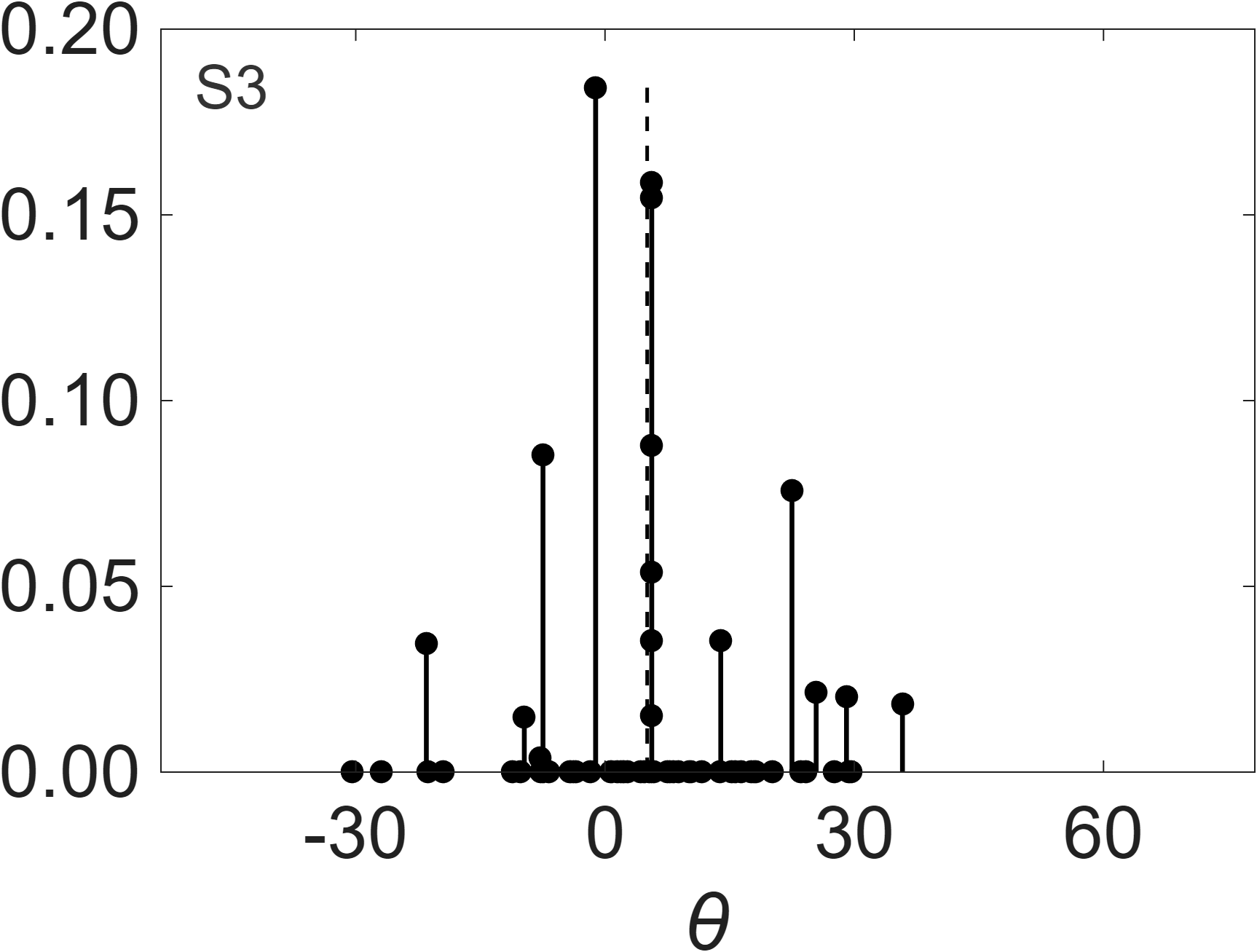} \\
\includegraphics[width=0.42\textwidth]{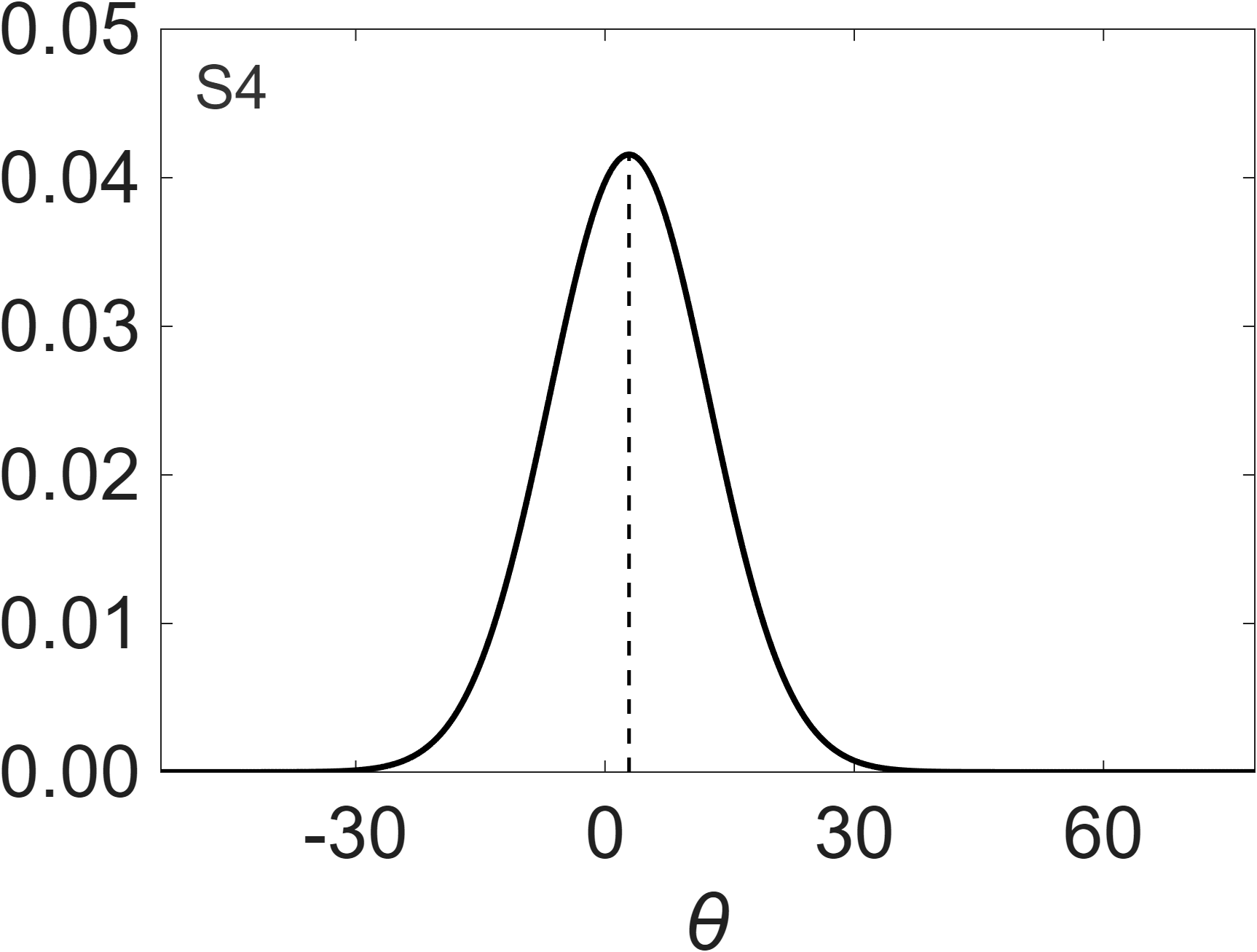} &
\includegraphics[width=0.42\textwidth]{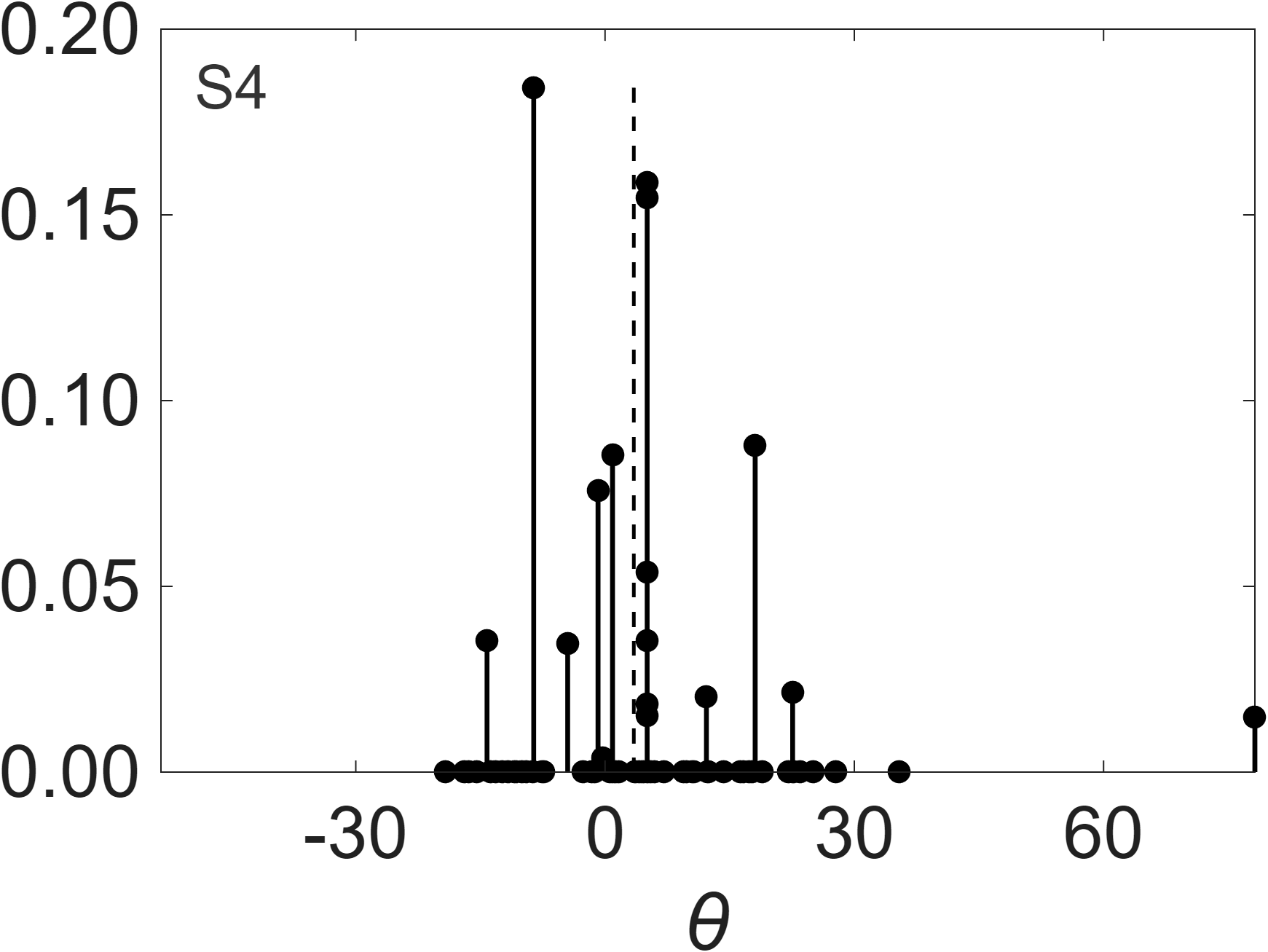}
\end{tabular}
\end{center}
\label{fig:star_app_prior_joint_rest}
{\footnotesize {\em Notes}: In each row, the left panel is Gaussian-joint and the right panel is NPMLE-joint for the same stratum. Rows correspond to $S_2$, $S_3$, and $S_4$.}
\end{figure}

Figure~\ref{fig:star_app_prior_diag} reports independent-subgroup prior counterparts of Figure~\ref{fig:star_prior_dist}.

\begin{figure}[htbp]
\caption{STAR reading: estimated prior distributions (independent-subgroup priors)}
\begin{center}
\begin{tabular}{cc}
\includegraphics[width=0.42\textwidth]{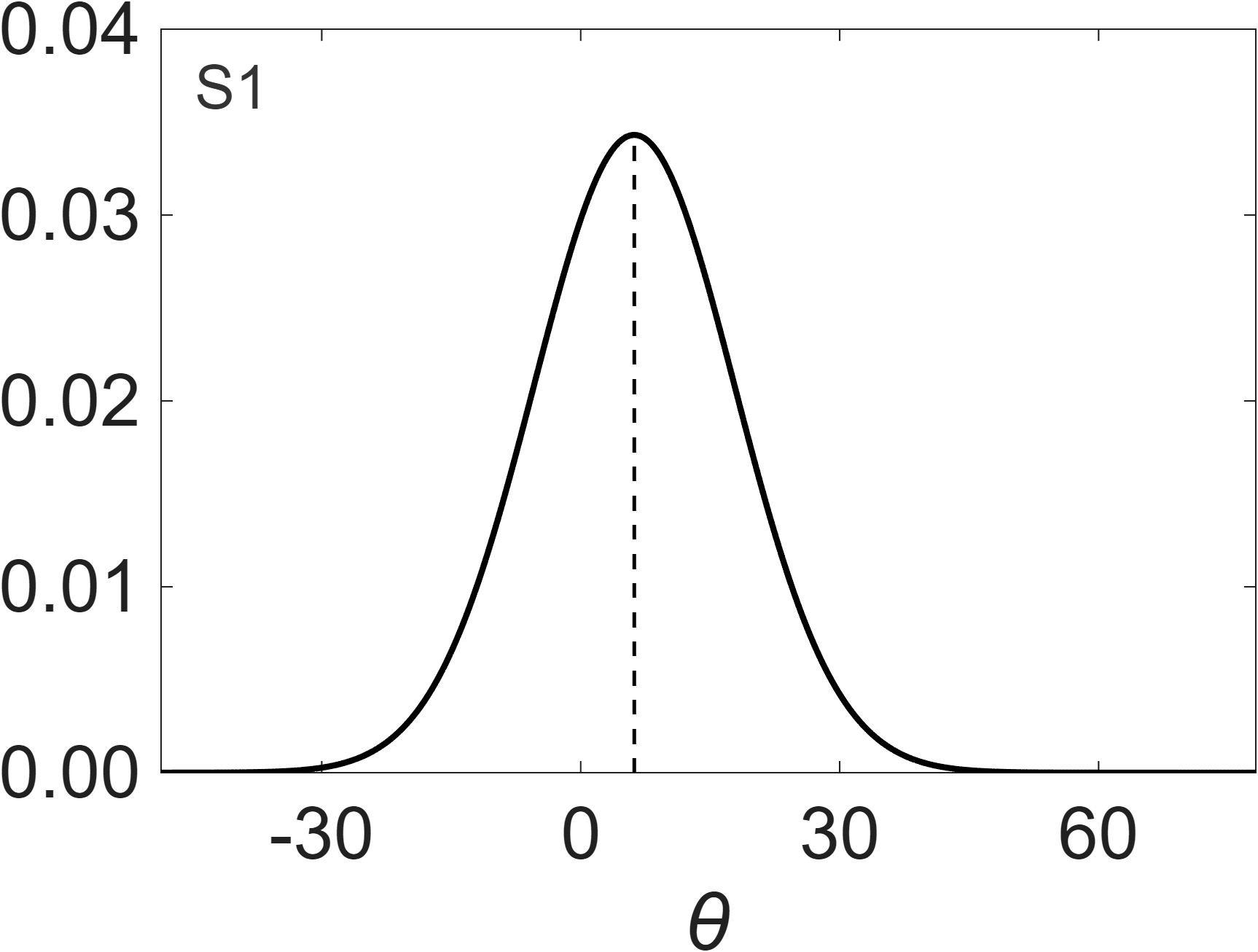} &
\includegraphics[width=0.42\textwidth]{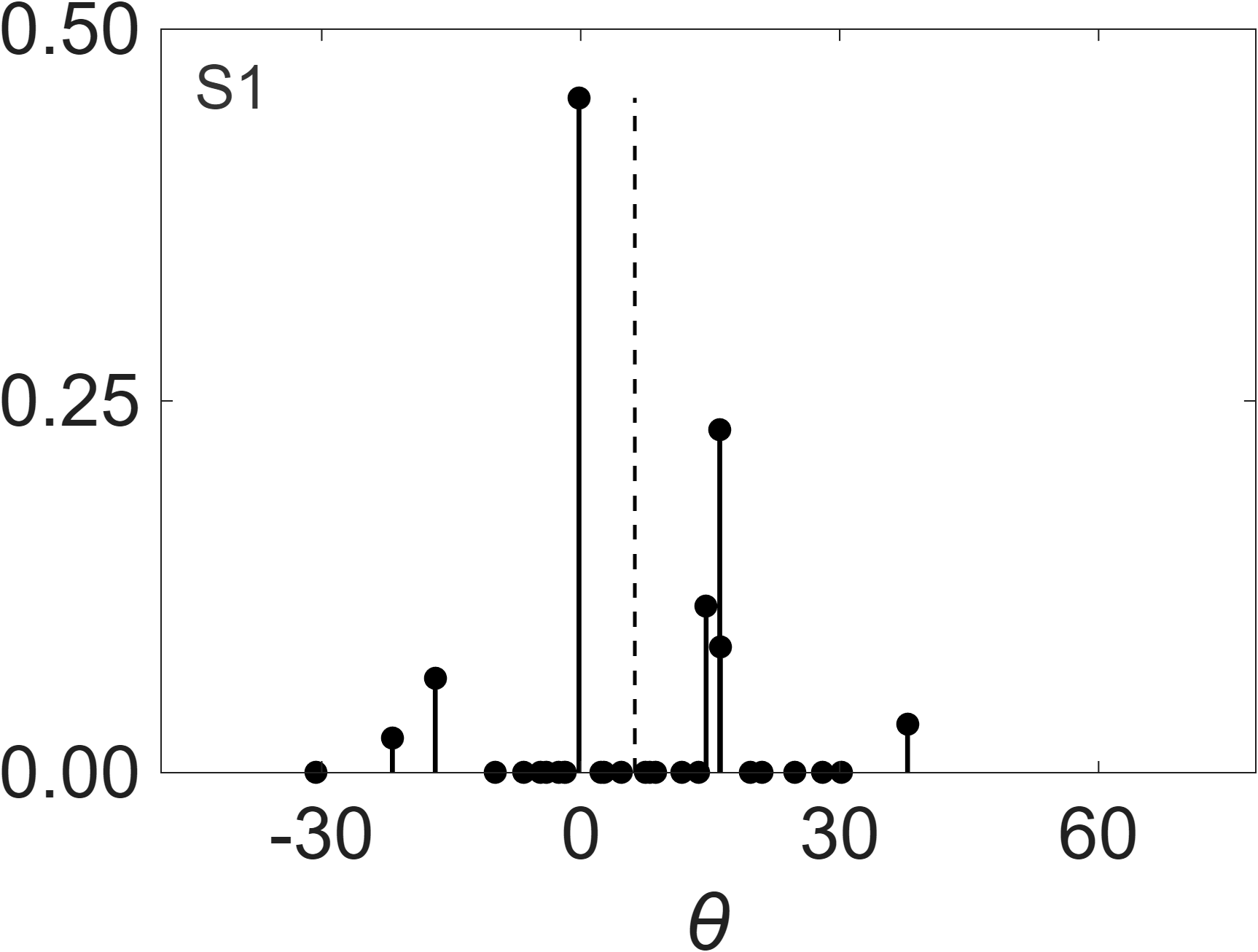} \\
\includegraphics[width=0.42\textwidth]{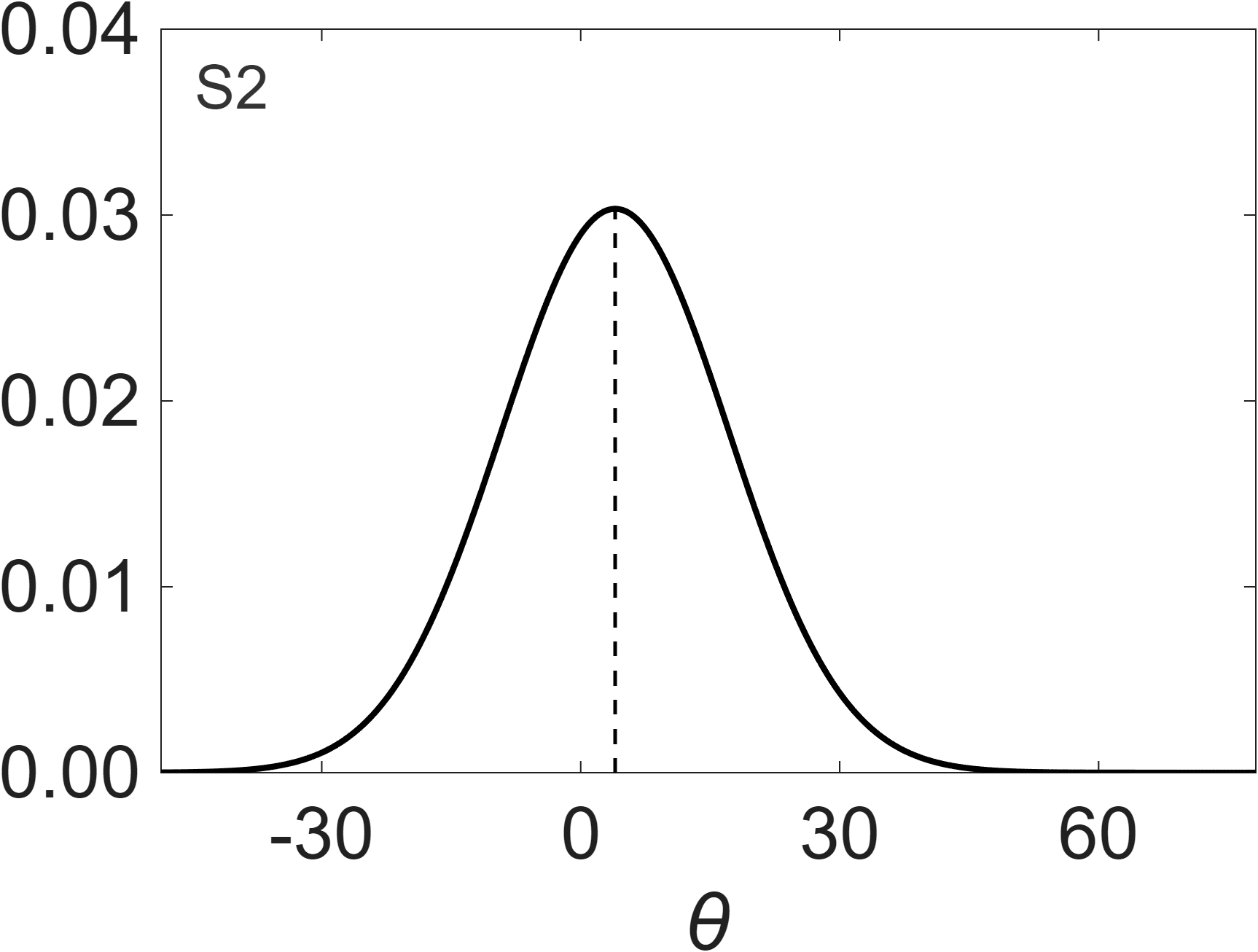} &
\includegraphics[width=0.42\textwidth]{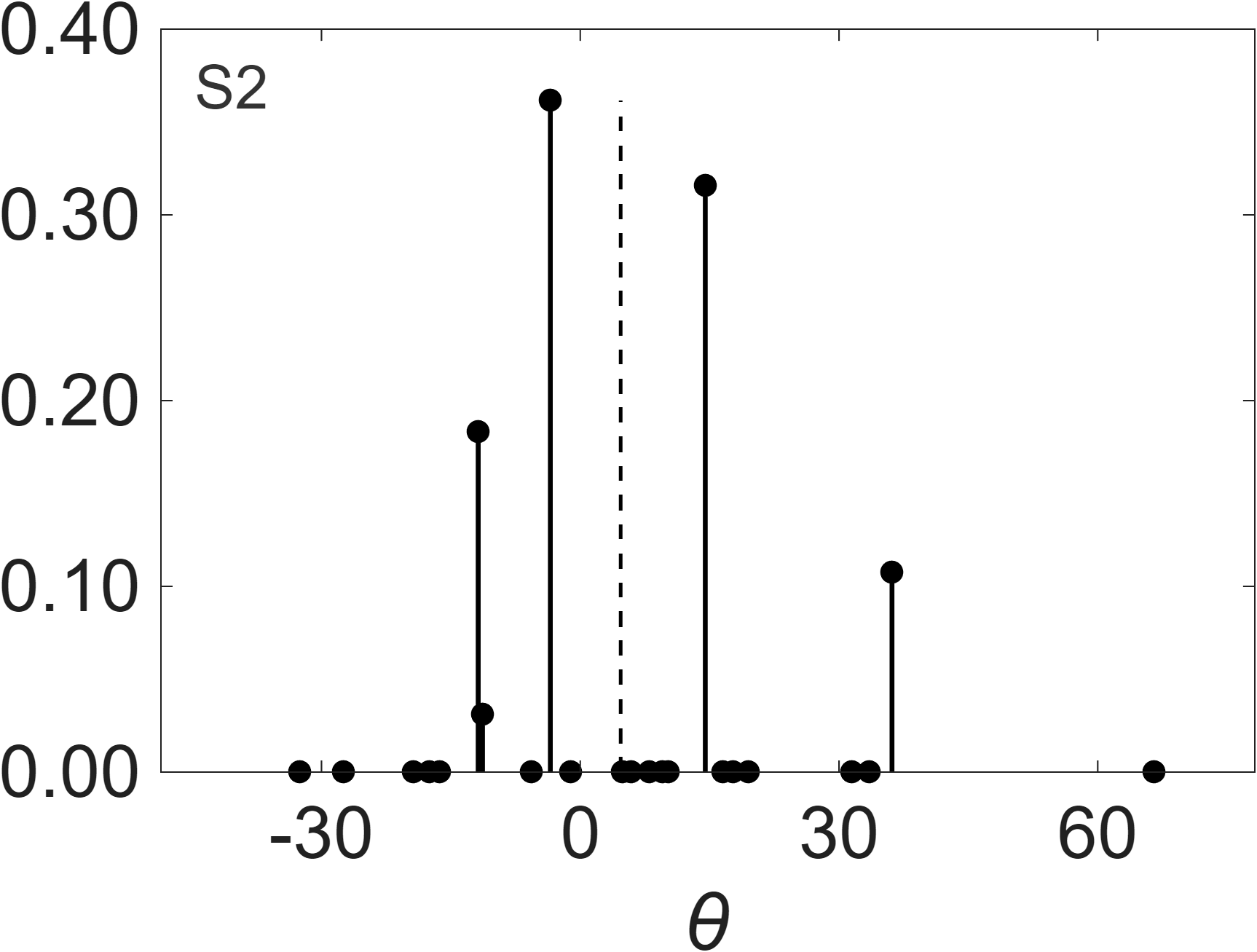} \\
\includegraphics[width=0.42\textwidth]{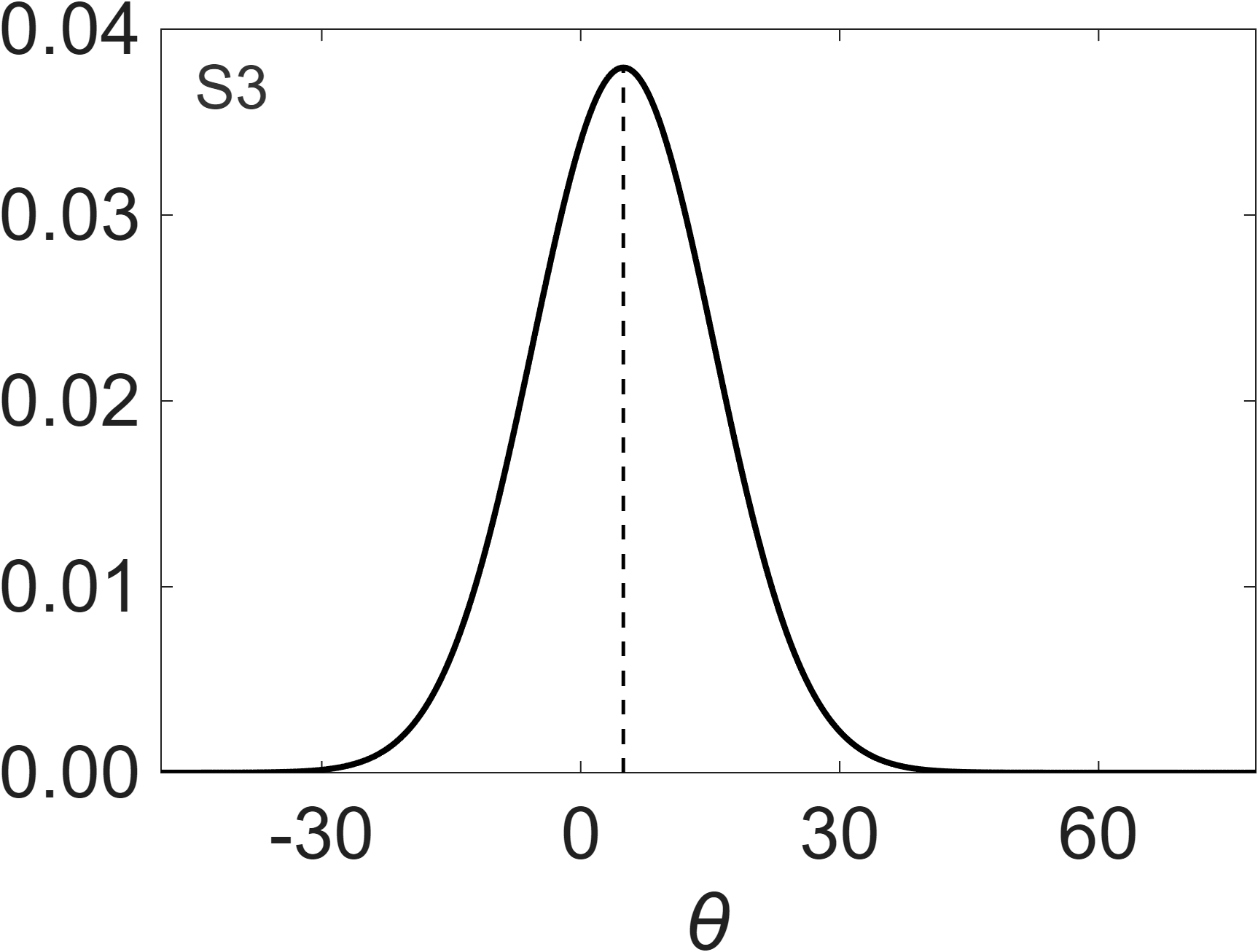} &
\includegraphics[width=0.42\textwidth]{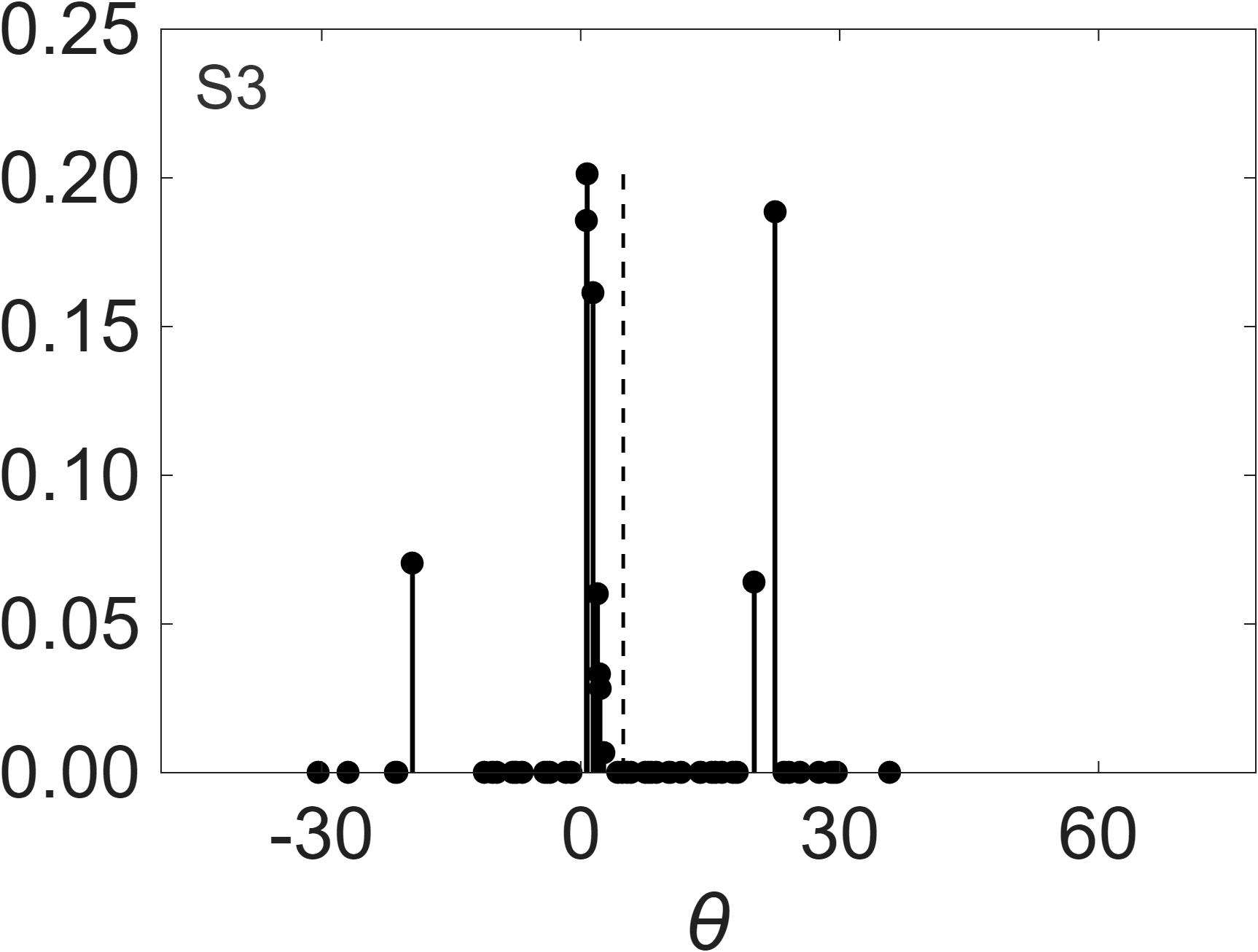} \\
\includegraphics[width=0.42\textwidth]{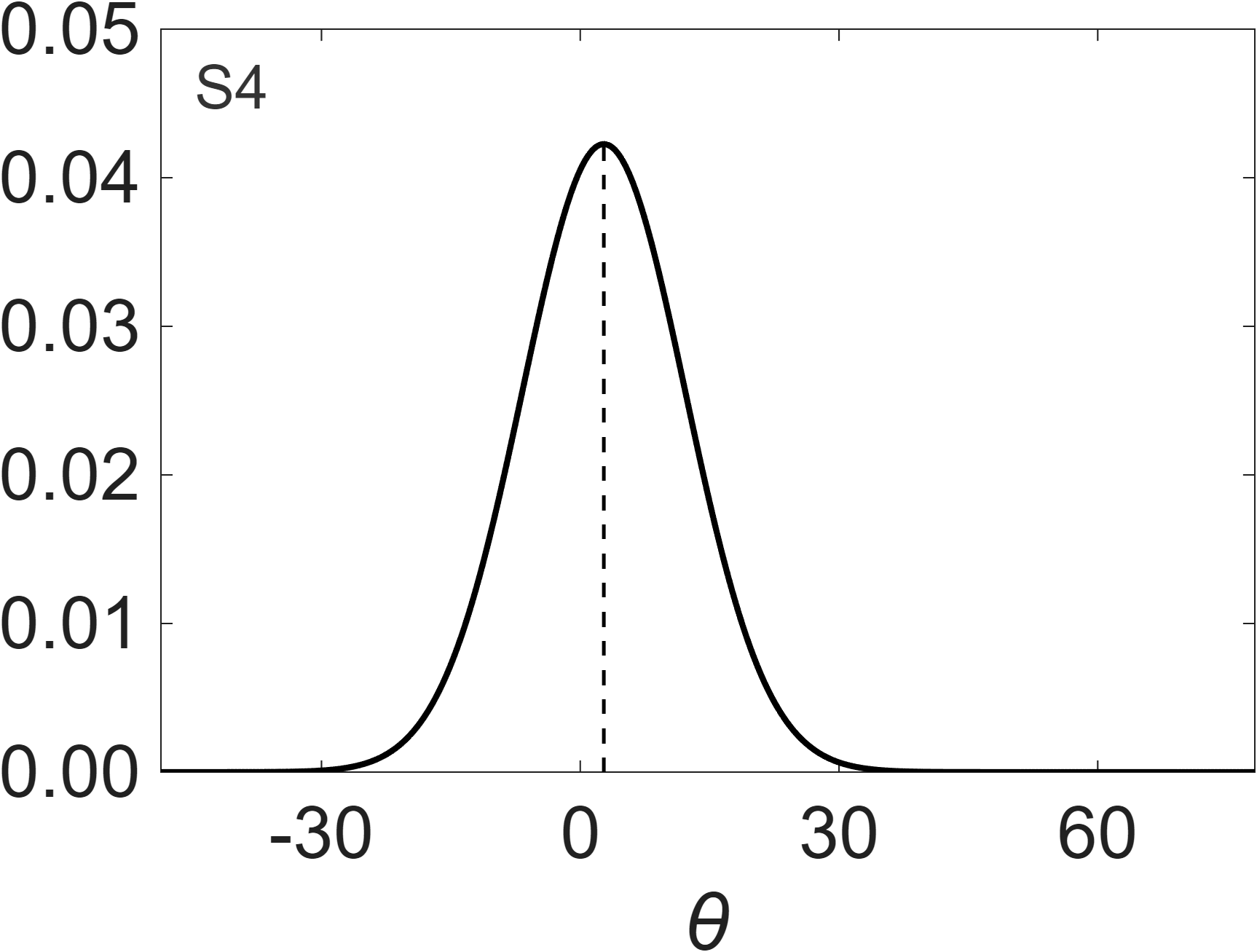} &
\includegraphics[width=0.42\textwidth]{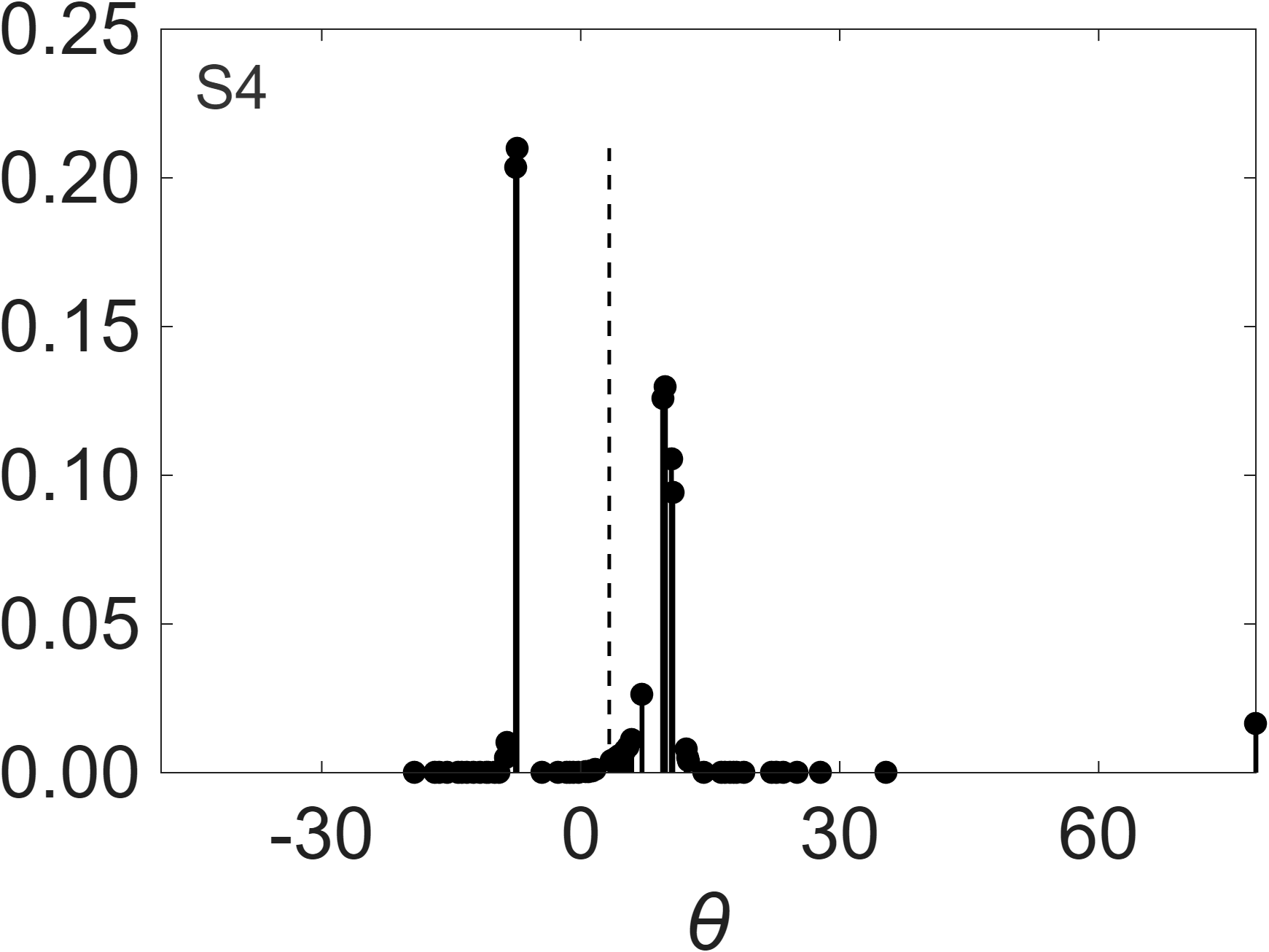}
\end{tabular}
\end{center}
\label{fig:star_app_prior_diag}
{\footnotesize {\em Notes}: In each row, the left panel is Gaussian independent-subgroup EB and the right panel is NPMLE independent-subgroup EB for the same stratum. Rows correspond to $S_1$--$S_4$.}
\end{figure}

Figure~\ref{fig:star_app_diag} reports objective-specific designs for the independent-subgroup versions of Gaussian EB and NPMLE EB. 

\begin{figure}[htbp]
\caption{STAR reading: objective-specific designs under independent-subgroup EB priors}
\begin{center}
\begin{tabular}{cc}
\includegraphics[width=0.47\textwidth]{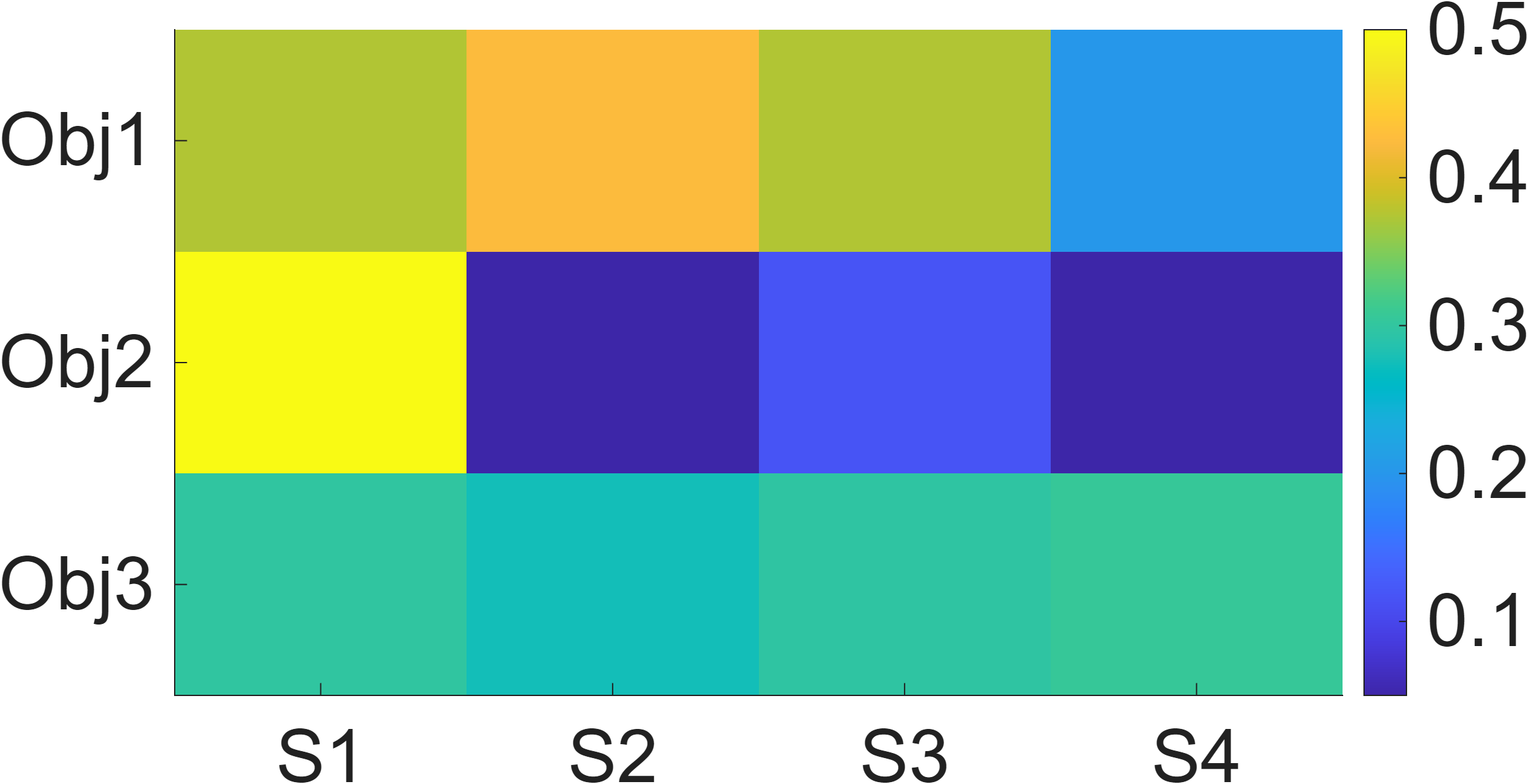} &
\includegraphics[width=0.47\textwidth]{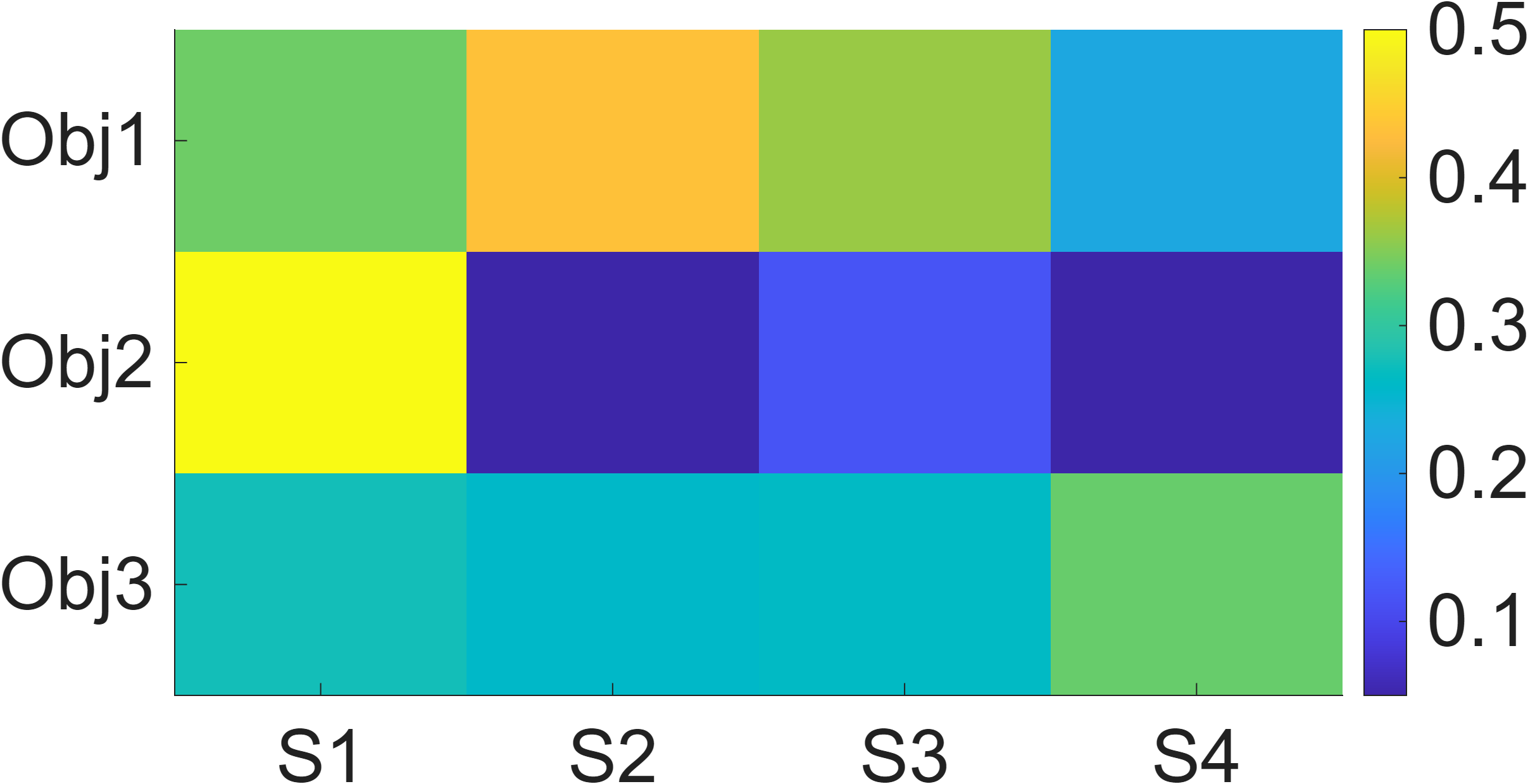}
\end{tabular}
\end{center}
\label{fig:star_app_diag}
{\footnotesize {\em Notes}: Left panel is Gaussian EB under the independent-subgroup specification; right panel is NPMLE EB under the independent-subgroup specification. Rows are Objective I--III and columns are $S_1$--$S_4$.}
\end{figure}

Figures~\ref{fig:star_app_obj1}--\ref{fig:star_app_obj3} show, for each objective, the full prior-robustness comparison across five priors: no information, Gaussian EB (independent/joint), and NPMLE EB (independent/joint).

\begin{figure}[htbp]
\caption{STAR reading: Objective I across all priors}
\begin{center}
\includegraphics[width=0.78\textwidth]{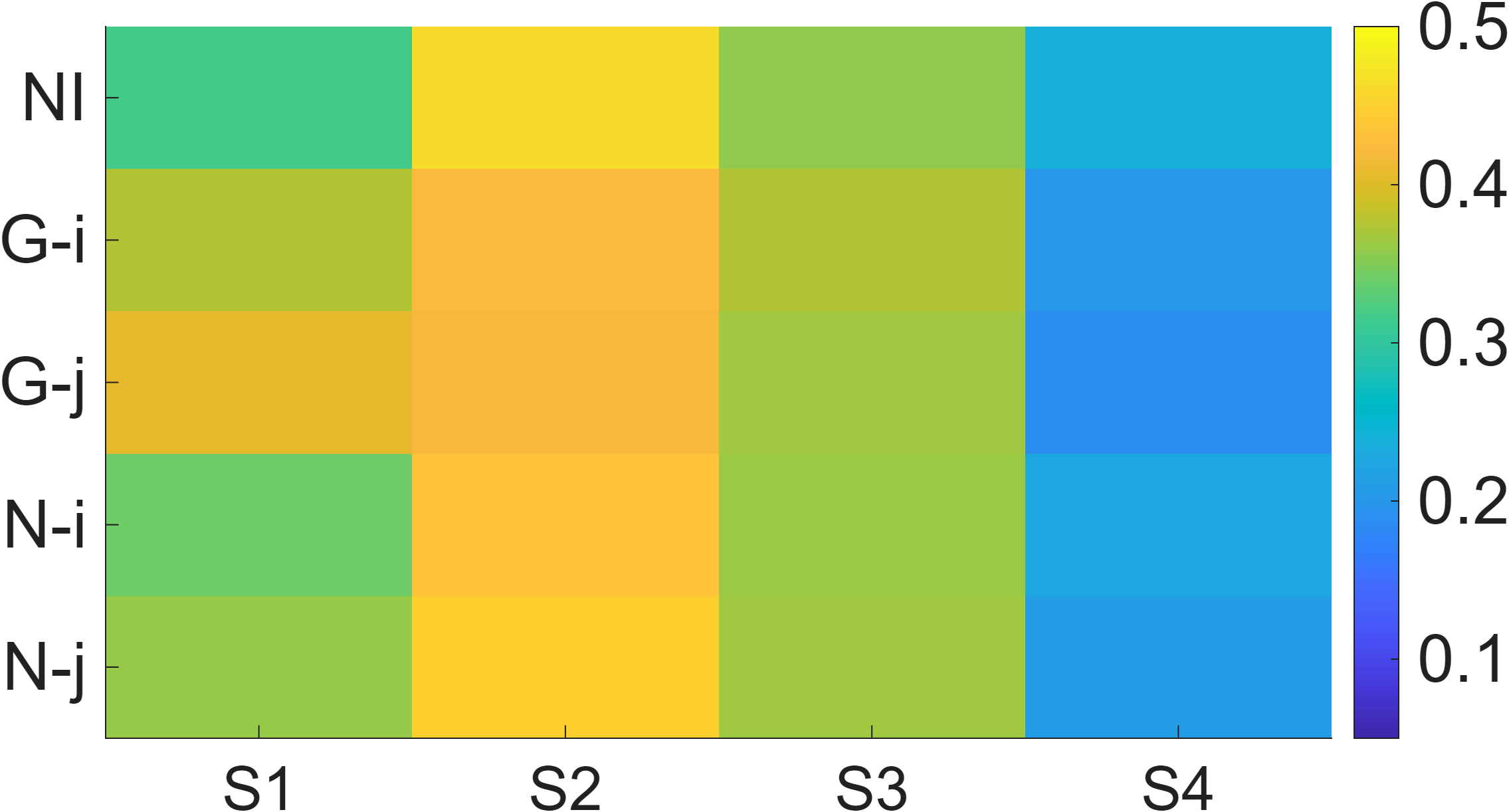}
\end{center}
\label{fig:star_app_obj1}
{\footnotesize {\em Notes}: Rows are priors (NI, G-i, G-j, N-i, N-j); columns are strata $S_1$--$S_4$. Here NI is the no-information benchmark, G-i and G-j are Gaussian EB with independent-subgroup and joint priors, and N-i and N-j are NPMLE EB with independent-subgroup and joint priors.}
\end{figure}

\begin{figure}[htbp]
\caption{STAR reading: Objective II across all priors}
\begin{center}
\includegraphics[width=0.78\textwidth]{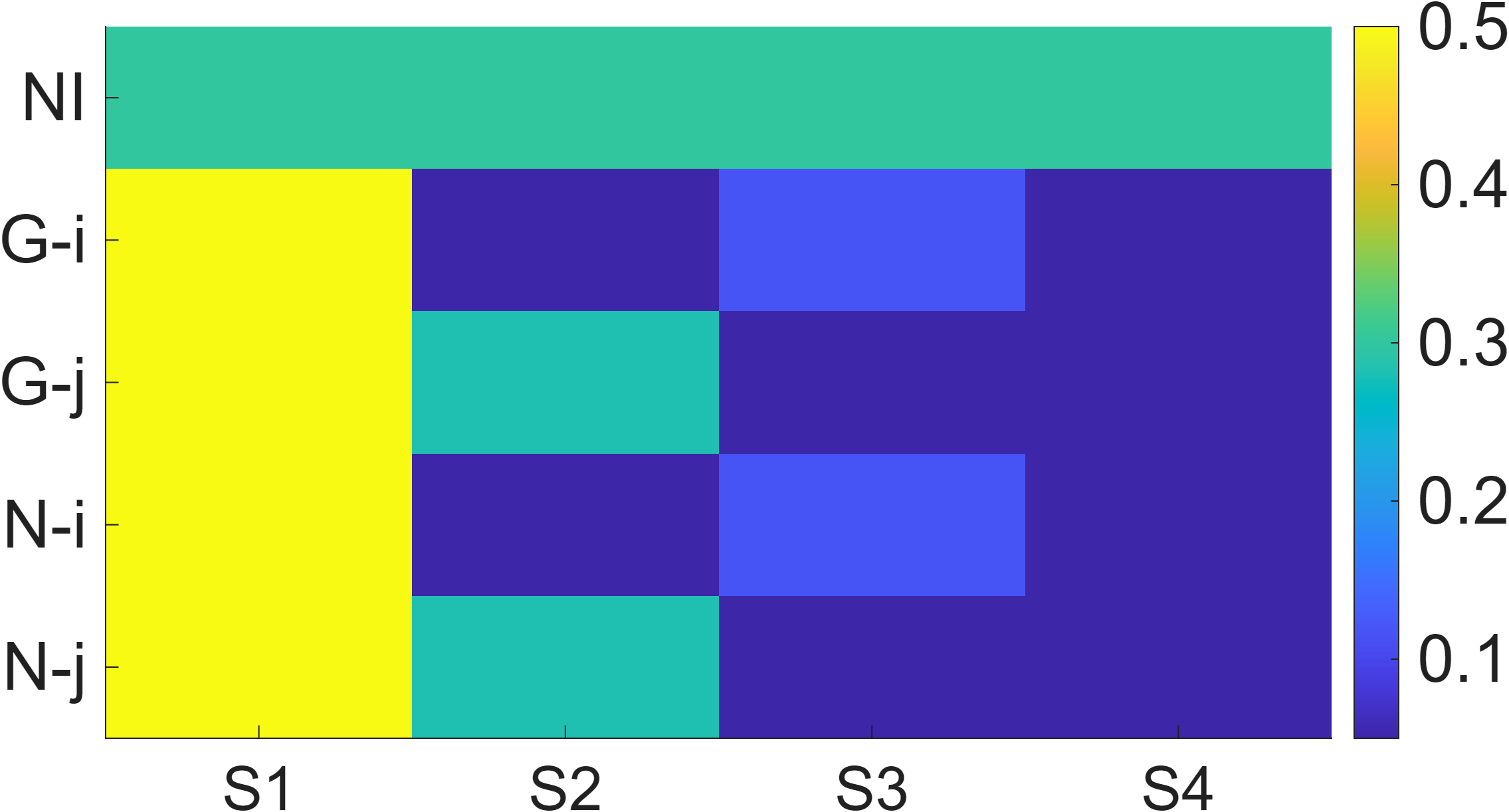}
\end{center}
\label{fig:star_app_obj2}
{\footnotesize {\em Notes}: Rows are priors (NI, G-i, G-j, N-i, N-j); columns are strata $S_1$--$S_4$. Here NI is the no-information benchmark, G-i and G-j are Gaussian EB with independent-subgroup and joint priors, and N-i and N-j are NPMLE EB with independent-subgroup and joint priors.}
\end{figure}

\begin{figure}[htbp]
\caption{STAR reading: Objective III across all priors}
\begin{center}
\includegraphics[width=0.78\textwidth]{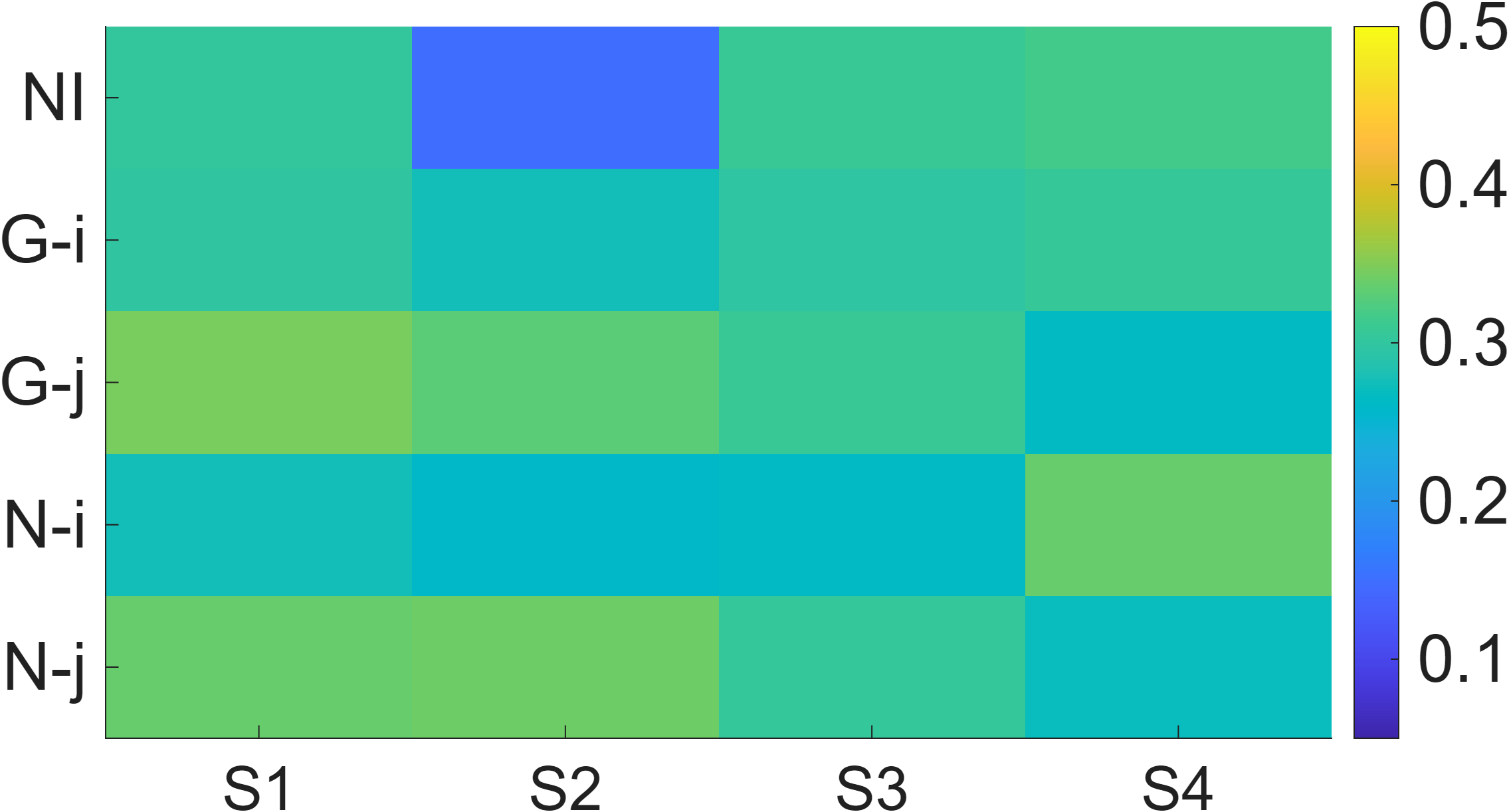}
\end{center}
\label{fig:star_app_obj3}
{\footnotesize {\em Notes}: Rows are priors (NI, G-i, G-j, N-i, N-j); columns are strata $S_1$--$S_4$. Here NI is the no-information benchmark, G-i and G-j are Gaussian EB with independent-subgroup and joint priors, and N-i and N-j are NPMLE EB with independent-subgroup and joint priors.}
\end{figure}

\end{appendix}

\end{document}